\documentclass[1p]{elsarticle}

\usepackage{lineno,hyperref}
\usepackage{amsfonts}
\usepackage{amsmath}
\usepackage{amssymb}
\usepackage{amsthm}
\usepackage{mathrsfs}
\usepackage{subcaption}
\usepackage{float}

\modulolinenumbers[5]
\hypersetup{
  bookmarksnumbered = true,
  bookmarksopen=false,
  pdfborder=0 0 0,         % make all links invisible, so the pdf looks good when printed
  pdffitwindow=true,      % window fit to page when opened
  pdfnewwindow=true, % links in new window
  colorlinks=true,           % false: boxed links; true: colored links
  linkcolor=blue,            % color of internal links
  citecolor=magenta,    % color of links to bibliography
  filecolor=magenta,     % color of file links
  urlcolor=cyan              % color of external links
}

\journal{Journal of Computational Physics}

\newtheorem{define}{Definition}
\newtheorem{lemma}{Lemma}

\newtheorem{rem}{Remark}
\newtheorem{theorem}{Theorem}

\usepackage{./definitions}

\begin{document}

\begin{frontmatter}

\title{Realizability-Preserving DG-IMEX Method for the Two-Moment Model of Fermion Transport \tnoteref{support}\tnoteref{copyright}}
\tnotetext[support]{
This research is sponsored, in part, by the Laboratory Directed Research and Development Program of Oak Ridge National Laboratory (ORNL), managed by UT-Battelle, LLC for the U. S. Department of Energy under Contract No. De-AC05-00OR22725.  
This research was supported by the Exascale Computing Project (17-SC-20-SC), a collaborative effort of the U.S. Department of Energy Office of Science and the National Nuclear Security Administration.  
This material is based, in part, upon work supported by the U.S. Department of Energy, Office of Science, Office of Advanced Scientific Computing Research.  
Eirik Endeve was supported in part by NSF under Grant No. 1535130.}
\tnotetext[copyright]{
This manuscript has been authored by UT-Battelle, LLC under Contract No. DE-AC05-00OR22725 with the U.S. Department of Energy. The United States Government retains and the publisher, by accepting the article for publication, acknowledges that the United States Government retains a non-exclusive, paid-up, irrevocable, world-wide license to publish or reproduce the published form of this manuscript, or allow others to do so, for United States Government purposes. The Department of Energy will provide public access to these results of federally sponsored research in accordance with the DOE Public Access Plan(http://energy.gov/downloads/doe-public-access-plan).}

\author[utk-phys]{Ran Chu}
\ead{rchu@vols.utk.edu}

\author[ornl,utk-phys,jics]{Eirik Endeve\corref{cor}}
\ead{endevee@ornl.gov}

\author[ornl,utk-math]{Cory D. Hauck}
\ead{hauckc@ornl.gov}

\author[utk-phys,jics]{Anthony Mezzacappa}
\ead{mezz@utk.edu}

\cortext[cor]{Corresponding author. Tel.:+1 865 576 6349; fax:+1 865 241 0381}

\address[ornl]{Computational and Applied Mathematics Group, Oak Ridge National Laboratory, Oak Ridge, TN 37831 USA }

\address[utk-phys]{Department of Physics and Astronomy, University of Tennessee Knoxville, TN 37996-1200}

\address[jics]{Joint Institute for Computational Sciences, Oak Ridge National Laboratory, Oak Ridge, TN 37831-6354}

\address[utk-math]{Department of Mathematics, University of Tennessee Knoxville, TN 37996-1320}

\begin{abstract}
Building on the framework of Zhang \& Shu \cite{zhangShu_2010a,zhangShu_2010b}, we develop a realizability-preserving method to simulate the transport of particles (fermions) through a background material using a two-moment model that evolves the angular moments of a phase space distribution function $f$.  
The two-moment model is closed using algebraic moment closures; e.g., as proposed by Cernohorsky \& Bludman \cite{cernohorskyBludman_1994} and Banach \& Larecki \cite{banachLarecki_2017a}.  
Variations of this model have recently been used to simulate neutrino transport in nuclear astrophysics applications, including core-collapse supernovae and compact binary mergers.  
We employ the discontinuous Galerkin (DG) method for spatial discretization (in part to capture the asymptotic diffusion limit of the model) combined with implicit-explicit (IMEX) time integration to stably bypass short timescales induced by frequent interactions between particles and the background.  
Appropriate care is taken to ensure the method preserves strict algebraic bounds on the evolved moments (particle density and flux) as dictated by Pauli's exclusion principle, which demands a bounded distribution function (i.e., $f\in[0,1]$).  
This realizability-preserving scheme combines a suitable CFL condition, a realizability-enforcing limiter, a closure procedure based on Fermi-Dirac statistics, and an IMEX scheme whose stages can be written as a convex combination of forward Euler steps combined with a backward Euler step.  
The IMEX scheme is formally only first-order accurate, but works well in the diffusion limit, and --- without interactions with the background --- reduces to the optimal second-order strong stability-preserving explicit Runge-Kutta scheme of Shu \& Osher \cite{shuOsher_1988}.  
Numerical results demonstrate the realizability-preserving properties of the scheme.  
We also demonstrate that the use of algebraic moment closures not based on Fermi-Dirac statistics can lead to unphysical moments in the context of fermion transport.  
\end{abstract}


\end{frontmatter}

\section{Introduction}
\label{sec:intro}

In this paper we design numerical methods to solve a two-moment model that governs the transport of particles obeying Fermi-Dirac statistics (e.g., neutrinos), with the ultimate target being nuclear astrophysics applications (e.g., neutrino transport in core-collapse supernovae and compact binary mergers).  
The numerical method is based on the discontinuous Galerkin (DG) method for spatial discretization and implicit-explicit (IMEX) methods for time integration, and it is designed to preserve certain physical constraints of the underlying model.  
The latter property is achieved by considering the spatial and temporal discretization together with the closure procedure for the two-moment model.  

In many applications, the particle mean free path is comparable to or exceeds other characteristic length scales in the system under consideration, and non-equilibrium effects may become important.  
In these situations, a kinetic description based on a particle distribution function may be required.  
The distribution function, a phase space density $f$ depending on momentum $\vect{p}\in\bbR^{3}$ and position $\vect{x}\in\bbR^{3}$, is defined such that $f(\vect{p},\vect{x},t)$ gives at time $t\in\bbR^{+}$ the number of particles in the phase space volume element $d\vect{p}\,d\vect{x}$ (i.e., $d\cN=f\,d\vect{p}\,d\vect{x}$).  
The evolution of the distribution function is governed by the Boltzmann equation, which states a balance between phase space advection and particle collisions (see, e.g., \cite{braginskii_1965,chapmanCowling_1970,lifshitzPitaevskii_1981}).  

Solving the Boltzmann equation numerically for $f$ is challenging, in part due to the high dimensionality of phase space.  
To reduce the dimensionality of the problem and make it more computationally tractable, one may instead solve (approximately) for a finite number of angular moments $\vect{m}_{N}=(m^{(0)},m^{(1)},\ldots,m^{(N)})^{T}$ of the distribution function, defined as
\begin{equation}
  m^{(k)}(\varepsilon,\vect{x},t)=\f{1}{4\pi}\int_{\bbS^{2}}f(\omega,\varepsilon,\vect{x},t)\,g^{(k)}(\omega)\,d\omega,
\end{equation}
where $\varepsilon=|\vect{p}|$ is the particle energy, $\omega$ is a point on the unit sphere $\bbS^{2}$ indicating the particle propagation direction, and $g^{(k)}$ are momentum space angular weighing functions.  
In problems where collisions are sufficiently frequent, solving a \emph{truncated moment problem} can provide significant reductions in computational cost since only a few moments are needed to represent the solution accurately.  
On the other hand, in problems where collisions do not sufficiently isotropize the distribution function, more moments may be needed.  
In the two-moment model considered here ($N=1$), angular moments representing the particle density and flux (or energy density and momentum) are solved for.  
Two-moment models for relativistic systems appropriate for nuclear astrophysics applications have been discussed in, e.g., \cite{lindquist_1966,andersonSpiegel_1972,thorne_1981,shibata_etal_2011,cardall_etal_2013a}.  
However, in this paper, for simplicity (and clarity), we consider a non-relativistic model, leaving extensions to relativistic systems for future work.  

In a truncated moment model, the equation governing the evolution of the \mbox{$N$-th} moment $m^{(N)}$ contains higher moments $\{m^{(k)}\}_{k=N+1}^{M}$ ($M>N$), which must be specified in order to form a closed system of equations.  
For the two-moment model, the symmetric rank-two Eddington tensor (proportional to the pressure tensor) must be specified.  
Approaches to this \emph{closure problem} include setting $m^{(k)}=0$, for $k>N$ ($P_N$ equations \cite{brunnerHolloway_2005} and filtered versions thereof \cite{mcclarrenHauck_2010,laboure_etal_2016}), Eddington approximation (when $N=0$) \cite{mihalasMihalas_1999}, Kershaw-type closure \cite{kershaw_1976}, and maximum entropy closure \cite{minerbo_1978,cernohorskyBludman_1994,olbrant_etal_2013}.  
The closure procedure often results in a system of nonlinear hyperbolic conservation laws, which can be solved using suitable numerical methods (e.g., \cite{leveque_1992}).  

One challenge in solving the closure problem is constructing a sequence of moments that are consistent with a positive distribution function, which typically implies algebraic constraints on the moments \cite{kershaw_1976,levermore_1984}.  
Moments satisfying these constraints are called \emph{realizable moments} (e.g., \cite{levermore_1996}).  
When evolving a truncated moment model numerically, maintaining realizable moments is challenging, but necessary in order to ensure the well-posedness of the closure procedure \cite{levermore_1996,junk_1998,hauck_2008}.  
In addition to putting the validity of the numerical results into question, failure to maintain moment realizability in a numerical model may, in order to continue a simulation, require ad hoc post-processing steps with undesirable consequences such as loss of conservation.  

Here we consider a two-moment model for particles governed by Fermi-Dirac statistics.  
It is well known from the two-moment model for particles governed by Maxwell-Boltzmann statistics (``classical'' particles with $f\ge0$), that the particle density is nonnegative and the magnitude of the flux vector is bounded by the particle density.  
(There are further constraints on the components of the Eddington tensor \cite{levermore_1984}.)  
Furthermore, the set of realizable moments generated by the particle density and flux vector constitutes a convex cone \cite{olbrant_etal_2012}.  
In the fermionic case, there is also an upper bound on the distribution function (e.g., $f\le1$) because Pauli's exclusion principle prevents particles from occupying the same microscopic state.  
The fermionic two-moment model has recently been studied theoretically in the context of maximum entropy closures \cite{lareckiBanach_2011,banachLarecki_2013,banachLarecki_2017b} and Kershaw-type closures \cite{banachLarecki_2017a}.  
Because of the upper bound on the distribution function, the algebraic constraints on realizable moments differ from the classical case with no upper bound, and can lead to significantly different dynamics when the occupancy is high (i.e., when $f$ is close to its upper bound).  
In the fermionic case, the set of realizable moments generated by the particle density and flux vector is also convex. 
It is ``eye-shaped'' (as will be shown later; cf. Figure~\ref{fig:RealizableSetFermionic} in Section~\ref{sec:realizability}) and tangent to the classical realizability cone on the end representing low occupancy, but is much more restricted for high occupancy.  

In this paper, the two-moment model is discretized in space using high-order Discontinuous Galerkin (DG) methods (e.g., \cite{cockburnShu_2001,hesthavenWarburton_2008}).  
DG methods combine elements from both spectral and finite volume methods and are an attractive option for solving hyperbolic partial differential equations (PDEs).  
They achieve high-order accuracy on a compact stencil; i.e., data is only communicated with nearest neighbors, regardless of the formal order of accuracy, which can lead to a high computation to communication ratio, and favorable parallel scalability on heterogeneous architectures has been demonstrated \cite{klockner_etal_2009}.  
Furthermore, they can easily be applied to problems involving curvilinear coordinates (e.g., beneficial in numerical relativity \cite{teukolsky_2016}).  
Importantly, DG methods exhibit favorable properties when collisions with a background are included, as they recover the correct asymptotic behavior in the diffusion limit, characterized by frequent collisions (e.g., \cite{larsenMorel_1989,adams_2001,guermondKanschat_2010}).  
The DG method was introduced in the 1970s by Reed \& Hill \cite{reedHill_1973} to solve the neutron transport equation, and has undergone remarkable developments since then (see, e.g., \cite{shu_2016} and references therein).  

We are concerned with the development and application of DG methods for the fermionic two-moment model that can preserve the aforementioned algebraic constraints and ensure realizable moments, provided the initial condition is realizable.  
Our approach is based on the constraint-preserving (CP) framework introduced in \cite{zhangShu_2010a}, and later extended to the Euler equations of gas dynamics in \cite{zhangShu_2010b}.  
(See, e.g., \cite{xing_etal_2010,zhangShu_2011,olbrant_etal_2012,cheng_etal_2013,zhang_etal_2013,endeve_etal_2015,wuTang_2015} for extensions and applications to other systems.)  
The main ingredients include (1) a realizability-preserving update for the cell averaged moments based on forward Euler time stepping, which evaluates the polynomial representation of the DG method in a finite number of quadrature points in the local elements and results in a Courant-Friedrichs-Lewy (CFL) condition on the time step; (2) a limiter to modify the polynomial representation to ensure that the algebraic constraints are satisfied point-wise without changing the cell average of the moments; and (3) a time stepping method that can be expressed as a convex combination of Euler steps and therefore preserves the algebraic constraints (possibly with a modified CFL condition).  
As such, our method is an extension of the realizability-preserving scheme developed by Olbrant el al. \cite{olbrant_etal_2012} for the classical two-moment model.  

The DG discretization leaves the temporal dimension continuous.  
This semi-discretization leads to a system of ordinary differential equations (ODEs), which can be integrated with standard ODE solvers (i.e., the method of lines approach to solving PDEs).  
We use implicit-explicit (IMEX) Runge-Kutta (RK) methods \cite{ascher_etal_1997,pareschiRusso_2005} to integrate the two-moment model forward in time.  
This approach is motivated by the fact that we can resolve time scales associated with particle streaming terms in the moment equations, which will be integrated with explicit methods, while terms associated with collisional interactions with the background induce fast time scales that we do not wish to resolve, and will be integrated with implicit methods.  
This splitting has some advantages when solving kinetic equations since the collisional interactions may couple across momentum space, but are local in position space, and are easier to parallelize than a fully implicit approach.  

The CP framework of \cite{zhangShu_2010a} achieves high-order (i.e., greater than first-order) accuracy in time by employing strong stability-preserving explicit Runge-Kutta (SSP-RK) methods \cite{shuOsher_1988,gottlieb_etal_2001}, which can be written as a convex combination of forward Euler steps.  
Unfortunately, this strategy to achieve high-order temporal accuracy does not work as straightforwardly for standard IMEX Runge-Kutta (IMEX-RK) methods because implicit SSP Runge-Kutta methods with greater than first-order accuracy have time step restrictions similar to explicit methods \cite{gottlieb_etal_2001}.  
To break this ``barrier,'' recently proposed IMEX-RK schemes \cite{chertock_etal_2015,hu_etal_2018} have resorted to first-order accuracy in favor of the SSP property in the standard IMEX-RK scheme, and recover second-order accuracy with a correction step.  

We consider the application of the correction approach to the two-moment model.  
However, with the correction step from \cite{chertock_etal_2015} we are unable to prove the realizability-preserving property without invoking an overly restrictive time step.  
With the correction step from \cite{hu_etal_2018} the realizability-preserving property is guaranteed with a time step comparable to that of the forward Euler method applied to the explicit part of the scheme, but the resulting scheme performs poorly in the asymptotic diffusion limit.  
Because of these challenges, we resort to first-order temporal accuracy, and propose IMEX-RK schemes that are convex-invariant with a time step equal to that of forward Euler on the explicit part, perform well in the diffusion limit, and reduce to a second-order SSP-RK scheme in the streaming limit (no collisions with the background material).  

The realizability-preserving property of the DG-IMEX scheme depends sensitively on the adopted closure procedure.  
The explicit update of the cell average can, after employing the simple Lax-Friedrichs flux and imposing a suitable CFL condition on the time step, be written as a convex combination.  
Realizability of the updated cell average is then guaranteed from convexity arguments \cite{zhangShu_2010a}, provided all the elements in the convex combination are realizable.  
Realizability of individual elements in the convex combination is conditional on the closure procedure (components of the Eddington tensor must be computed to evaluate numerical fluxes).  
We prove that each element in the convex combination is realizable provided the moments involved in expressing the elements are moments of a distribution function satisfying the bounds implied by Fermi-Dirac statistics (i.e., $0\le f \le 1$).  
For algebraic two-moment closures, which we consider, the so-called Eddington factor is given by an algebraic expression depending on the evolved moments and completely determines the components of the Eddington tensor.  
Realizable components of the Eddington tensor demand that the Eddington factor satisfies strict lower and upper bounds (e.g., \cite{levermore_1984,lareckiBanach_2011}).  
We discuss algebraic closures derived from Fermi-Dirac statistics that satisfy these bounds, and demonstrate with numerical experiments that the DG-IMEX scheme preserves realizability of the moments when these closures are used.  
We also demonstrate that further approximations to algebraic two-moment closures for modeling particle systems governed by Fermi-Dirac statistics may give results that are incompatible with a bounded distribution and, therefore, unphysical.  
The example we consider is the Minerbo closure \cite{minerbo_1978}, which can be obtained as the low occupancy limit of the maximum entropy closure of Cernohorsky \& Bludman \cite{cernohorskyBludman_1994}.  

The paper is organized as follows.  
In Section~\ref{sec:model} we present the two-moment model.  
In Section~\ref{sec:realizability} we discuss moment realizability for the fermionic two-moment model, while algebraic moment closures are discussed in Section~\ref{sec:algebraicClosure}.  
In Section~\ref{sec:dg} we briefly introduce the DG method for the two-moment model, while the (convex-invariant) IMEX time stepping methods we use are discussed in Section~\ref{sec:imex}.  
The main results on the realizability-preserving DG-IMEX method for the fermionic two-moment model are worked out in Sections~\ref{sec:realizableDGIMEX} and \ref{sec:limiter}.  
In Section~\ref{sec:limiter} we also discuss the realizability-enforcing limiter.  
Numerical results are presented in Section~\ref{sec:numerical}, and summary and conclusions are given in Section~\ref{sec:conclusions}.  
Additional details on the IMEX schemes are provided in Appendices.  
\section{Mathematical Model}
\label{sec:model}

In this section we give a summary of the mathematical model.  

\subsection{Boltzmann Equation}

We consider approximate solutions to the Boltzmann equation for the transport of massless particles through a static material in Cartesian geometry, which, after scaling to dimensionless units, can be written as
\begin{equation}
  \pd{f}{t}+\vect{\ell}\cdot\nabla f
  =\f{1}{\tau}\,\cC(f),
  \label{eq:boltzmann}
\end{equation}
where the distribution function $f\colon(\omega,\varepsilon,\vect{x},t)\in\bbS^{2}\times\bbR^{+}\times\bbR^{3}\times\bbR^{+}\to\bbR^{+}$ gives the number of particles propagating in the direction $\omega\in\bbS^{2}:=\{\,\omega=(\thetaNu,\phiNu)~|~\thetaNu\in[0,\pi],\phiNu\in[0,2\pi)\,\}$, with energy $\varepsilon\in\bbR^{+}$, at position $\vect{x}\in\bbR^{3}$ and time $t\in\bbR^{+}$.  
Here we use spherical momentum space coordinates $(\varepsilon,\omega)$, and the unit vector $\vect{\ell}(\omega)\in\bbR^{3}$ (independent of $\varepsilon$ and $\vect{x}$) is parallel to the particle three-momentum $\vect{p}=\varepsilon\,\vect{\ell}$.  
We also define the energy-position coordinates $\vect{z}:=\{\varepsilon,\vect{x}\}\in\bbR^{+}\times\bbR^{3}$.  
On the right-hand side of Eq.~\eqref{eq:boltzmann}, $\tau$ is the ratio of the particle mean-free path (due to interactions with a background) to some characteristic length scale of the problem.  
In opaque regions, $\tau\ll1$, while for free streaming particles, $\tau\gg1$.  
The collision operator, which models emission, absorption, and isotropic and elastic scattering, is given by
\begin{equation}
  \cC(f)=\xi\,\big(\,f_{0}-f\,\big)
  +(1-\xi)\,\big(\,\f{1}{4\pi}\int_{\bbS^{2}}f\,d\omega-f\,\big),
  \label{eq:collisionTerm}
\end{equation}
where $\xi=\sigma_{\Ab}/\sigma_{\Tot}\in[0,1]$ is the ratio of the absorption opacity $\sigma_{\Ab}\,(\ge0)$ to the total opacity $\sigma_{\Tot}=\sigma_{\Ab}+\sigma_{\Scatt}$, and $\sigma_{\Scatt}\,(\ge0)$ is the scattering opacity.  
In particular, $\xi=1$ models pure emission and absorption, while $\xi=0$ models pure scattering.  
In general, $\sigma_{\Ab}$ and $\sigma_{\Scatt}$ (and $\tau$ and $\xi$) depend on $\vect{z}$.  
The equilibrium distribution function is denoted by $f_{0}(\vect{z})$.  
Here, we consider transport of Fermions (e.g., neutrinos), so the equilibrium distribution function takes the form
\begin{equation}
  f_{0}(\vect{z})=\f{1}{e^{(\varepsilon-\mu(\vect{x}))/T(\vect{x})}+1},  
  \label{eq:fermiDirac}
\end{equation}
where the temperature $T$ and the chemical potential $\mu$ depend on properties of the background.  

\subsection{Angular Moment Equations: Two-Moment Model}

The Boltzmann equation is often too expensive to solve directly.  
Instead, approximate equations for angular moments of the distribution function are solved.  
To this end, we define the angular moments of the distribution function
\begin{equation}
  \big\{\,\cJ,\vect{\cH},\vect{\cK}\,\big\}(\vect{z},t)
  =\f{1}{4\pi}\int_{\bbS^{2}}f(\omega,\vect{z},t)\,\{\,1,\vect{\ell},\vect{\ell}\otimes\vect{\ell}\,\}\,d\omega.  
  \label{eq:angularMoments}
\end{equation}
We refer to $\cJ$ (zeroth moment) as the particle density, $\vect{\cH}$ (first moment) as the particle flux, and $\vect{\cK}$ (second moment) as the stress tensor.  
Note that the moments defined in Eq.~\eqref{eq:angularMoments} are \emph{spectral moments} (depending on energy as well as position and time).  
The \emph{grey moments} (depending only on position and time) are obtained by integration over energy:
\begin{equation}
  \big\{\,J,\vect{H},\vect{K}\,\big\}(\vect{x},t)
  =\int_{\bbR^{+}}\big\{\,\cJ,\vect{\cH},\vect{\cK}\,\big\}(\varepsilon,\vect{x},t)\,\varepsilon^{2}d\varepsilon.  
\end{equation}

Taking the zeroth and first moments of Eq.~\eqref{eq:boltzmann} gives the two-moment model, comprising a system of conservation laws with sources
\begin{equation}
  \pd{\vect{\cM}}{t}+\nabla\cdot\vect{\cF}=\f{1}{\tau}\,\vect{\cC}(\vect{\cM}),
  \label{eq:momentEquations}
\end{equation}
where $\vect{\cM}=(\cJ,\vect{\cH})^{T}$ and $\vect{\cF}=(\vect{\cH},\vect{\cK})^{T}$.  
Components of the fluxes in each coordinate direction are $\vect{\cF}^{i}=\vect{e}_{i}\cdot\vect{\cF}=(\vect{e}_{i}\cdot\vect{\cH},\vect{e}_{i}\cdot\vect{\cK})^{T}$, where $\vect{e}_{i}$ is the unit vector parallel to the $i$th coordinate direction.  
On the right-hand side of Eq.~\eqref{eq:momentEquations}, the source term is
\begin{equation}
  \vect{\cC}(\vect{\cM})=\vect{\eta}-\vect{\cD}\,\vect{\cM}, 
  \label{eq:collisionTermMoments}
\end{equation}
where $\vect{\eta}=(\xi\,f_{0},\vect{0})^{T}$ and $\vect{\cD}=\mbox{diag}(\xi,\vect{I})$, with $\vect{I}$ the identity matrix.  

In order to close the system given by Eq.~\eqref{eq:momentEquations}, the components of the stress tensor $\vect{\cK}$ must be related to the lower moments through a closure procedure.  
To this end, Levermore \cite{levermore_1984} defined the Eddington tensor $\vect{k}=\vect{\cK}/\cJ$ and assumed that the radiation field is symmetric about a preferred direction $\widehat{\vect{h}}=\vect{\cH}/|\vect{\cH}|$ so that
\begin{equation}
  \vect{k}=\f{1}{2}\big[\,\big(1-\chi\big)\,\vect{I}+\big(3\,\chi-1\big)\,\widehat{\vect{h}}\otimes\widehat{\vect{h}}\,\big],
  \label{eq:eddingtonTensor}
\end{equation}
where $\chi=\chi(\cJ,|\vect{\cH}|)$ is the Eddington factor.  
The two-moment model is then closed once the Eddington factor is determined from $\cJ$ and $\vect{\cH}$.  
We will return to the issue of determining the Eddington factor in Section~\ref{sec:algebraicClosure}.  
\section{Moment Realizability for the Fermionic Two-Moment Model}
\label{sec:realizability}

Our goal is to simulate massless fermions (e.g., neutrinos) and study their interactions with matter.  
The principal objective is to obtain the fermionic distribution function $f$ (or moments of $f$ as in the two-moment model employed here).  
The Pauli exclusion principle requires the distribution function to satisfy the condition $0 \le f \le 1$, which puts restrictions on the admissible values for the moments of $f$.  
In this paper, we seek to design a numerical method for solving the system of moment equations given by Eq.~\eqref{eq:momentEquations} that preserves realizability of the moments; i.e., the moments evolve within the set of admissible values as dictated by Pauli's exclusion principle.  
(Since we are only concerned with the angular dependence of $f$ in this section, we simplify the notation by suppressing the $\vect{z}$ and $t$ dependence and write $f(\omega,\vect{z},t)=f(\omega)$.)  

We begin with the following definition of moment realizability.  
\begin{define}
  The moments $\vect{\cM}=\big(\cJ,\vect{\cH}\big)^{T}$ are realizable if they can be obtained from a distribution function satisfying $0 < f(\omega) < 1~\forall~\omega\in\bbS^{2}$.  
  The set of all realizable moments $\cR$ is
  \begin{equation}
    \cR:=\big\{\,\vect{\cM}=\big(\cJ,\vect{\cH}\big)^{T}~|~\cJ\in(0,1)~\text{and}~\gamma(\vect{\cM}) > 0\,\big\},
    \label{eq:realizableSet}
  \end{equation}
  where we have defined the concave function $\gamma(\vect{\cM})\equiv(1-\cJ)\cJ-|\vect{\cH}|$.  
  \label{def:momentRealizability}
\end{define}
\begin{rem}
  Following \cite{lareckiBanach_2011}, in Definition~\ref{def:momentRealizability}, and in the rest of this paper, we exclude the cases $f=0$ and $f=1$ almost everywhere (a.e.) on $\bbS^{2}$, which would give $\cJ=0$, $\vect{\cH}=0$ and $\cJ=1$, $\vect{\cH}=0$, respectively.  
\end{rem}

The algebraic constraints in Eq.~\eqref{eq:realizableSet} are proven in \cite{banachLarecki_2017a} (see also \cite{lareckiBanach_2011,banachLarecki_2013}).  
%\begin{lemma}
%  Suppose that $0 < f(\omega) < 1~\forall~\omega\in\bbS^{2}$.  
%  Then the moments $\vect{\cM}=(\cJ,\vect{\cH})^{T}$, defined as in Eq.~\eqref{eq:angularMoments}, satisfy $0 < \cJ < 1$ and $\big(1-\cJ\big)\,\cJ-|\vect{\cH}| > 0$. 
%  \label{lem:MomentRealizable} 
%\end{lemma}

%\begin{define}
%  For the fermion two-moment model, the realizable set is
%  \begin{equation}
%    \cR:=\big\{\,\vect{\cM}=\big(\cJ,\vect{\cH}\big)^{T}~|~\cJ\in(0,1)~\text{and}~\gamma(\vect{\cM})\equiv(1-\cJ)\cJ-|\vect{\cH}| > 0\,\big\}.
%    \label{eq:realizableSet}
%  \end{equation}
%\end{define}

\begin{lemma}
  The realizable set $\cR$ is convex.  
\end{lemma}
\begin{proof}
%  In order to prove that $\cR$ is convex, it is sufficient to show that any convex combination of elements in $\cR$ also belongs to $\cR$.  
  Let $\vect{\cM}_{a}=\big(\cJ_{a},\vect{\cH}_{a}\big)^{T}$ and $\vect{\cM}_{b}=\big(\cJ_{b},\vect{\cH}_{b}\big)^{T}$ be two arbitrary elements in $\cR$, and let $\vect{\cM}_{c} = \theta\,\vect{\cM}_{a} + (1-\theta)\,\vect{\cM}_{b}$, with $0\leq\theta\leq1$.
  The first component of $\vect{\cM}_{c}$ is
  \begin{equation*}
    \cJ_{c} = \theta\,\cJ_{a} + (1-\theta)\,\cJ_{b}.
  \end{equation*}
  Since $\cJ_{a},\cJ_{b} \in (0,1)$, it follows that $\cJ_{c} \in (0,1)$.  
  Concavity of $\gamma$ implies that
%  Since $\cJ_{a},\cJ_{b} \in (0,1)$, it is straightforward to verify that
  \begin{equation*}
  \gamma(\vect{\cM}_{c}) \geq \theta\,\gamma(\vect{\cM}_{a}) + (1-\theta)\,\gamma(\vect{\cM}_{b}) > 0.
  \end{equation*}
  Hence, $\vect{\cM}_{c}\in\cR$.
\end{proof}

Figure~\ref{fig:RealizableSetFermionic} illustrates the geometry of the convex set $\cR$ in the $(\cH,\cJ)$-plane (light blue region).  
The boundary $\partial\cR$ (black curves) is given by $\gamma(\vect{\cM})=0$.  
The realizable domain of positive distribution functions, $\cR^{+}$ (no upper bound on $f$), which is a convex cone defined by
\begin{equation}
  \cR^{+}:=\big\{\,\vect{\cM}=\big(\cJ,\vect{\cH}\big)^{T}~|~\cJ > 0~\text{and}~\cJ > |\vect{\cH}|\,\big\}, 
  \label{eq:realizableSetPositive}
\end{equation}
is partially shown as the light red region above the red lines, which mark the boundary of $\cR^{+}$ (denoted $\partial\cR^{+}$).  
The realizable set $\cR$ is a bounded subset of $\cR^{+}$.  
%The realizable domains $\cR$ and $\cR^{+}$ overlap for low particle densities ($\cJ\ll1$), but for larger values of $\cJ$, the realizable domain of particles governed by Fermi-Dirac statistics is much more restricted than that of particles described by positive distribution functions with no upper bound.  

\begin{figure}[H]
  \centering
  \includegraphics[width=1.0\linewidth]{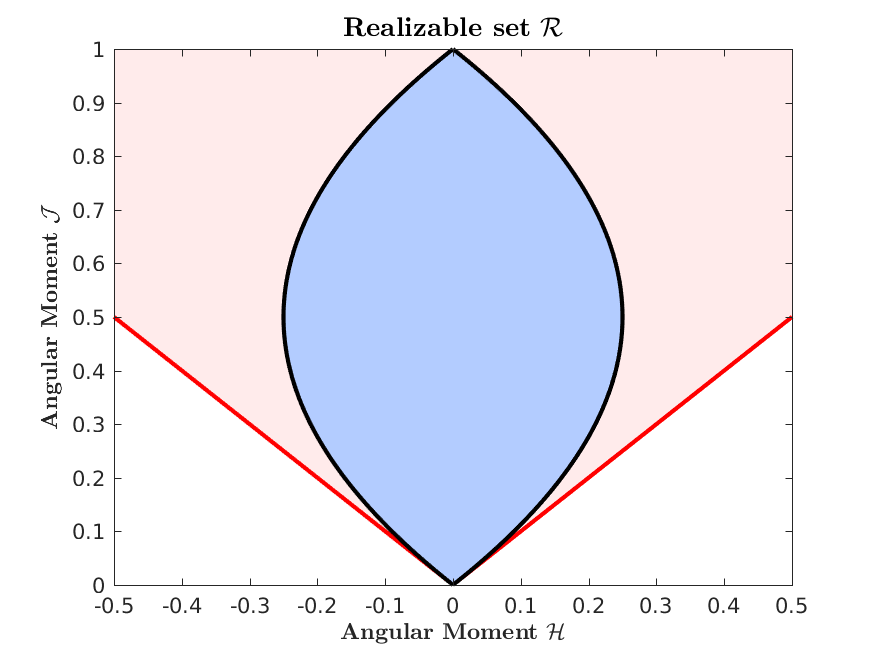}
  \caption{Illustration of the realizable set $\cR$ (light blue region) defined in Eq.~\eqref{eq:realizableSet}.  
  The black lines define the boundary $\partial\cR$, while the red lines indicate the boundary of the realizable set $\cR^{+}$ (light red region) defined in Eq.~\eqref{eq:realizableSetPositive}.}
  \label{fig:RealizableSetFermionic}
\end{figure}

For the realizability-preserving scheme developed in Section~\ref{sec:realizableDGIMEX}, we state some additional results.  
Lemma~\ref{lem:explicitStep} is used to help prove the realizability-preserving property of explicit steps in the IMEX scheme, while Lemmas~\ref{lem:implicitStep} and \ref{lem:correctionStep} are used to prove realizability-preserving properties of implicit steps.  
\begin{lemma}
  Let $\big\{\cJ_{a},\vect{\cH}_{a},\vect{\cK}_{a}\big\}$ and $\big\{\cJ_{b},\vect{\cH}_{b},\vect{\cK}_{b}\big\}$ be moments defined as in Eq.~\eqref{eq:angularMoments} with distribution functions $f_{a}$ and $f_{b}$, respectively, such that $f_{a}(\omega),f_{b}(\omega)\in(0,1)\,\forall\,\omega\in\bbS^{2}$.  
  Let $\Phi^{\pm}(\vect{\cM},\vect{\cK})=\f{1}{2}\big(\vect{\cM}\pm\widehat{\vect{e}}\cdot\vect{\cF}\big)$, where $\widehat{\vect{e}}\in\bbR^{3}$ is an arbitrary unit vector, and $\widehat{\vect{e}}\cdot\vect{\cF}=\big(\widehat{\vect{e}}\cdot\vect{\cH},\widehat{\vect{e}}\cdot\vect{\cK}\big)^{T}$.  
  Then
  \begin{equation*}
    \vect{\cM}_{ab} \equiv \Phi^{+}(\vect{\cM}_{a},\vect{\cK}_{a})+\Phi^{-}(\vect{\cM}_{b},\vect{\cK}_{b})\in\cR.
  \end{equation*}
  \label{lem:explicitStep}
\end{lemma}
\begin{proof}
  The components of $\vect{\cM}_{ab}$ are
  \begin{equation*}
    \cJ_{ab}=\f{1}{4\pi}\int_{\bbS^{2}}f_{ab}(\omega)\,d\omega
    \quad\text{and}\quad
    \vect{\cH}_{ab}=\f{1}{4\pi}\int_{\bbS^{2}}f_{ab}(\omega)\,\vect{\ell}(\omega)\,d\omega,
  \end{equation*}
  where $f_{ab}(\omega)=\vartheta\,f_{a}(\omega)+(1-\vartheta)\,f_{b}(\omega)$ and $\vartheta(\omega)=(1+\widehat{\vect{e}}\cdot\vect{\ell}(\omega))/2\in[0,1]$.  
  Then, since $f_{ab}(\omega)\in(0,1)\,\forall\,\omega\in\bbS^{2}$, it follows that $\vect{\cM}_{ab}\in\cR$.  
\end{proof}

\begin{lemma}
  Let $\vect{\cM}_{a}=(\cJ_{a},\vect{\cH}_{a})^{T}\in\cR$ and $\alpha>0$.  
  Let $\vect{\cM}_{b}=(\cJ_{b},\vect{\cH}_{b})^{T}$ satisfy
  \begin{equation}
    \vect{\cM}_{b}=\vect{\cM}_{a}+\alpha\,\vect{\cC}(\vect{\cM}_{b}), 
    \label{eq:implicitStep}
  \end{equation}
  where $\vect{\cC}(\vect{\cM})=\vect{\eta}-\vect{\cD}\,\vect{\cM}$ is the collision term in Eq.~\eqref{eq:collisionTermMoments}.  
  Then $\vect{\cM}_{b}\in\cR$.  
  \label{lem:implicitStep}
\end{lemma}
\begin{proof}
  Solving Eq.~\eqref{eq:implicitStep} for $\vect{\cM}_{b}$ gives $\vect{\cM}_{b} = \big(\vect{I}+\alpha\,\vect{\cD}\big)^{-1}\big(\vect{\cM}_{a}+\alpha\,\vect{\eta}\big)$.  
  The first component of $\vect{\cM}_{b}$ can be written as
  \begin{equation*}
    \cJ_{b}=\f{1}{4\pi}\int_{\bbS}f_{b}(\omega)\,d\omega,
  \end{equation*}
  where $f_{b}(\omega)=\zeta\,f_{a}(\omega)+(1-\zeta)\,f_{0}$, $\zeta=1/(1+\alpha\,\xi)\in[0,1]$, and $f_{0}$ and $\xi$ are defined in Eq.~\eqref{eq:collisionTerm} in Section~\ref{sec:model}.  
  Then, since $f_{b}(\omega)\in(0,1)\,\forall\,\omega\in\bbS^{2}$, $\cJ_{b}\in(0,1)$.  
  Meanwhile, 
  \begin{equation*}
    \vect{\cH}_{b}=\f{(1+\alpha\,\xi)}{(1+\alpha)}\,\widetilde{\vect{\cH}}_{b},
    \quad\text{where}\quad
    \widetilde{\vect{\cH}}_{b}=\f{1}{4\pi}\int_{\bbS^{2}}f_{b}(\omega)\,\vect{\ell}(\omega)\,d\omega.  
  \end{equation*}
  It follows that $\widetilde{\vect{\bcM}}_{b}=(\cJ_{b},\widetilde{\vect{\cH}}_{b})^{T}\in\cR$.  
  Then, since $0\le\xi\le1$, $|\vect{\cH}_{b}|\le|\widetilde{\vect{\cH}}_{b}| < (1-\cJ_{b})\,\cJ_{b}$.  
\end{proof}

\begin{lemma}
  Let $\vect{\cM}_{a}=(\cJ_{a},\vect{\cH}_{a})^{T}\in\cR$ and $\alpha>0$.  
  Let $\vect{\cM}_{b}$ satisfy
  \begin{equation*}
    \vect{\cM}_{b}=\vect{\cM}_{a}+\alpha\,\vect{\cD}\,\vect{\cC}(\vect{\cM}_{b}),    
  \end{equation*}
  where $\vect{\cD}$ and $\vect{\cC}(\vect{\cM})$ are given by Eq.~\eqref{eq:collisionTermMoments}.  
  Then $\vect{\cM}_{b}\in\cR$.  
  \label{lem:correctionStep}
\end{lemma}
The proof of Lemma~\ref{lem:correctionStep} follows along the same lines as the proof of Lemma~\ref{lem:implicitStep} and is omitted.  

\section{Algebraic Moment Closures}
\label{sec:algebraicClosure}

The two-moment model given by Eq.~\eqref{eq:angularMoments} is not closed because of the appearance of the second moments $\vect{\cK}$ (the normalized pressure tensor).  
Algebraic moment closures for the two-moment model are computationally efficient as they provide the Eddington factor in Eq.~\eqref{eq:eddingtonTensor} in closed form as a function of the density $\cJ$ and the flux factor $h=|\vect{\cH}|/\cJ$.  
For this reason they are used in applications where transport plays an important role, but where limited computational resources preclude the use of higher fidelity models.  
Examples include simulation of neutrino transport in core-collapse supernovae \cite{roberts_etal_2016} and compact binary mergers \cite{foucart_etal_2015}.  
Algebraic moment closures in the context of these aforementioned applications have also been discussed elsewhere (e.g., \cite{janka_etal_1992,pons_etal_2000,smit_etal_2000,just_etal_2015,murchikova_etal_2017}).  
Here we focus on properties of the algebraic closures that are critical to the development of numerical methods for the two-moment model of fermion transport.  
For the algebraic closures we consider, the Eddington factor in Eq.~\eqref{eq:eddingtonTensor} can be written in the following form \cite{cernohorskyBludman_1994}
\begin{equation}
  \chi(\cJ,h)=\f{1}{3}+\f{2\,(1-\cJ)\,(1-2\cJ)}{3}\,\Theta\Big(\f{h}{1-\cJ}\Big),
  \label{eq:eddingtonFactor}
\end{equation}
where the \emph{closure function} $\Theta(x)$ depends on the specifics of the closure procedure.  
We will consider two basic closure procedures in more detail below: the maximum entropy (ME) closure and the Kershaw (K) closure.  

In the low occupancy limit ($\cJ\ll1$), the Eddington factor in Eq.~\eqref{eq:eddingtonFactor} depends solely on $h$; i.e.,
\begin{equation}
  \chi(\cJ,h)\to\chi_{0}(h)=\f{1}{3}+\f{2}{3}\,\Theta\big(h\big).  
  \label{eq:eddingtonFactorLow}
\end{equation}
This for of $\chi$ yields a moment closure that is suitable for particle systems obeying Maxwell-Boltzmann statistics.  

\subsection{Maximum Entropy (ME) Closure}

The ME closure constructs an approximation of the angular distribution as a function of $\cJ$ and $\vect{\cH}$ \cite{cernohorskyBludman_1994,lareckiBanach_2011}.  
The ME distribution $f_{\mbox{\tiny ME}}$ is found by maximizing the entropy functional, which for particles obeying Fermi-Dirac statistics is given by
\begin{equation}
  S[f_{\mbox{\tiny ME}}] 
  = \int_{\bbS^{2}}\big[\,(1-f_{\mbox{\tiny ME}})\log(1-f_{\mbox{\tiny ME}}) + f_{\mbox{\tiny ME}}\log f_{\mbox{\tiny ME}}\,]\,d\omega,
  \label{eq:entropyFunctional}
\end{equation} 
subject to the constraints
\begin{equation}
  \f{1}{4\pi}\int_{\bbS^{2}}f_{\mbox{\tiny ME}}(\omega)\,d\omega=\cJ
  \quad\text{and}\quad
  \f{1}{4\pi}\int_{\bbS^{2}}f_{\mbox{\tiny ME}}(\omega)\,\vect{\ell}(\omega)\,d\omega=\vect{\cH}.  
  \label{eq:closureConstraints}
\end{equation}
The solution that maximizes Eq.~\eqref{eq:entropyFunctional} takes the general form \cite{cernohorskyBludman_1994}
\begin{equation}
  f_{\mbox{\tiny ME}}(\omega;a,\vect{b})=\f{1}{e^{a + \vect{b}\cdot\vect{\ell}(\omega)}+1}, 
  \label{eq:fME}
\end{equation}
where the Lagrange multipliers $a$ and $\vect{b}$ are implicit functions of $\cJ$ and $\vect{\cH}$.  
The ME distribution function satisfies $0 < f_{\mbox{\tiny ME}} < 1$, but $a$ and $\vect{b}$ are unconstrained.  
Specification of $a$ and $\vect{b}$ from $\vect{\cM}=(\cJ,\vect{\cH})^{T}$ gives $f_{\mbox{\tiny ME}}$, and any number of moments can in principle be computed.  
Importantly, for the maximum entropy problem to be solvable, we must have $\vect{\cM}\in\cR$ \cite{lareckiBanach_2011}.  

To arrive at an algebraic form of the ME closure, Cernohorsky \& Bludman \cite{cernohorskyBludman_1994} postulate (but see \cite{lareckiBanach_2011}) that, as a function of the flux saturation
\begin{equation}
  x := h/(1-\cJ),
  \label{eq:fluxSaturation}
\end{equation} 
the closure function $\Theta$ is independent of $\cJ$ and can be written explicitly in terms of the inverse Langevin function.  
To avoid inverting the Langevin function for $\Theta$, they provide a polynomial fit (accurate to $2\%$) given by
\begin{equation}
  \Theta_{\mbox{\tiny ME}}^{\mbox{\tiny CB}}(x)
  =\f{1}{5}\,\big(\,3-x+3\,x^{2}\,\big)\,x^{2}.
  \label{eq:closureMECB}
\end{equation}
More recently, Larecki \& Banach \cite{lareckiBanach_2011} have shown that the explicit expression given in \cite{cernohorskyBludman_1994} is not exact and provide another approximate expression
\begin{equation}
  \Theta_{\mbox{\tiny ME}}^{\mbox{\tiny BL}}(x)
  =\f{1}{8}\,\big(\,9\,x^{2}-5+\sqrt{33\,x^{4}-42\,x^{2}+25}\,\big),
  \label{eq:closureMEBL}
\end{equation}
which is accurate to within $0.35\%$.  
On the interval $x\in[0,1]$, the curves given by Eqs.~\eqref{eq:closureMECB} and \eqref{eq:closureMEBL} lie practically on top of each other.  
The closure functions given by Eqs.~\eqref{eq:closureMECB} and \eqref{eq:closureMEBL}, together with the Eddington factor in Eq.~\eqref{eq:eddingtonFactor} and the pressure tensor in Eq~\eqref{eq:eddingtonTensor}, constitute the algebraic maximum entropy closures for fermionic particle systems considered in this paper.  
We will refer to the ME closures with $\Theta_{\mbox{\tiny ME}}^{\mbox{\tiny CB}}$ and $\Theta_{\mbox{\tiny ME}}^{\mbox{\tiny BL}}$ as the CB (Cernohorsky \& Bludman) and BL (Banach \& Larecki) closures, respectively.  

We also note that using the closure function given by Eq.~\eqref{eq:closureMECB} with the low occupancy Eddington factor in Eq~\eqref{eq:eddingtonFactorLow} results in the algebraic maximum entropy closure attributed to Minerbo \cite{minerbo_1978}, which is currently in use in simulation of neutrino (fermion) transport in the aforementioned nuclear astrophysics applications.  
In a recent comparison of algebraic (or analytic) closures for the two-moment model applied to neutrino transport around proto-neutron stars, Murchikova et al. \cite{murchikova_etal_2017} obtained nearly identical results when using the closures of CB and Minerbo.  
For these reasons, we include the Minerbo closure in the subsequent discussion and in the numerical tests in Section~\ref{sec:numerical}.  

\subsection{Kershaw (K) Closure}

Another algebraic closure we consider is a Kershaw-type closure \cite{kershaw_1976}, developed for fermion particle systems in \cite{banachLarecki_2017a}.  
The basic principle of the Kershaw closure for the two-moment model is derived from the fact that the realizable set generated by the triplet of scalar moments
\begin{equation}
  \{\cJ,\cH,\cK\}=\f{1}{2}\int_{-1}^{1}f(\mu)\,\mu^{\{0,1,2\}}\,d\mu,
  \label{eq:scalarMoments}
\end{equation} 
is convex.  
For the moments in Eq.~\eqref{eq:scalarMoments}, the realizable set is the set of moments obtained from distribution functions satisfying $0<f(\mu)<1,\,\forall\mu\in[-1,1]$.  
(The moments in Eq.~\eqref{eq:scalarMoments} are the unique moments obtained from the moments in Eq.~\eqref{eq:angularMoments} under the assumption that the distribution function is isotropic about a preferred direction, and $\mu$ is the cosine of the angle between this preferred direction and the particle propagation direction given by $\vect{\ell}$.)  

For a bounded distribution $0<f<1$, it is possible to show (e.g., \cite{banachLarecki_2013}) that the second moment satisfies
\begin{equation}
  \cK_{\mbox{\tiny L}}(\cJ,h) < \cK < \cK_{\mbox{\tiny U}}(\cJ,h),
\end{equation}
where $\cK_{\mbox{\tiny L}}=\cJ\,\big(\,\f{1}{3}\,\cJ^{2}+h^{2}\,\big)$, $\cK_{\mbox{\tiny U}}=\cK_{\mbox{\tiny L}} + \cJ\,(1-\cJ)\,(1-x^{2})$, and $x$ is the flux saturation defined in Eq.~\eqref{eq:fluxSaturation}.  
By convexity of the realizable set generated by the moments in Eq.~\eqref{eq:scalarMoments}, the convex combination
\begin{equation}
  \cK(\beta,\cJ,h)=\beta\,\cK_{\mbox{\tiny L}}(\cJ,h)+(1-\beta)\,\cK_{\mbox{\tiny U}}(\cJ,h),
  \label{eq:kershawAnsatz}
\end{equation}
with $\beta\in[0,1]$, is realizable whenever $(\cJ,\cH)^{T}\in\cR$.  
The Kershaw closure for the two-moment model is then obtained from Eq.~\eqref{eq:kershawAnsatz} with the additional requirement that it be correct in the limit of isotropic distribution functions; i.e., $\cK(\beta,\cJ,0)=\cJ/3$.  
One choice for $\beta$, which leads to a strictly hyperbolic and causal two-moment model (and a particularly simple closure function) \cite{banachLarecki_2017a}, is $\beta=(2-\cJ)/3$, so that $\cK_{\mbox{\tiny K}}(\cJ,h)=\chi_{\mbox{\tiny K}}(\cJ,h)\,\cJ$, where
\begin{equation}
  \chi_{\mbox{\tiny K}}(\cJ,h)=\f{1}{3}+\f{2\,(1-\cJ)\,(1-2\cJ)}{3}\,\Theta_{\mbox{\tiny K}}\Big(\f{h}{1-\cJ}\Big),
  \label{eq:eddingtonFactorKershaw}
\end{equation}
and the Kershaw closure function is given by
\begin{equation}
  \Theta_{\mbox{\tiny K}}(x)=x^{2}.  
\end{equation}
For multidimensional problems, the Kershaw closure is obtained by using the Eddington factor in Eq.~\eqref{eq:eddingtonFactorKershaw} in Eq.~\eqref{eq:eddingtonTensor}.  
Finally, we point out that for the two-moment Kershaw closure (see \cite{banachLarecki_2017a} for details), a distribution function $f_{\mbox{\tiny K}}(\omega,\cJ,\vect{\cH})$, satisfying $0 < f_{\mbox{\tiny K}} < 1$, and reproducing the moments $\cJ$, $\vect{\cH}$, and $\vect{\cK}$, can be written explicitly in terms of Heaviside functions.  

\subsection{Realizability of Algebraic Moment Closures}

It is not immediately obvious that all the algebraic moment closures discussed above are suitable for designing realizability-preserving methods for the two-moment model of fermion transport.  
In particular, the realizability-preserving scheme developed in this paper is based on the result in Lemma~\ref{lem:explicitStep}, which must hold for the adapted closure.  
The Kershaw closure is consistent with a bounded distribution, $f_{_{\mbox{\tiny K}}}\in(0,1)$, and should be well suited, but the algebraic ME closures are based on approximations to the closure function, and we need to consider if these approximate closures remain consistent with the assumed bounds on the underlying distribution function.  
To this end, we rely on results in \cite{levermore_1984,lareckiBanach_2011} (see also \cite{kershaw_1976,shohatTamarkin_1943}), which state that realizability of the moment triplet $\{\cJ,\vect{\cH},\vect{\cK}\}$ (with $\vect{\cK}$ given by Eq.~\eqref{eq:eddingtonTensor}), is equivalent to the following requirement for the Eddington factor
\begin{equation}
  \chi_{\mbox{\tiny min}}
  =\max\big(1-\f{2}{3\cJ},h^{2}\big)
  <\chi<\min\big(1,\f{1}{3\cJ}-\f{\cJ}{1-\cJ}h^{2}\big)=\chi_{\mbox{\tiny max}}.  
  \label{eq:eddingtonFactorBounds}
\end{equation}
Fortunately, these bounds are satisfied by the algebraic closures based on Fermi-Dirac statistics.  
(Note that for $\cJ\ll1$ the bounds in Eq.~\eqref{eq:eddingtonFactorBounds} limit to the bounds for positive distributions given by Levermore \cite{levermore_1984}; i.e., $h^{2}<\chi<1$.)

In Figure~\ref{fig:EddingtonFactorsWithDifferentClosure}, we plot the Eddington factor $\chi$ versus the flux factor $h$ for the various algebraic closures discussed above and for different values of $\cJ\in(0,1)$: $0.01$ (upper left panel), $0.4$ (upper right panel), $0.6$ (lower left panel), and $0.99$ (lower right panel).  
The lower and upper bounds on the Eddington factor for realizable closures ($\chi_{\mbox{\tiny min}}$ and $\chi_{\mbox{\tiny max}}$, respectively) are also plotted.  
We note that for all the closures, the Eddington factor $\chi\to1/3$ as $h\to0^{+}$.  

When $\cJ=0.01$, the maximum entropy closures (CB, BL, and Minerbo) are practically indistinguishable, while the Eddington factor of the Kershaw closure is larger than that of the other closures over most of the domain.  
When $\cJ=0.4$, the Eddington factor for the closures based on Fermi-Dirac statistics (CB, BL, and Kershaw) remain close together, while the Eddington factor for the Minerbo closure is larger than the other closures for $h\gtrsim0.2$.  
The Eddington factor for all closures remain between $\chi_{\mbox{\tiny min}}$ and $\chi_{\mbox{\tiny max}}$ when $\cJ=0.01$ and $\cJ=0.4$.  

When $\cJ=0.6$, the Eddington factor for the closures based on Fermi-Dirac statistics remain close together and within the bounds in Eq.~\eqref{eq:eddingtonFactorBounds}.  
The dependence of the Eddington factor on $h$ for the Minerbo closure differs from the other closures (i.e., increases vs decreases with increasing $h$), and exceeds $\chi_{\mbox{\tiny max}}$ for $h\gtrsim0.34$.  
When $\cJ=0.99$, the Eddington factor of the CB and BL closures (indistinguishable) and the Kershaw closure remain within the bounds given in Eq.~\eqref{eq:eddingtonFactorBounds}.  
The Eddington factor of the Minerbo closure is nearly flat, and exceeds $\chi_{\mbox{\tiny max}}$ for $h\gtrsim0.006$.  

\begin{figure}[H]
  \centering
  \begin{tabular}{cc}
    \includegraphics[width=0.5\textwidth]{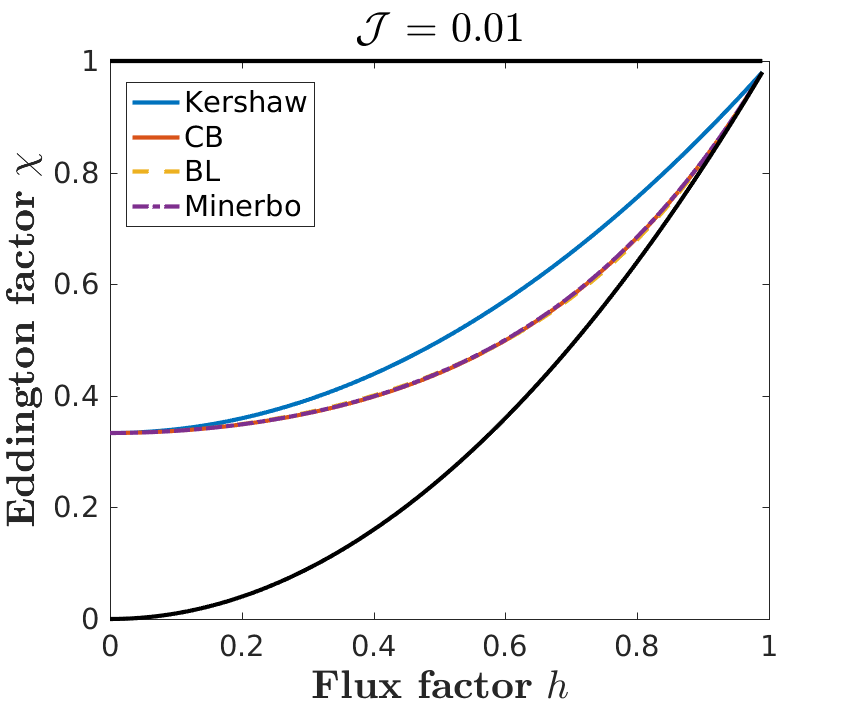}
    \includegraphics[width=0.5\textwidth]{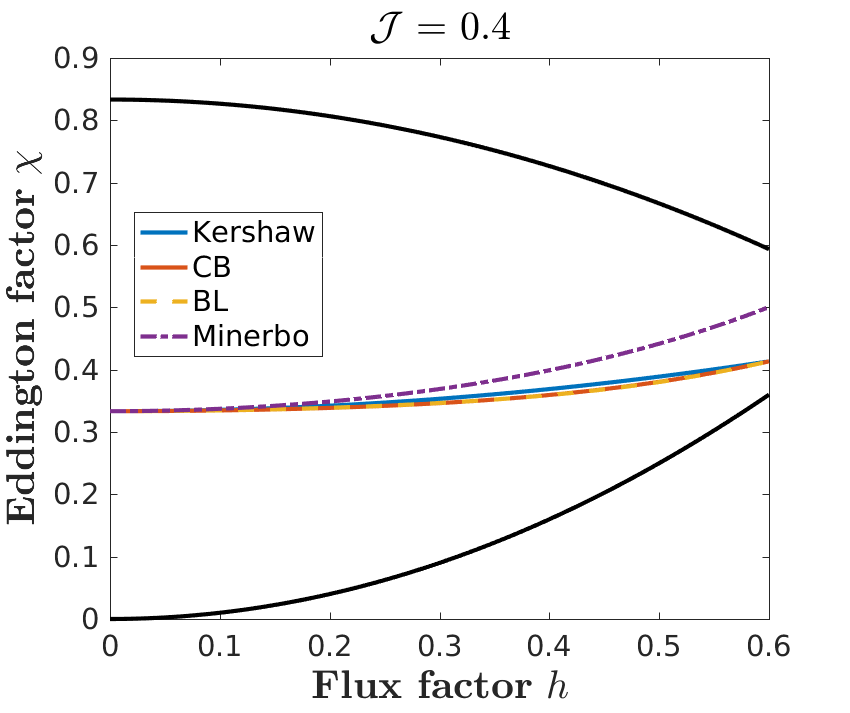} \\
    \includegraphics[width=0.5\textwidth]{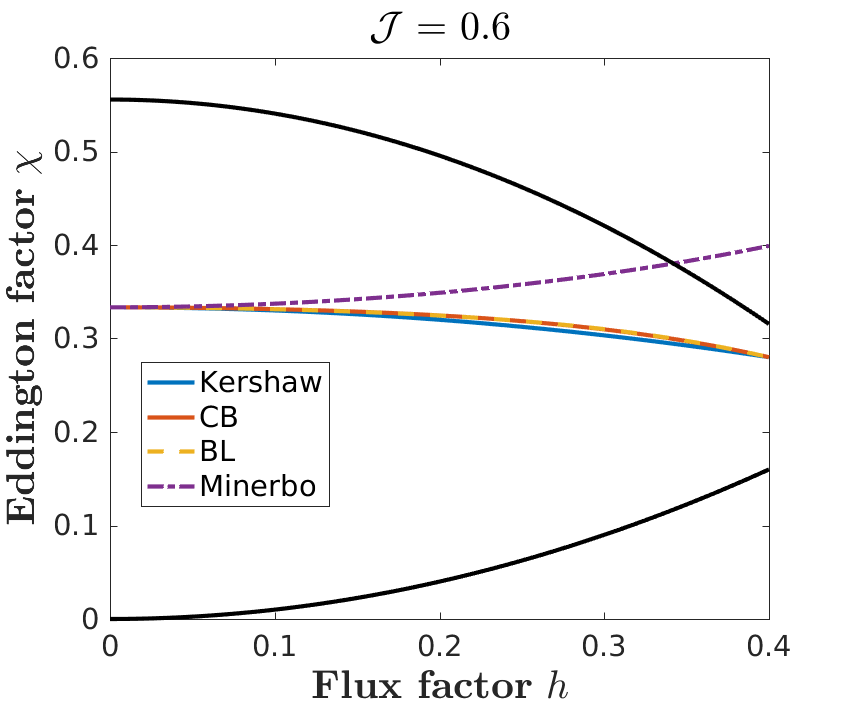}
    \includegraphics[width=0.5\textwidth]{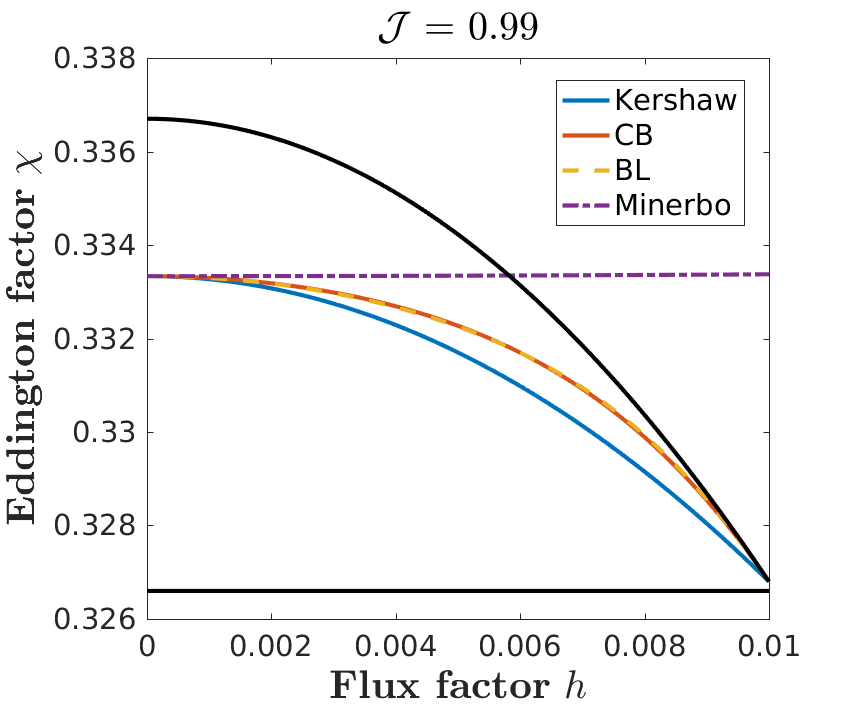}
  \end{tabular}
   \caption{Plot of Eddington factors $\chi$ versus flux factor $h$ for different values of $\cJ$ for various algebraic closures: $\cJ=0.01$ (upper left panel), $\cJ=0.4$ (upper right panel), $\cJ=0.6$ (lower left panel), and $\cJ=0.99$ (lower right panel).  In each panel we plot the Eddington factors of Kershaw (solid blue lines), Cernohorsky \& Bludman (CB, solid red lines), Banach \& Larecki (BL, dashed orange lines), and Minerbo (dash-dot purple lines).  We also plot $\chi_{\mbox{\tiny min}}$ and $\chi_{\mbox{\tiny max}}$ defined in Eq.~\eqref{eq:eddingtonFactorBounds} (lower and upper solid black lines, respectively).}
  \label{fig:EddingtonFactorsWithDifferentClosure}
\end{figure}
We have also checked numerically that for all the algebraic closures based on Fermi-Dirac statistics (CB, BL, and Kershaw), the bounds on the Eddington factor in Eq.~\eqref{eq:eddingtonFactorBounds} holds for all $\vect{\cM}\in\cR$.  
Thus, we conclude that these closures are suited for development of realizability-preserving numerical methods for the two-moment model of fermion transport.  

In Figure~\ref{fig:MabWithDifferentClosure}, we further illustrate properties of the algebraic closures by plotting $\vect{\cM}_{ab}$ as defined in Lemma~\ref{lem:explicitStep} for the maximum entropy closures of CB and Minerbo.  
In both panels, we plot $\vect{\cM}_{ab}$ constructed from randomly selected pairs $\vect{\cM}_{a},\vect{\cM}_{b}\in\cR$ (each blue dot represents one realization of $\vect{\cM}_{ab}$).  
Results for the maximum entropy closure of CB are plotted in the left panel, while results for the Minerbo closure are plotted in the right panel.  
As expected for the closure consistent with moments of Fermi-Dirac distributions (CB), we find $\vect{\cM}_{ab}\in\cR$.  
For the Minerbo closure, which is consistent with positive distributions, $\vect{\cM}_{ab}$ is not confined to $\cR$.  
\begin{figure}[H]
  \centering
  \begin{tabular}{cc}
    \includegraphics[width=0.5\textwidth]{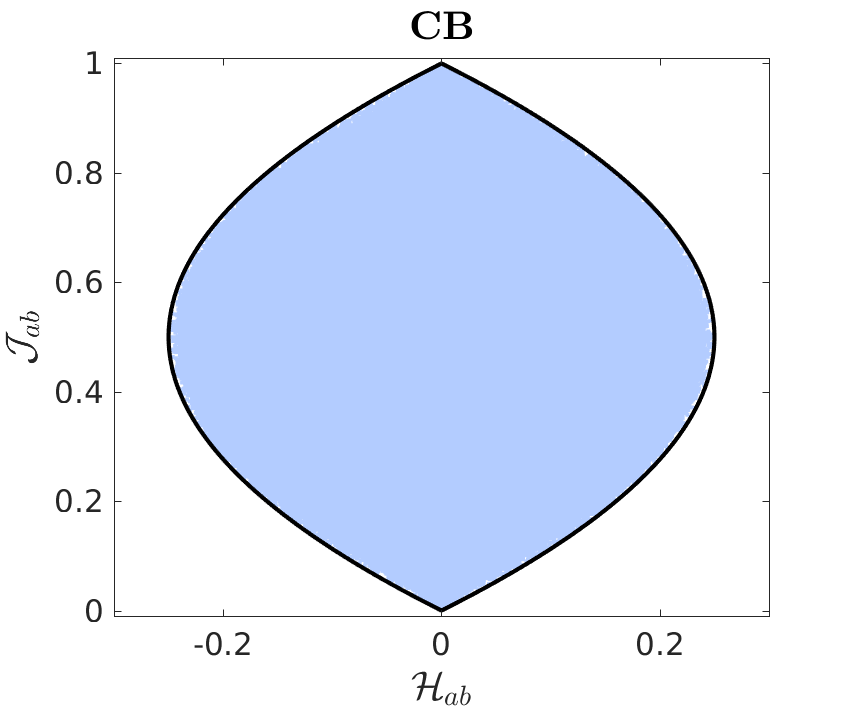}
    \includegraphics[width=0.5\textwidth]{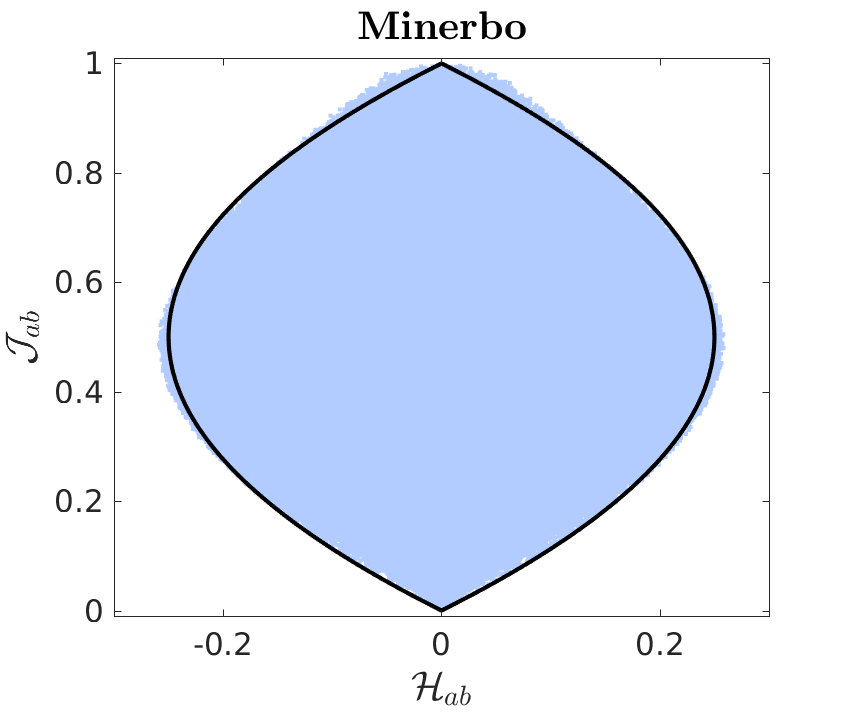}
  \end{tabular}
   \caption{Illustration of $\vect{\cM}_{ab}$, as defined in Lemma~\ref{lem:explicitStep}, computed with algebraic maximum entropy closures of Cernohorsky \& Bludman (left) and Minerbo (right).  In each panel, $\vect{\cM}_{ab}$ was computed using the respective closure, using $10^{6}$ random pairs ($\vect{\cM}_{a},\vect{\cM}_{b}\in\cR$), and plotted as a light-blue point.  The solid black lines mark the boundary of $\cR$: $\gamma(\vect{\cM}) = 0$.}
  \label{fig:MabWithDifferentClosure}
\end{figure}
\section{Discontinuous Galerkin Method}
\label{sec:dg}

Here we briefly outline the DG method for the moment equations.  
(See, e.g., \cite{cockburnShu_2001}, for a comprehensive review on the application of DG methods to solve hyperbolic conservation laws.)  
Since we do not include any physics that couples the energy dimension, the particle energy $\epsilonNu$ is simply treated as a parameter.  
For notational convenience, we will suppress explicit energy dependence of the moments.  
Employing Cartesian coordinates, we write the moment equations in $d$ spatial dimensions as
\begin{equation}
  \pd{\vect{\cM}}{t}+\sum_{i=1}^{d}\pderiv{}{x^{i}}\big(\,\vect{\cF}^{i}(\vect{\cM})\,\big)
  =\frac{1}{\tau}\,\vect{\cC}(\vect{\cM}),
  \label{eq:angularMomentsCartesian}
\end{equation}
where $x^{i}$ is the coordinate along the $i$th coordinate dimension.  
We divide the spatial domain $D$ into a disjoint union $\mathscr{T}$ of open elements $\bK$, so that $D = \cup_{\bK \in \mathscr{T}}\bK$.  
We require that each element is a $d$-dimensional box in the logical coordinates; i.e.,
\begin{equation}
  \bK=\{\,\vect{x} : x^{i} \in K^{i} := (\xL^{i},\xH^{i}),~|~i=1,\ldots,d\,\}, 
\end{equation}
with surface elements denoted $\tilde{\bK}^{i}=\times_{j\ne i}K^{j}$.  
We let $|\bK|$ denote the volume of an element
\begin{equation}
  |\bK| = \int_{\bK}d\vect{x}, \quad\text{where}\quad d\vect{x} = \prod_{i=1}^{d}dx^{i}.  
\end{equation}
We also define $\tilde{\vect{x}}^{i}$ as the coordinates orthogonal to the $i$th dimension, so that as a set $\vect{x}=\{\tilde{\vect{x}}^{i},x^{i}\}$.  
The width of an element in the $i$th dimension is $|K^{i}|=\xH^{i}-\xL^{i}$.  

We let the approximation space for the DG method, $\mathbb{V}^{k}$, be constructed from the tensor product of one-dimensional polynomials of maximal degree $k$.  
Note that functions in $\mathbb{V}^{k}$ can be discontinuous across element interfaces.  
The semi-discrete DG problem is to find $\vect{\cM}_{h}\in\mathbb{V}^{k}$ (which approximates $\vect{\cM}$ in Eq.~\eqref{eq:angularMomentsCartesian}) such that
\begin{align}
  &\pd{}{t}\int_{\bK}\vect{\cM}_{h}\,v\,d\vect{x}
  +\sum_{i=1}^{d}\int_{\tilde{\bK}^{i}}
  \big(\,
    \widehat{\bcF}^{i}(\vect{\cM}_{h})\,v\big|_{\xH^{i}}
    -\widehat{\bcF}^{i}(\vect{\cM}_{h})\,v\big|_{\xL^{i}}
  \,\big)\,d\tilde{\bx}^{i} \nonumber \\
  &\hspace{24pt}
  -\sum_{i=1}^{d}\int_{\bK}\bcF^{i}(\vect{\cM}_{h})\,\pderiv{v}{x^{i}}\,d\vect{x}
  =\f{1}{\tau}\int_{\bK}\bcC(\vect{\cM}_{h})\,v\,d\vect{x},
  \label{eq:semidiscreteDG}
\end{align}
for all $v\in\mathbb{V}^{k}$ and all $\bK\in\mathscr{T}$.  

In Eq.~\eqref{eq:semidiscreteDG}, $\widehat{\bcF}^{i}(\vect{\cM}_{h})$ is a numerical flux, approximating the flux on the surface of $\bK$ with unit normal along the $i$th coordinate direction.  
It is evaluated with a flux function $\vect{\mathscr{F}}^{i}$ using the DG approximation from both sides of the element interface; i.e.,
\begin{equation}
  \widehat{\bcF}^{i}(\vect{\cM}_{h})\big|_{x^{i}}=\vect{\mathscr{F}}^{i}(\vect{\cM}_{h}(x^{i,-},\tilde{\bx}^{i}),\vect{\cM}_{h}(x^{i,+},\tilde{\bx}^{i})),
\end{equation}
where superscripts $-/+$ in the arguments of $\vect{\cM}_{h}$ indicate that the function is evaluated to the immediate left/right of $x^{i}$.  
In this paper we use the simple Lax-Friedrichs (LF) flux given by
\begin{equation}
  \vect{\mathscr{F}}_{\mbox{\tiny LF}}^{i}(\vect{\cM}_{a},\vect{\cM}_{b})
  =\f{1}{2}\,\big(\,\bcF^{i}(\vect{\cM}_{a})+\bcF^{i}(\vect{\cM}_{b})-\alpha^{i}\,(\,\vect{\cM}_{b}-\vect{\cM}_{a}\,)\,\big),
  \label{eq:fluxFunctionLF}
\end{equation}
where $\alpha^{i}$ is the largest eigenvalue (in absolute value) of the flux Jacobian $\partial\bcF^{i}/\partial\vect{\cM}$.  
For particles propagating at the speed of light, we can simply take $\alpha^{i}=1$ (i.e., the global LF flux).  

\begin{rem}
For simplicity, in Eq.~\eqref{eq:semidiscreteDG}, we have approximated the opacities $\sigma_{\Ab}$ and $\sigma_{\Scatt}$ (and thus $\xi$ and $\tau$) on the right-hand side of Eq.~\eqref{eq:angularMomentsCartesian} with constants in each element; i.e., $\sigma_{\Ab},\sigma_{\Scatt}\in\bbV^{0}$.  
\end{rem}
\section{Convex-Invariant IMEX Schemes}
\label{sec:imex}

In this section we discuss the class of IMEX schemes that are used for the realizability-preserving DG-IMEX method developed in Section~\ref{sec:realizableDGIMEX} (see also \ref{app:butcherTables} for additional details).  
The semi-discretization of the moment equations with the DG method given by Eq.~\eqref{eq:semidiscreteDG} results in a system of ordinary differential equations (ODEs) of the form
\begin{equation}
  \dot{\vect{u}}
  =\vect{\cT}(\vect{u})+\f{1}{\tau}\,\vect{\cQ}(\vect{u}),
  \label{eq:ode}
\end{equation}
where $\vect{u}=\{\vect{u}_{\bK}\}_{\bK\in\mathscr{T}}$ are the degrees of freedom evolved with the DG method; i.e., for a test space spanned by $\{\phi_{i}(\vect{x})\}_{i=1}^{N}\in\bbV^{k}$, we let
\begin{equation}
  \vect{u}_{\bK}=\f{1}{|\bK|}\Big(\,\int_{\bK}\vect{\cM}_{h}\,\phi_{1}\,d\vect{x},\int_{\bK}\vect{\cM}_{h}\,\phi_{2}\,d\vect{x},\ldots,\int_{\bK}\vect{\cM}_{h}\,\phi_{N}\,d\vect{x}\,\Big)^{T}.
\end{equation}
Thus, for $\phi_{1}=1$, the first components of $\vect{u}_{\bK}$ are the cell averaged moments.  
In Eq.~\eqref{eq:ode}, the transport operator $\vect{\cT}$ corresponds to the second and third term on the left-hand side of Eq.~\eqref{eq:semidiscreteDG}, while the collision operator $\vect{\cQ}$ corresponds to the right-hand side of Eq.~\eqref{eq:semidiscreteDG}.  

Eq.~\eqref{eq:ode} is to be integrated forward in time with an ODE integrator.  
Since the realizable set $\cR$ is convex, convex-invariant schemes can be used to design realizability-preserving schemes for the two-moment model.  
\begin{define}
  Given sufficient conditions, a convex-invariant time integration scheme preserves the constraints of a model if the set of admissible states satisfying the constraints forms a convex set.  
  \label{def:constraintPreserving}
\end{define}
As an example, high-order explicit strong stability-preserving Runge-Kutta (SSP-RK) methods form a class of convex-invariant schemes.  

\subsection{Second-Order Accurate, Convex-Invariant IMEX Schemes}

In many applications, collisions with a background induce stiffness ($\tau\ll1$) in regions of the computational domain that must be treated with implicit methods.  
Meanwhile, the time scales induced by the transport term can be treated with explicit methods.  
This motivates the use of IMEX methods \cite{ascher_etal_1997,pareschiRusso_2005}.  
Our goal is to employ IMEX schemes that preserve realizability of the moments, subject only to a time step governed by the explicit transport operator and comparable to the time step required for numerical stability of the explicit scheme.  
We seek to achieve this goal with convex-invariant IMEX schemes.  

Unfortunately, high-order (second or higher order temporal accuracy) convex-invariant IMEX methods with time step restrictions solely due to the transport operator do not exist (see for example Proposition~6.2 in \cite{gottlieb_etal_2001}, which rules out the existence of implicit SSP-RK methods of order higher than one).  
To overcome this barrier, Chertock et al. \cite{chertock_etal_2015} presented IMEX schemes with a correction step.  
These schemes are SSP but only first-order accurate within the standard IMEX framework.  
The correction step is introduced to recover second-order accuracy.  
However, the correction step in \cite{chertock_etal_2015} involves both the transport and collision operators, and we have found that, when applied to the fermionic two-moment model, realizability is subject to a time step restriction that depends on $\tau$ in a way that becomes too restrictive for stiff problems.  
More recently, Hu et al. \cite{hu_etal_2018}, presented similar IMEX schemes for problems involving BGK-type collision operators, but with a correction step that does not include the transport operator.  
In this case, the scheme is convex-invariant, subject only to time step restrictions stemming from the transport operator, which is more attractive for our target application.  
These second-order accurate, $s$-stage IMEX schemes take the following form \cite{hu_etal_2018}
\begin{align}
  \vect{u}^{(i)}
  &=\vect{u}^{n}
  +\dt\sum_{j=1}^{i-1}\tilde{a}_{ij}\,\vect{\cT}(\vect{u}^{(j)})
  +\dt\sum_{j=1}^{i}a_{ij}\,\f{1}{\tau}\,\vect{\cQ}(\vect{u}^{(j)}),
  \quad i=1,\ldots,s, \label{imexStages} \\
  \tilde{\vect{u}}^{n+1}
  &=\vect{u}^{n}
  +\dt\sum_{i=1}^{s}\tilde{w}_{i}\,\vect{\cT}(\vect{u}^{(i)})
  +\dt\sum_{i=1}^{s}w_{i}\,\f{1}{\tau}\,\vect{\cQ}(\vect{u}^{(i)}), \label{imexIntermediate} \\
  \vect{u}^{n+1}
  &=\tilde{\vect{u}}^{n+1}-\alpha\,\dt^{2}\,\f{1}{\tau^{2}}\,\vect{\cQ}'(\vect{u}^{*})\,\vect{\cQ}(\vect{u}^{n+1}), \label{eq:imexCorrection}
\end{align}
where, as in standard IMEX schemes, $(\tilde{a}_{ij})$ and $(a_{ij})$, components of $s\times s$ matrices $\tilde{A}$ and $A$, respectively, and the vectors $\tilde{\vect{w}}=(\tilde{w}_{1},\ldots,\tilde{w}_{s})^{T}$ and $\vect{w}=(w_{1},\ldots,w_{s})^{T}$ must satisfy certain order conditions \cite{pareschiRusso_2005}.  
The coefficient $\alpha$ in the correction step is positive, and $\vect{\cQ}'$ is the Fr{\'e}chet derivative of the collision term evaluated at $\vect{u}^{*}$.  
For second-order accuracy, $\vect{\cQ}'$ can be evaluated using any of the stage values ($\vect{u}^{n}$, $\vect{u}^{(i)}$, or $\tilde{\vect{u}}^{n+1}$).  
For second-order temporal accuracy, the order conditions for the IMEX scheme in Eqs.~\eqref{imexStages}-\eqref{eq:imexCorrection} are \cite{hu_etal_2018}
\begin{equation}
  \sum_{i=1}^{s}\tilde{w}_{i}=\sum_{i=1}^{s}w_{i}=1,
  \label{orderConditions1}
\end{equation}
and
\begin{equation}
  \sum_{i=1}^{s}\tilde{w}_{i}\,\tilde{c}_{i}
  =\sum_{i=1}^{s}\tilde{w}_{i}\,c_{i}
  =\sum_{i=1}^{s}w_{i}\,\tilde{c}_{i}
  =\sum_{i=1}^{s}w_{i}\,c_{i}-\alpha=\f{1}{2}, 
  \label{orderConditions2}
\end{equation}
where $\tilde{c}_{i}$ and $c_{i}$ are given in \ref{app:butcherTables}.  
For globally stiffly accurate (GSA) IMEX schemes, $\tilde{w}_{i}=\tilde{a}_{si}$ and $w_{i}=a_{si}$ for $i=1,\ldots,s$, so that $\tilde{\vect{u}}^{n+1}=\vect{u}^{(s)}$ \cite{ascher_etal_1997}.  
This property is beneficial for very stiff problems and also simplifies the proof of the realizability-preserving property of the DG-IMEX scheme given in Section~\ref{sec:realizableDGIMEX}, since it eliminates the assembly step in Eq.~\eqref{imexIntermediate}.  

Hu et al. \cite{hu_etal_2018} rewrite the stage values in Eq.~\eqref{imexStages} in the following form
\begin{equation}
  \vect{u}^{(i)}
  =\sum_{j=0}^{i-1}c_{ij}\Big[\,\vect{u}^{(j)}+\hat{c}_{ij}\,\dt\,\vect{\cT}(\vect{u}^{(j)})\,\Big]
  +a_{ii}\,\dt\,\f{1}{\tau}\,\vect{\cQ}(\vect{u}^{(i)}),\quad i=1,\ldots,s,
  \label{eq:imexStagesRewrite}
\end{equation}
where $c_{ij}$, and $\hat{c}_{ij}=\tilde{c}_{ij}/c_{ij}$ are computed from $\tilde{a}_{ij}$ and $a_{ij}$ (see \ref{app:butcherTables}).  
In Eq.~\eqref{eq:imexStagesRewrite}, $\vect{u}^{(0)}=\vect{u}^{n}$.  
Two types of IMEX schemes are considered: \emph{type~A} \cite{pareschiRusso_2005,dimarcoPareschi2013} and \emph{type~ARS} \cite{ascher_etal_1997}.  
For IMEX schemes of type~A, the matrix $A$ is invertible.  
For IMEX schemes of type~ARS, the matrix $A$ can be written as
\begin{equation*}
  \left( 
    \begin{matrix} 
       0 & 0 \\ 
       0 & \hat{A}
    \end{matrix}
  \right),
\end{equation*}
where $\hat{A}$ is invertible.  
In writing the stages in the IMEX scheme in the general form given by Eq.~\eqref{eq:imexStagesRewrite}, it should be noted that $\tilde{c}_{i0}=0$ for IMEX schemes of Type~A, and $c_{i1}=\tilde{c}_{i1}=0$ for IMEX schemes of Type~ARS \cite{hu_etal_2018}.  
Type A schemes can be made convex-invariant by requiring 
\begin{align}
  &a_{ii}>0, \quad c_{i0}\ge0, \quad \text{for} \quad i=1,\ldots,s, \nonumber \\
  &\text{and} \quad c_{ij},\tilde{c}_{ij}\ge0, \quad \text{for} \quad i=2,\ldots,s, \quad\text{and}\quad j=1,\ldots,s-1.  
  \label{eq:positivityConditionsTypeA}
\end{align}
Similarly, type ARS schemes can be made convex-invariant by requiring 
\begin{align}
  &a_{ii}>0, \quad c_{i0},\tilde{c}_{i0}\ge0, \quad \text{for} \quad i=2,\ldots,s, \nonumber \\
  &\text{and} \quad c_{ij},\tilde{c}_{ij}\ge0, \quad \text{for} \quad i=3,\ldots,s, \quad\text{and}\quad j=2,\ldots,i-1.  
  \label{eq:positivityConditionsTypeARS}
\end{align}
Coefficients were given in \cite{hu_etal_2018} for GSA schemes of type A with $s=3$ and type ARS with $s=4$.  
(It was also proven that $s=3$ and $s=4$ are the necessary number of stages needed for GSA second-order convex-invariant IMEX schemes of type A and type ARS, respectively.)  

In Eq.~\eqref{eq:imexStagesRewrite}, the explicit part of the IMEX scheme has been written in the so-called Shu-Osher form \cite{shuOsher_1988}.  
For the scheme to be convex-invariant, the coefficients $c_{ij}$ must also satisfy $\sum_{j=0}^{i-1}c_{ij}=1$.  
Then, if the expression inside the square brackets in Eq.~\eqref{eq:imexStagesRewrite} --- which is in the form of a forward Euler update with time step $\hat{c}_{ij}\,\dt$ --- is in the (convex) set of admissible states for all $i=1,\ldots,s$, and $j=0,\ldots,i-1$, it follows from convexity arguments that the entire sum on the right-hand side of Eq.~\eqref{eq:imexStagesRewrite} is also admissible.  
Thus, if the explicit update with the transport operator is admissible for a time step $\dt_{\mbox{\tiny Ex}}$, the IMEX scheme is convex-invariant for a time step $\dt\le c_{\Sch}\,\dt_{\mbox{\tiny Ex}}$, where
\begin{equation}
  c_{\Sch}=\min_{\substack{i = 2,\ldots,s \\ 
       j = 1,\ldots,i-1}}\,\f{1}{\hat{c}_{ij}} \quad \text{(Type A)}, \quad  c_{\Sch}=\min_{\substack{i = 2,\ldots,s \\ 
              j = 0,2,\ldots,i-1}}\,\f{1}{\hat{c}_{ij}} \quad \text{(Type ARS)}.
  \label{eq:imexCFL}
\end{equation}
Here, $c_{\Sch}$ is the CFL condition, relative to $\dt_{\mbox{\tiny Ex}}$, for the IMEX scheme to be convex-invariant.  
It is desirable to make $c_{\Sch}$ as large (close to $1$) as possible.  
Note that for $\vect{u}^{(i)}$ to be admissible also requires the implicit solve (equivalent to implicit Euler) to be convex-invariant.  
In \cite{hu_etal_2018}, Hu et al. provide examples of GSA, convex-invariant IMEX schemes of type A (see scheme PA2 in \ref{app:butcherTables}) and type ARS.  
In \ref{app:butcherTables}, we provide another example of a GSA, convex-invariant IMEX scheme of type A (scheme PA2+), with a larger $c_{\Sch}$ (a factor of about $1.7$ larger).  

\subsection{Convex-Invariant, Diffusion Accurate IMEX Schemes}

Unfortunately, the correction step in Eq.~\eqref{eq:imexCorrection} deteriorates the accuracy of the IMEX scheme when applied to the moment equations in the diffusion limit.  
The diffusion limit is characterized by frequent collisions ($\tau\ll 1$) in a purely scattering medium ($\xi=0$), and exhibits long-time behavior governed by (e.g., \cite{jinLevermore_1996})
\begin{equation}
  \pd{\cJ}{t} + \nabla\cdot\vect{\cH} = 0
  \quad\text{and}\quad
  \vect{\cH} = - \tau\,\nabla\cdot\vect{\cK}.  
  \label{eq:diffusionLimit}
\end{equation}
Here the time derivative term in the equation for the particle flux has been dropped (formally $\cO(\tau^{2})$) so that, to leading order in $\tau$, the second equation in Eq.~\eqref{eq:diffusionLimit} states a balance between the transport term and the collision term.  
Furthermore, in the diffusion limit, the distribution function is nearly isotropic so that $\vect{\cK}\approx\f{1}{3}\,\cJ\,\vect{I}$ and $\vect{\cH}\approx-\f{1}{3}\,\tau\,\nabla\cJ$.  
The absence of the transport operator in the correction step in Eq.~\eqref{eq:imexCorrection}, destroys the balance between the transport term and the collision term.  
We demonstrate the inferior performance of IMEX schemes with this correction step in the diffusion limit in Section~\ref{sec:smoothProblems}.  
We have also implemented and tested one of the IMEX schemes in Chertock et al. \cite{chertock_etal_2015} (not included in Section~\ref{sec:smoothProblems}), where the transport operator is part of the correction step, and found it to perform very well in the diffusion limit.  
However, we have not been able to prove the realizability-preserving property with this approach without invoking a too severe time step restriction.  
We therefore proceed to design convex-invariant IMEX schemes without the correction step that perform well in the diffusion limit.  
We limit the scope to IMEX schemes of type~ARS.  
(It can be shown that IMEX schemes of type~A conforming to Definition~\ref{def:PD-IMEX} below do not exist; cf. \ref{app:noTypeA}.)

We take a heuristic approach to determine conditions on the coefficients of the IMEX scheme to ensure it performs well in the diffusion limit.  
We leave the spatial discretization unspecified.  
Define the vectors 
\begin{equation}
  \vec{\cJ}=(\cJ^{(1)},\ldots,\cJ^{(s)})^{T}
  \quad\text{and}\quad
  \vec{\vect{\cH}}=(\vect{\cH}^{(1)},\ldots,\vect{\cH}^{(s)})^{T}.  
\end{equation}
(The components of $\vec{\cJ}$ and $\vec{\vect{\cH}}$ can, e.g., be the cell averages evolved with the DG method.)  
We can then write the stages of the IMEX scheme in Eq.~\eqref{imexStages} applied to the particle density equation as
\begin{equation}
  \vec{\cJ}=\cJ^{n}\,\vec{\vect{e}} - \dt\,\tilde{A}\,\nabla\cdot\vec{\vect{\cH}}, 
  \label{eq:numberDensityIMEX}
\end{equation}
where $\vec{\vect{e}}$ is a vector of length $s$ containing all ones, and the divergence operator acts individually on the components of $\vec{\vect{\cH}}$.  
Similarly, for the particle flux equation we have
\begin{equation}
  \vec{\vect{\cH}}=\vect{\cH}^{n}\,\vec{\vect{e}}-\dt\,\big(\,\f{1}{3}\tilde{A}\,\nabla\vec{\cJ}+\f{1}{\tau}\,A\,\vec{\vect{\cH}}\,\big).
  \label{eq:numberFluxIMEX}
\end{equation}
In the context of IMEX schemes, the diffusion limit (cf. the second equation in \eqref{eq:diffusionLimit}) implies that the relation $A\,\vec{\vect{\cH}}=-\f{1}{3}\,\tau\,\tilde{A}\,\nabla\vec{\cJ}$ should hold.  
Define the pseudoinverse of the implicit coefficient matrix for IMEX schemes of type~ARS as
\begin{equation*}
    A^{-1}
    =\left(
        \begin{matrix}
          0 & 0 \\ 
          0 & \hat{A}^{-1}
        \end{matrix}
      \right).
\end{equation*}
Then, for the stages $i=1,\ldots,s$, 
\begin{equation}
  \vect{\cH}^{(i)}=-\f{1}{3}\,\tau\,\vec{\vect{e}}_{i}^{T}A^{-1}\tilde{A}\,\vec{\vect{e}}\,\nabla\cJ^{n}+\cO(\dt\,\tau^{2}),
\end{equation}
where $\vec{\vect{e}}_{i}$ is the $i$th column of the $s\times s$ identity matrix and we have introduced the expansion $\vec{\cJ}=\cJ^{n}\vec{\vect{e}}+\cO(\dt\,\tau)$.  
For $\vect{\cH}^{(i)}$ to be accurate in the diffusion limit, we require that
\begin{equation}
  \vect{e}_{i}^{T}A^{-1}\tilde{A}\,\vect{e} = 1, \quad i=2,\ldots,s.
  \label{eq:diffusionCondition}
\end{equation}
(The case $i=1$ is trivial and does not place any constraints on the components of $A$ and $\tilde{A}$.)  
If Eq.~\eqref{eq:diffusionCondition} holds, then
\begin{equation}
  \f{\cJ^{(s)}-\cJ^{n}}{\dt}
  =-\tilde{\vect{w}}^{T}(\nabla\cdot\vec{\vect{\cH}})
  =\f{1}{3}\tau\,\nabla^{2}\cJ^{n}+\cO(\dt\,\tau^{2}), 
\end{equation}
which approximates a diffusion equation for $\cJ$ with the correct diffusion coefficient $\tau/3$.  
(For a GSA IMEX scheme without the correction step, $\cJ^{n+1}=\cJ^{(s)}$.)  

Unfortunately, the ``diffusion limit requirement" in Eq.~\eqref{eq:diffusionCondition}, together with the order conditions given by Eqs.~\eqref{orderConditions1} and \eqref{orderConditions2} (with $\alpha=0$), and the positivity conditions on $c_{ij}$ and $\tilde{c}_{ij}$, result in too many constraints to obtain a second-order accurate convex-invariant IMEX scheme.%
\footnote{It was shown in \cite{hu_etal_2018} --- without the diffusion limit requirement --- that the minimum number of stages for convex-invariant IMEX schemes of type~ARS is four.)  
We are also concerned about increasing the number of stages, and thereby the number of implicit solves, since the implicit solve will dominate the computational cost of the IMEX scheme with more realistic collision operators (e.g., inelastic scattering).}  
To reduce the number of constraints and accommodate accuracy in the diffusion limit, we relax the requirement of overall second-order accuracy of the IMEX scheme.  
Instead, we only require the scheme to be second-order accurate in the streaming limit ($\vect{\cQ}=0$).  
This gives the order conditions
\begin{equation}
  \sum_{i=1}^{s}\tilde{w}_{i}=1
  \quad\text{and}\quad
  \sum_{i=1}^{s}\tilde{w}_{i}\,\tilde{c}_{i}=\f{1}{2},
  \label{eq:orderConditionsEx}
\end{equation}
where the first condition (consistency condition) is required for first-order accuracy.  
We then seek to design IMEX schemes of Type~ARS conforming to the following working definition
\begin{define}
  Let PD-IMEX be an IMEX scheme satisfying the following properties:
  \begin{enumerate}
    \item Consistency of the implicit coefficients
    \begin{equation}
      \sum_{i=1}^{s}w_{i}=1.
      \label{eq:implicitConsistency}
    \end{equation}
    \item Second-order accuracy in the streaming limit; i.e., satisfies Eq.~\eqref{eq:orderConditionsEx}.
    \item Convex-invariant; i.e. satisfies Eq.~\eqref{eq:positivityConditionsTypeARS}, with $\sum_{j=0}^{i-1}c_{ij}=1$, for $i=1,\ldots,s$, and $c_{\Sch}>0$.
    \item Well-behaved in the diffusion limit; i.e., satisfies Eq.~\eqref{eq:diffusionCondition}.
    \item Less than four stages ($s\le3$).
    \item Globally stiffly accurate (GSA): $a_{si}=w_{i}$ and $\tilde{a}_{si}=\tilde{w}_{i},\quad i=1,\ldots,s$.
  \end{enumerate}
  \label{def:PD-IMEX}
\end{define}
Fortunately, IMEX schemes of type~ARS satisfying these properties are easy to find, and we provide an example with $s=3$ (two implicit solves) in \ref{app:butcherTables} (scheme PD-ARS; see \ref{app:PD-ARS} for further details).  
In the streaming limit, this scheme is identical to the optimal second-order accurate SSP-RK method \cite{gottlieb_etal_2001}.  
It is also very similar to the scheme given in \cite{mcclarren_etal_2008} (see scheme PC2 in \ref{app:butcherTables}), which is also a GSA IMEX scheme of type ARS with $s=3$.  
Scheme PC2 is second-order in the streaming limit, has been demonstrated to work well in the diffusion limit \cite{mcclarren_etal_2008,radice_etal_2013}, and satisfies the positivity conditions in Eq.~\eqref{eq:positivityConditionsTypeARS}.  
However, $c_{\Sch}=0$ (our primary motivation for finding an alternative).  
In Section~\ref{sec:numerical}, we show numerically that the accuracy of scheme PD-ARS is comparable to the accuracy of scheme PC2.  

\subsection{Absolute Stability}

Here we analyze the absolute stability of the proposed IMEX schemes, PA2+ and PD-ARS (given in \ref{app:butcherTables}), following \cite{hu_etal_2018}.  
As is commonly done, we do this in the context of the linear scalar equation
\begin{equation}
  \dot{u}=\lambda_{1}\,u+\lambda_{2}\,u,
  \label{eq:scalarODE}
\end{equation}
where $\lambda_{1}\in\mathbb{C}$ and $\lambda_{2}\le0$.  
On the right-hand side of Eq.~\eqref{eq:scalarODE}, the first (oscillatory) term is treated explicitly, while the second (damping) term is treated implicitly.  
The IMEX schemes can then be written as $u^{n+1} =P(z_{1},z_{2})\,u^{n}$, where $P(z_{1},z_{2})$ is the amplification factor of the scheme, $z_{1}=\dt\,\lambda_{1} = x +iy$, and $z_{2} = \dt\,\lambda_{2} \le 0$.  
Stability of the IMEX scheme requires $|P(z_1,z_2)|\leq 1$.  
The stability regions of PA2+ and PD-ARS are plotted in Figure~\ref{fig:AbsoluteStability}.  
As can be seen from the figures, the absolute stability region of both schemes increases with increasing $|z_{2}|$ (increased damping).  
For a given $|z_{2}|>0$, the stability region of PD-ARS is larger than that of PA2+.  
In the linear model in Eq.~\eqref{eq:scalarODE}, a time step that satisfies the absolute stability for the explicit part of the IMEX scheme ($|z_2| = 0$) fulfills the stability requirement for the IMEX scheme as whole ($|z_2| \geq 0$).  
This holds for all PD-ARS schemes with $\epsilon\in [0,0.5)$.
\begin{figure}[h]
  \centering
  \begin{tabular}{cc}
    \includegraphics[width=0.5\textwidth]{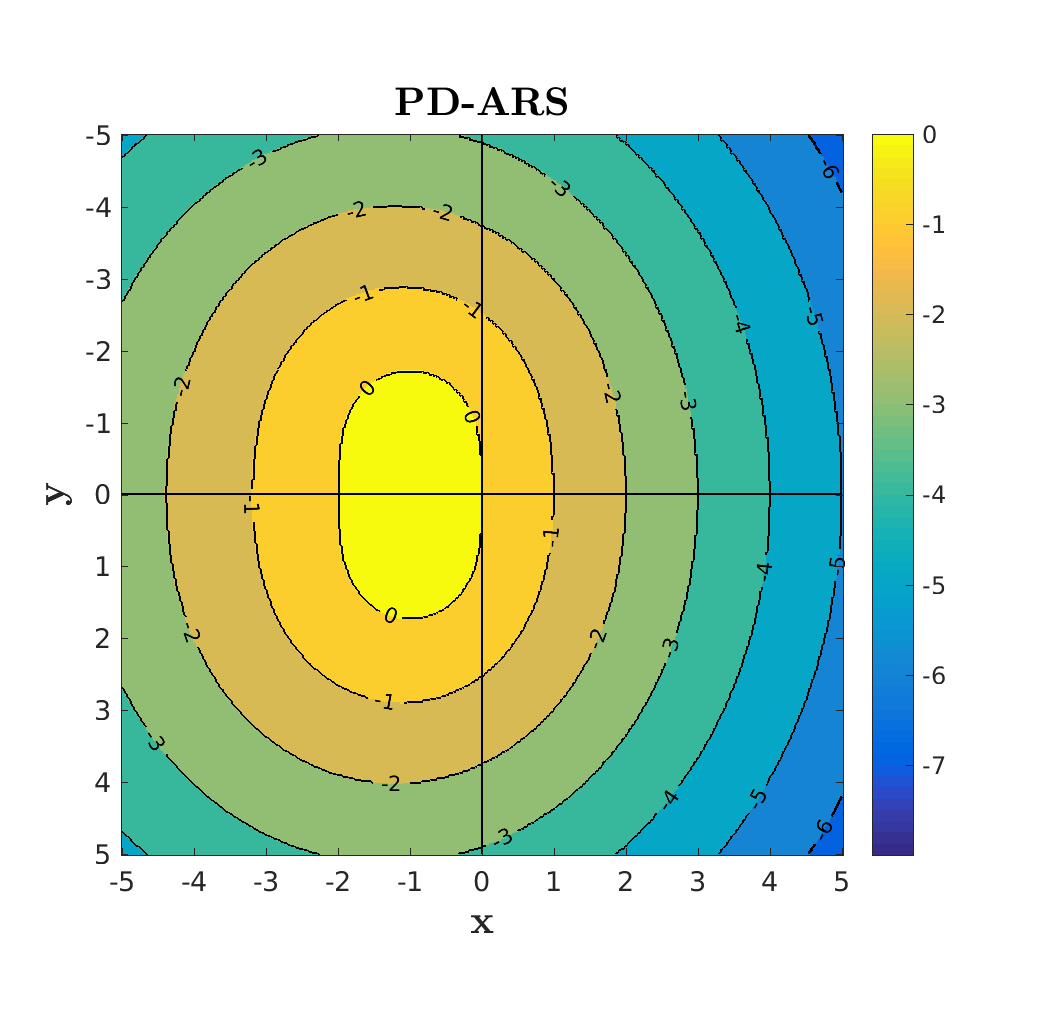}
    \includegraphics[width=0.5\textwidth]{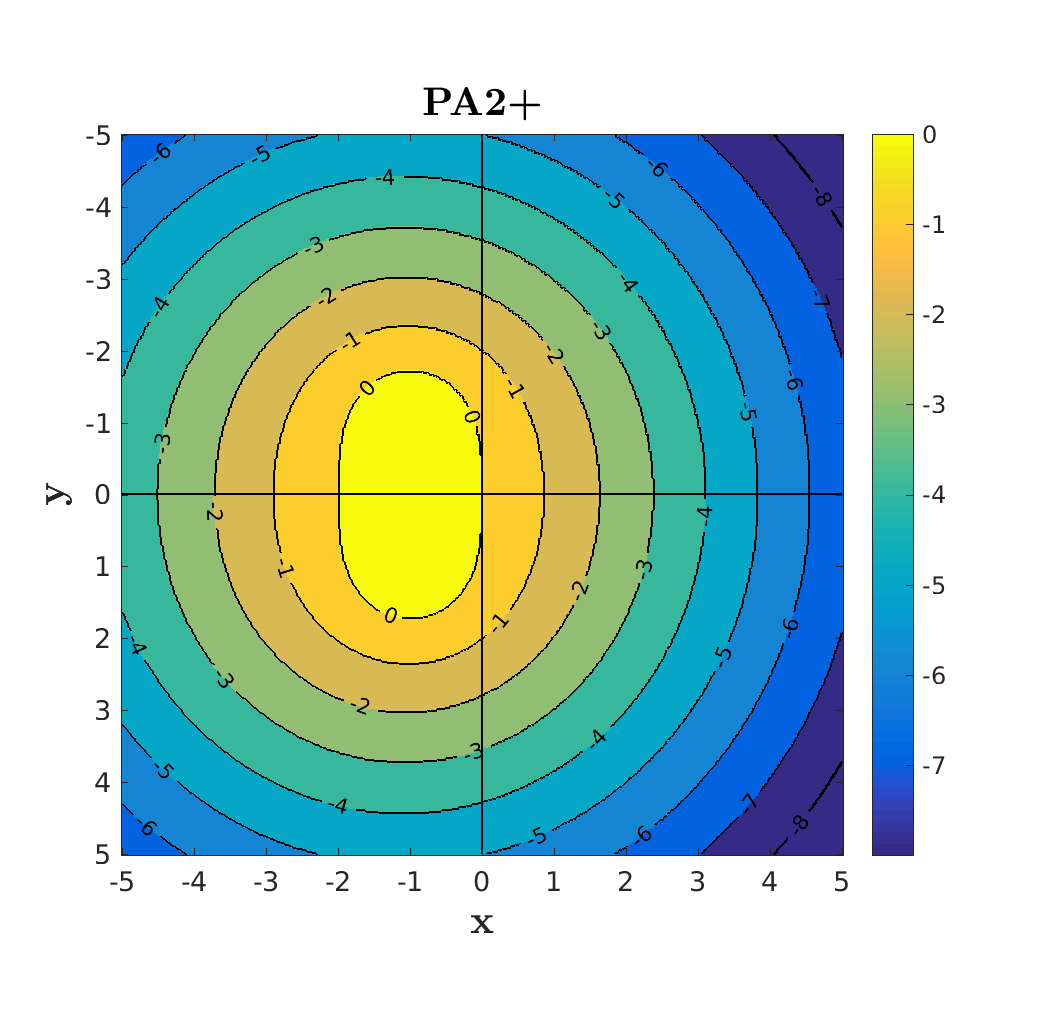}
  \end{tabular}
   \caption{Boundary of the absolute stability region in the $xy$-plane ($x=\operatorname{Re}(z_{1})$, $y=\operatorname{Im}(z_{1})$) for different values of $z_2$ for PD-ARS (with $\epsilon = 0.1$; left panel) and PA2+ (right panel).  Contours of constant $z_{2}$ are included, and the stability region for a given $z_{2}$ is enclosed by the corresponding contour.}
  \label{fig:AbsoluteStability}
\end{figure}
\section{Realizability-Preserving DG-IMEX Scheme}
\label{sec:realizableDGIMEX}

We proceed to develop realizability-preserving DG schemes for the two-moment model based on the IMEX schemes discussed in the previous section (cf. Eqs.~\eqref{imexStages}-\eqref{eq:imexCorrection}).  
Following the framework in \cite{zhangShu_2010b} for high-order DG schemes, the realizability-preserving DG-IMEX scheme is designed to preserve realizability of cell averages over a time step $\dt$ in each element $\bK$.  
Realizability of the cell average is leveraged to limit the polynomial approximation $\vect{\cM}_{h}$ in the stages of the IMEX scheme.  
If the polynomial approximation $\vect{\cM}_{h}$ is not realizable in a finite number of points in $\bK$, the high-order components are damped (see Section~\ref{sec:limiter}).  
The main result of this section is stated in Theorem~\ref{the:realizableDGIMEX}.  
The realizability-preserving property of the DG-IMEX scheme is stated in Theorem~\ref{the:realizableDGIMEX2} in Section~\ref{sec:limiter}, after the discussion of the limiter.  

The cell average of the moments is defined as
\begin{equation}
  \vect{\cM}_{\bK}
  =\f{1}{|\bK|}\int_{\bK}\vect{\cM}_{h}\,d\bx.  
\end{equation}
With $v=1$ in Eq.~\eqref{eq:semidiscreteDG}, the stage values for the cell average in the IMEX scheme (cf. Eq.~\eqref{eq:imexStagesRewrite}) can be written as
\begin{align}
  \cM_{\bK}^{(0)}
  &=\cM_{\bK}^{n}, \nonumber \\
  \vect{\cM}_{\bK}^{(i)}
  &=\sum_{j=0}^{i-1}c_{ij}\,\vect{\cM}_{\bK}^{(ij)}
  +a_{ii}\,\dt\,\f{1}{\tau}\,\big(\,\vect{\eta}-\vect{\cD}\,\vect{\cM}_{\bK}^{(i)}\,\big), \quad i=1,\ldots,s,
  \label{eq:imexStagesCellAverage}
\end{align}
where $c_{ij}\ge0$, $\sum_{j=0}^{i-1}c_{ij}=1$, $a_{ii}>0$, 
\begin{equation}
  \vect{\cM}_{\bK}^{(ij)}
  =\vect{\cM}_{\bK}^{(j)}-\hat{c}_{ij}\,\dt\,\big\langle\,\nabla\cdot\vect{\cF}(\vect{\cM}_{h}^{(j)})\,\big\rangle_{\bK},
\end{equation}
and $\hat{c}_{ij}\ge0$.  
The cell average of the divergence operator is
\begin{equation}
  \big\langle\,\nabla\cdot\vect{\cF}(\vect{\cM}_{h}^{(j)})\,\big\rangle_{\bK}
  =\f{1}{|\bK|}\sum_{k=1}^{d}\int_{\tilde{\bK}^{k}}
  \big(\,\widehat{\bcF}^{k}(\vect{\cM}_{h}^{(j)})\big|_{\xH^{k}}-\widehat{\bcF}^{k}(\vect{\cM}_{h}^{(j)})\big|_{\xL^{k}}\,\big)\,d\tilde{\vect{x}}^{k}.  
\end{equation}

We first establish conditions for realizability of the stage values in Eq.~\eqref{eq:imexStagesCellAverage}.  
\begin{lemma}
  Let $\vect{\cM}_{\bK}^{(i)}$ satisfy Eq.~\eqref{eq:imexStagesCellAverage}.
  Assume that $\vect{\cM}_{\bK}^{(ij)}\in\cR~\forall\,i=1,\ldots,s,\,j\le i-1$.  
  Then, $\vect{\cM}_{\bK}^{(i)}\in\cR$, for $i=1,\ldots,s$.  
  \label{lem:imexStagesCellAverage}
\end{lemma}
\begin{proof}
  For $i=1,\ldots,s$, 
  \begin{equation*}
    \sum_{j=0}^{i-1}c_{ij}\,\vect{\cM}_{\bK}^{(ij)}\in\cR,
  \end{equation*}
  since it is a convex combination of elements in $\bK$.  
  The Result follows from Lemma~\ref{lem:implicitStep}.  
\end{proof}

We next establish conditions under which $\vect{\cM}_{\bK}^{(ij)}\in\cR$.  
\begin{lemma}
  Let $\{\beta_{k}\}_{k=1}^{d}$ be a set of strictly positive constants satisfying $\sum_{k=1}^{d}\beta_{k}=1$.  
  If for each $k\in\{1,\ldots,d\}$, 
  \begin{align}
    &\vect{\Gamma}^{k}\big[\vect{\cM}_{h}^{(j)};\hat{c}_{ij}\big] \label{eq:realizableGamma} \\
    &:=\f{1}{|K^{k}|}
    \Big[\,\int_{K^{k}}\vect{\cM}_{h}^{(j)}\,dx^{k}-\f{\hat{c}_{ij}\,\dt}{\beta_{k}}\big(\,\widehat{\bcF}^{k}(\vect{\cM}_{h}^{(j)})\big|_{\xH^{k}}-\widehat{\bcF}^{k}(\vect{\cM}_{h}^{(j)})\big|_{\xL^{k}}\,\big)\,\Big] \nonumber    
  \end{align}
  is realizable, then $\vect{\cM}_{\bK}^{(ij)}\in\cR$.  
  \label{lem:realizableMij}
\end{lemma}
\begin{proof}
  It is easy to show that $\vect{\cM}_{\bK}^{(ij)}$ can be expressed as the convex combination
  \begin{equation}
    \vect{\cM}_{\bK}^{(ij)}
    =\sum_{k=1}^{d}\beta_{k}\,\f{1}{|\tilde{\vect{\bK}}^{k}|}\int_{\tilde{\bK}^{k}}\vect{\Gamma}^{k}\big[\vect{\cM}_{h}^{(j)};\hat{c}_{ij}\big]\,d\tilde{\vect{x}}^{k}.  
    \label{eq:cellAverageInTermsOfGamma}
  \end{equation}
  The result follows immediately.  
\end{proof}
\begin{rem}
  If a quadrature rule $\tilde{\vect{Q}}^{k}\colon C^{0}(\tilde{\bK}^{k})\to\bbR$, with positive weights, and points defined by the set $\tilde{\vect{S}}^{k}$, is used to approximate the integral over $\tilde{\bK}^{k}$ in Eq.~\eqref{eq:cellAverageInTermsOfGamma}, then it is sufficient for Eq.~\eqref{eq:realizableGamma} to hold in the quadrature points $\tilde{\vect{S}}^{k}\subset\tilde{\bK}^{k}$.  
\end{rem}

Next, we establish conditions for which Eq.~\eqref{eq:realizableGamma} holds.  
To this end, let $\hat{Q}^{k}\colon C^{0}(K^{k})\to\bbR$ denote the $N$-point \emph{Gauss-Lobatto (GL)} quadrature rule on the interval $K^{k}=(\xL^{k},\xH^{k})$, with points
\begin{equation}
  \hat{S}^{k}=\left\{\xL^{k}=\hat{x}_{1}^{k},\ldots,\hat{x}_{N}^{k}=\xH^{k}\right\}, 
  \label{eq:quadraturePointsGL}
\end{equation}
and weights $\hat{w}_{q} \in (0,1]$, normalized so that $\sum_{q=1}^{N} \hat{w}_{q} = 1$.  
(The hat is used to denote the GL rule, which includes the endpoints of the interval $K^{k}$.)
This quadrature integrates polynomials in $x^{k} \in \bbR$ with degree $\le2N-3$ exactly.  
If $\vect{\cM}_{h}^{(j)}$ is represented by such polynomials, then
\begin{equation}
  \int_{K^{k}} \bcM_{h}^{(j)}(x^{k})\,dx^{k} = \hat{Q}^{k}[\bcM_{h}^{(j)}] \equiv
  |K^{k}| \sum_{q=1}^{N} \hat{w}_{q}\,\bcM_{h}^{(j)}(\hat{x}_{q}^{k}),
  \label{eq:quadratureRuleGL}
\end{equation}
where for simplicity of notation, we have suppressed the explicit dependence on $\tilde{\vect{x}}^{k}$ to denote $\bcM_{h}^{(j)}(\hat{x}_q^{k},\tilde{\vect{x}}^{k}) = \bcM_{h}^{(j)}(\hat{x}_{q}^{k})$.  
In each element, we also denote $\vect{\cM}_{h}^{(j)}(\hat{x}_{1}^{k})=\vect{\cM}_{h}^{(j)}(\xL^{k,+})$ and $\vect{\cM}_{h}^{(j)}(\hat{x}_{N}^{k})=\vect{\cM}_{h}^{(j)}(\xH^{k,-})$.  
Similarly, the solution on $\tilde{\vect{\bK}}^{k}$ to the immediate left of $\xL^{k}$ is denoted $\vect{\cM}_{h}^{(j)}(\xL^{k,-})$, and the solution on $\tilde{\vect{\bK}}^{k}$ to the immediate right of $\xH^{k}$ is denoted $\vect{\cM}_{h}^{(j)}(\xH^{k,+})$.  

Using Eq.~\eqref{eq:quadratureRuleGL}, $\vect{\Gamma}^{k}$ can be expressed as the convex combination
\begin{align}
  \vect{\Gamma}^{k}\big[\bcM_{h}^{(j)};\hat{c}_{ij}\big]
  &=\sum_{q=2}^{N-1}\hat{w}_{q}\,\vect{\cM}_{h}^{(j)}(\hat{x}_{q}^{k}) \nonumber \\
  &\hspace{-12pt}
  +\hat{w}_{1}\,\big[\,\vect{\cM}_{h}^{(j)}(\xL^{k,+})+\lambda_{ij}^{k}\,\mathscr{F}^{k}\big(\vect{\cM}_{h}^{(j)}(\xL^{k,-}),\vect{\cM}_{h}^{(j)}(\xL^{k,+})\big)\,\big] \nonumber \\
  &\hspace{-12pt}
  +\hat{w}_{N}\,\big[\,\vect{\cM}_{h}^{(j)}(\xH^{k,-})-\lambda_{ij}^{k}\,\mathscr{F}^{k}\big(\vect{\cM}_{h}^{(j)}(\xH^{k,-}),\vect{\cM}_{h}^{(j)}(\xH^{k,+})\big)\,\big],
  \label{eq:realizableGammaConvex}
\end{align}
where $\lambda_{ij}^{k}=\hat{c}_{ij}\,\dt/(\beta_{k}\,\hat{w}_{N}\,|K^{k}|)$.  
(With the GL quadrature rule in Eq.~\eqref{eq:quadratureRuleGL}, $\hat{w}_{1}=\hat{w}_{N}$).  
The following Lemma establishes sufficient conditions for realizability of $\vect{\Gamma}^{k}\big[\bcM_{h}^{(j)};\hat{c}_{ij}\big]$, and hence $\vect{\cM}_{\bK}^{(ij)}$.  
\begin{lemma}
  Assume that $\vect{\cM}_{h}^{(j)}(\hat{x}_{q}^{k})\in\cR$ for all $q=1,\ldots,N$ and all $\bK\in\mathscr{T}$.  
  Let the time step $\dt$ be chosen so that $\lambda_{ij}^{k}\le1$.  
  Let the numerical flux be given by the Lax-Friedrichs flux in Eq.~\eqref{eq:fluxFunctionLF} with $\alpha^{k}=1$.  
  Then $\vect{\Gamma}^{k}\big[\bcM_{h}^{(j)};\hat{c}_{ij}\big]\in\cR$.  
  \label{lem:realizableGamma}
\end{lemma}
\begin{proof}
  In Eq.~\eqref{eq:realizableGammaConvex}, $\vect{\Gamma}^{k}\big[\bcM_{h}^{(j)};\hat{c}_{ij}\big]$ is expressed as a convex combination.  
  By assumption, $\vect{\cM}_{h}^{(j)}(\hat{x}_{q}^{k})\in\cR$ ($q=2,\ldots,N-1$).  
  Thus it remains to show that
  \begin{align*}
    &\Phi\big[\,\vect{\cM}_{h}^{(j)}(\xL^{k,-}),\vect{\cM}_{h}^{(j)}(\xL^{k,+}),\vect{\cM}_{h}^{(j)}(\xH^{k,-}),\vect{\cM}_{h}^{(j)}(\xH^{k,+});\hat{c}_{ij}\,\big] \\
    &\hspace{6pt}
    :=\f{1}{2}\big[\,\vect{\cM}_{h}^{(j)}(\xL^{k,+})+\lambda_{ij}^{k}\,\mathscr{F}^{k}\big(\vect{\cM}_{h}^{(j)}(\xL^{k,-}),\vect{\cM}_{h}^{(j)}(\xL^{k,+})\big)\,\big] \\
    &\hspace{18pt}
    +\f{1}{2}\big[\,\vect{\cM}_{h}^{(j)}(\xH^{k,-})-\lambda_{ij}^{k}\,\mathscr{F}^{k}\big(\vect{\cM}_{h}^{(j)}(\xH^{k,-}),\vect{\cM}_{h}^{(j)}(\xH^{k,+})\big)\,\big]
  \end{align*}
  is realizable.  
  Using the Lax-Friedrichs flux in Eq.~\eqref{eq:fluxFunctionLF}, with $\alpha^{k}=1$ ($k\in\{1,\ldots,d\}$), it is straightforward to show that
  \begin{equation}
    \Phi=(1-\lambda_{ij}^{k})\,\Phi_{0} + \f{1}{2}\,\lambda_{ij}^{k}\,\Phi_{1} + \f{1}{2}\,\lambda_{ij}^{k}\,\Phi_{2},
    \label{eq:phiConvex}
  \end{equation}
  where
  \begin{align*}
    \Phi_{0} &= \f{1}{2}\,\big(\,\vect{\cM}_{h}^{(j)}(\xL^{k,+})+\vect{\cM}_{h}^{(j)}(\xH^{k,-})\,\big), \\
    \Phi_{1} &= \Phi^{+}\big[\vect{\cM}_{h}^{(j)}(\xL^{k,-})\big]+\Phi^{-}\big[\vect{\cM}_{h}^{(j)}(\xH^{k,-})\big], \\
    \Phi_{2} &= \Phi^{+}\big[\vect{\cM}_{h}^{(j)}(\xL^{k,+})\big]+\Phi^{-}\big[\vect{\cM}_{h}^{(j)}(\xH^{k,+})\big],
  \end{align*}
  and $\Phi^{\pm}(\vect{\cM})=\f{1}{2}\,\big(\vect{\cM}\pm\vect{e}_{k}\cdot\vect{\cF}^{k}(\vect{\cM})\big)$; cf. Lemma~\ref{lem:explicitStep}.  
  Since $\lambda_{ij}^{k}\le1$, $\Phi$ is expressed as a convex combination of $\Phi_{0}$, $\Phi_{1}$, and $\Phi_{2}$.  
  By assumption, $\vect{\cM}_{h}^{(j)}(\xL^{k,-})$, $\vect{\cM}_{h}^{(j)}(\xL^{k,+})$, $\vect{\cM}_{h}^{(j)}(\xH^{k,-})$, $\vect{\cM}_{h}^{(j)}(\xH^{k,+})\in\cR$, which immediately implies realizability of $\Phi_{0}$.  
  Realizability of $\Phi_{1}$ and $\Phi_{2}$ follows by invoking Lemma~\ref{lem:explicitStep}.  
  This completes the proof.  
\end{proof}

\begin{rem}
  For the IMEX scheme to be realizability-preserving it is sufficient to set the time step such that
  \begin{equation}
    \dt \le c_{\Sch}\,\min_{k} \Big(\, \beta_{k}\,\hat{w}_{N}\,|K^{k}| \,\Big),
    \label{eq:imexTimeStep}
  \end{equation}
  where $c_{\Sch}$ is defined in \eqref{eq:imexCFL}.  
\end{rem}

\begin{rem}
  Lemma~\eqref{lem:realizableGamma} is proven without specification of $\tilde{\vect{x}}^{k}$.  
  In the numerical scheme, we need $\vect{\Gamma}^{k}\big[\bcM_{h}^{(j)};\hat{c}_{ij}\big]$ to be realizable in the quadrature set $\tilde{\vect{S}}^{k}$ used to approximate the integral over $\tilde{\bK}^{k}$ in Eq.~\eqref{eq:realizableGamma}.  
  (Typically, $\tilde{\vect{S}}^{k}$ is a tensor product of Gauss-Legendre quadrature points.)  
  We thus require $\vect{\cM}_{h}^{(j)}$ to be realizable in the quadrature set $\hat{\vect{S}}^{k}=\tilde{\vect{S}}^{k}\otimes\hat{S}^{k}\subset\bK$, where $\tilde{\vect{S}}^{k}\subset\tilde{\bK}^{k}$, and $\hat{S}^{k}\subset K^{k}$ are the GL quadrature points.  
\end{rem}

\begin{rem}
  Lemmas~\ref{lem:imexStagesCellAverage}, \ref{lem:realizableMij}, and \ref{lem:realizableGamma} establish sufficient conditions for realizability of the cell average of each IMEX stage $\vect{\cM}_{\bK}^{(i)}$, for $i=1,\ldots,s$.  
  Within each IMEX stage, the realizability-enforcing limiter discussed in Section~\ref{sec:limiter} is invoked to ensure that the numerical solution $\vect{\cM}_{h}^{(i)}(\vect{x})$ is realizable in the quadrature points $\hat{\vect{S}}^{k}$, for $k=1,\ldots,d$.  
\end{rem}

For GSA IMEX schemes, $\tilde{\vect{\cM}}_{\bK}^{n+1}=\vect{\cM}_{\bK}^{(s)}$.  
For IMEX schemes incorporating the correction step in Eq.~\eqref{eq:imexCorrection}, the cell average at $t^{n+1}$ is obtained by solving
\begin{equation}
  \vect{\cM}_{\bK}^{n+1}=\tilde{\vect{\cM}}_{\bK}^{n+1}+\alpha\,\f{\dt^{2}}{\tau^{2}}\,\vect{\cD}\,\big(\,\vect{\eta}-\vect{\cD}\,\vect{\cM}_{\bK}^{n+1}\,\big),
  \label{eq:imexCorrectionCellAverage}
\end{equation}
where $\vect{\eta}$ and $\vect{\cD}$ are defined in Eq.~\eqref{eq:collisionTermMoments}.  
For these IMEX schemes, realizability of $\vect{\cM}_{\bK}^{n+1}$ is established by the following lemma.  
\begin{lemma}
  Suppose that $\tilde{\vect{\cM}}_{\bK}^{n+1}\in\cR$ and $\vect{\cM}_{\bK}^{n+1}$ is obtained by solving Eq.~\eqref{eq:imexCorrectionCellAverage}.  
  Then $\vect{\cM}_{\bK}^{n+1}\in\cR$.  
  \label{lem:imexCorrectionCellAverage}
\end{lemma}
\begin{proof}
  The result follows immediately from Lemma~\ref{lem:correctionStep}.  
\end{proof}

\begin{rem}
  For IMEX scheme PD-ARS (\ref{app:PD-ARS}), which does not invoke the correction step (i.e., $\alpha=0$), $\vect{\cM}_{\bK}^{n+1}=\vect{\cM}_{\bK}^{(s)}$.  
\end{rem}

We are now ready to state the main result of this section.  
\begin{theorem}
  Consider the stages of the IMEX scheme in Eq.~\eqref{imexStages} applied to the DG discretization of the two-moment model in Eq.~\eqref{eq:semidiscreteDG}.  
  Suppose that
  \begin{itemize}
    \item[1.] For all $k\in\{1,\ldots,d\}$, the Gauss-Lobatto quadrature rule $\hat{Q}^{k}$ is chosen such that Eq.~\eqref{eq:quadratureRuleGL} holds.  
    \item[2.] For all $k\in\{1,\ldots,d\}$, $\vect{x}\in\hat{\vect{S}}^{k}$, and $0\le j \le i-1<s$, 
    \begin{equation*}
      \vect{\cM}_{h}^{(j)}(\vect{x})\in\cR.
    \end{equation*}
    \item[3.] The time step $\dt$ is chosen such that Eq.~\eqref{eq:imexTimeStep} holds.  
  \end{itemize}
  Then $\vect{\cM}_{\bK}^{(i)}\in\cR$.  
  \label{the:realizableDGIMEX}
\end{theorem}
\begin{proof}
  The cell average $\vect{\cM}_{\bK}^{(i)}$ is obtained by solving (cf. Eq.~\eqref{eq:imexStagesCellAverage})
  \begin{equation*}
    \vect{\cM}_{\bK}^{(i)}
    =\sum_{j=0}^{i-1}c_{ij}\,\vect{\cM}_{\bK}^{(ij)}
    +a_{ii}\,\f{\dt}{\tau}\big(\,\vect{\eta}-\vect{\cD}\,\vect{\cM}_{\bK}^{(i)}\,\big),
  \end{equation*}
  where, after invoking the quadrature rule $\tilde{\vect{Q}}^{k}$ to integrate over $\tilde{\bK}^{k}$ in Eq.~\eqref{eq:cellAverageInTermsOfGamma},
  \begin{equation*}
    \vect{\cM}_{\bK}^{(ij)}
    =\sum_{k=1}^{d}\f{\beta_{k}}{|\tilde{\bK}^{k}|}\tilde{\vect{Q}}^{k}\big(\vect{\Gamma}^{k}\big[\vect{\cM}_{h}^{(j)};\hat{c}_{ij}\big]\big).  
  \end{equation*}
  Since $\vect{\cM}_{h}^{(j)}\in\cR$ on $\hat{\vect{S}}^{k}$, for each $k\in\{1,\ldots,d\}$ and $j\in\{0,\ldots,i-1\}$, it follows from Lemma~\ref{lem:realizableGamma} that $\vect{\Gamma}^{k}\big[\vect{\cM}_{h}^{(j)};\hat{c}_{ij}\big]\in\cR$ on $\tilde{\vect{S}}^{k}$.  
  Then, realizability of $\vect{\cM}_{\bK}^{(ij)}$ follows from Lemma~\ref{lem:realizableMij}, after which $\vect{\cM}_{\bK}^{(i)}\in\cR$ follows by invoking Lemma~\ref{lem:imexStagesCellAverage}.  
\end{proof}
\section{Realizability-Enforcing Limiter}
\label{sec:limiter}

Condition~2 of Theorem~\ref{the:realizableDGIMEX} requires that the polynomial approximation $\vect{\cM}_{h}=\vect{\cM}_{h}^{(j)}$ ($j\in\{0,\ldots,i-1\}$) is realizable in every point in the quadrature set $S=\cup_{k=1}^{d}\hat{\vect{S}}^{k}$.  
Following Zhang \& Shu \cite{zhangShu_2010a} we use the limiter in \cite{liuOsher_1996} to enforce the bounds on the zeroth moment $\cJ$.  
%The bound-preserving DG-IMEX method developed in previous sections is designed to preserve realizability of the cell averaged moments, i.e., $\vect{\cM}_{\bK}\in\cR$, provided sufficiently accurate quadratures are used to integrate integrals in the DG method, a CFL condition is satisfied, and that the polynomial approximation $\vect{\cM}_{h}$, at time $t^{n}$, is realizable in a set of quadrature points in each element $\bK$.  
%We denote this quadrature set by $S=\cup_{k=1}^{d}\hat{\vect{S}}^{k}\subset\bK$.  
%In the DG method, we use the limiter proposed by Zhang \& Shu \cite{zhangShu_2010a} for scalar conservation laws to enforce the bounds on the zeroth moment $\cJ$ (see also \cite{liuOsher_1996}).  
We replace the polynomial $\cJ_{h}(\vect{x})$ with the limited polynomial
\begin{equation}
  \tilde{\cJ}_{h}(\vect{x})
  =\vartheta_{1}\,\cJ_{h}(\vect{x})+(1-\vartheta_{1})\,\cJ_{\bK},
  \label{eq:limitDensity}
\end{equation}
where the limiter parameter $\vartheta_{1}$ is given by
\begin{equation}
  \vartheta_{1}
  =\min\Big\{\,\Big|\f{M-\cJ_{\bK}}{M_{S}-\cJ_{\bK}}\Big|,\Big|\f{m-\cJ_{\bK}}{m_{S}-\cJ_{\bK}}\Big|,1\,\Big\},
\end{equation}
with $m=0$ and $M=1$, and
\begin{equation}
  M_{S}=\max_{\vect{x}\in S}\cJ_{h}(\vect{x})
  \quad\text{and}\quad
  m_{S}=\min_{\vect{x}\in S}\cJ_{h}(\vect{x}).  
\end{equation}

In the next step, we ensure realizability of the moments by following the framework of \cite{zhangShu_2010b}, developed to ensure positivity of the pressure when solving the Euler equations of gas dynamics.  
We let $\widetilde{\vect{\cM}}_{h}=\big(\tilde{\cJ}_{h},\vect{\cH}_{h}\big)^{T}$.  
Then, if $\widetilde{\bcM}_{h}$ lies outside $\cR$ for any quadrature point $\vect{x}_{q}\in S$, i.e., $\gamma(\widetilde{\bcM}_{h})<0$, there exists an intersection point of the straight line, $\vect{s}_{q}(\psi)$, connecting $\vect{\cM}_{\bK}\in\cR$ and $\widetilde{\vect{\cM}}_{h}$ evaluated in the troubled quadrature point $\vect{x}_{q}$, denoted $\widetilde{\vect{\cM}}_{q}$, and the boundary of $\cR$.  
This line is given by the convex combination 
\begin{equation}
  \vect{s}_{q}(\psi)=\psi\,\widetilde{\vect{\cM}}_{q}+(1-\psi)\,\bcM_{\bK},
\end{equation}
where $\psi\in[0,1]$, and the intersection point $\psi_{q}$ is obtained by solving $\gamma(\bs_{q}(\psi))=0$ for $\psi$, using the bisection algorithm\footnote{In practice, $\psi$ needs not be accurate to many significant digits, and the bisection algorithm can be terminated after a few iterations.}.  
We then replace the polynomial representation $\widetilde{\vect{\cM}}_{h}\to\widehat{\vect{\cM}}_{h}$, where
\begin{equation}
  \widehat{\vect{\cM}}_{h}(\vect{x})=\vartheta_{2}\,\widetilde{\vect{\cM}}_{h}(\vect{x})+(1-\vartheta_{2})\,\vect{\cM}_{\bK},
  \label{eq:limitMoments}
\end{equation}
and $\vartheta_{2}=\min_{q}\psi_{q}$ is the smallest $\psi$ obtained in the element by considering all the troubled quadrature points.  
This limiter is conservative in the sense that it preserves the cell-average $\widehat{\vect{\cM}}_{\bK}=\widetilde{\vect{\cM}}_{\bK}=\vect{\cM}_{\bK}$.  

The realizability-preserving property of the DG-IMEX scheme results from the following theorem.
\begin{theorem}
  Consider the IMEX scheme in Eqs.~\eqref{imexStages}-\eqref{eq:imexCorrection} applied to the DG discretization of the two-moment model in Eq.~\eqref{eq:semidiscreteDG}.  
  Suppose that
  \begin{itemize}
    \item[1.] The conditions of Theorem~\ref{the:realizableDGIMEX} hold.  
    \item[2.] With $\vect{\cM}_{\bK}^{(i)}\in\cR$, the limiter described above is invoked to enforce 
    \begin{equation*}
      \vect{\cM}_{h}^{(i)}(\vect{x})\in\cR ~ \text{for all} ~ \vect{x} \in S.  
    \end{equation*}
    \item[3.] The IMEX scheme is GSA.  
  \end{itemize}
  Then $\vect{\bcM}_{\bK}^{n+1}\in\cR$.  
  \label{the:realizableDGIMEX2}
\end{theorem}
\begin{proof}
  By Theorem~\ref{the:realizableDGIMEX} (with $i=1$), we have $\vect{\cM}_{\bK}^{(1)}\in\cR$.  
  Application of the realizability-enforcing limiter gives $\vect{\cM}_{h}^{(1)}(\vect{x})\in\cR$ for all $\vect{x} \in S$.  
  Repeated application of these steps give $\vect{\cM}_{h}^{(i)}(\vect{x})\in\cR$ for all $\vect{x} \in S$ and $i\in\{1,\ldots,s\}$.  
  Since the IMEX scheme is GSA, $\tilde{\vect{\cM}}_{\bK}^{n+1}\in\cR$.  
  Finally, $\vect{\bcM}_{\bK}^{n+1}\in\cR$ follows from Lemma~\ref{lem:imexCorrectionCellAverage}.  
\end{proof}
\section{Numerical Tests}
\label{sec:numerical}

In this section we present numerical results obtained with the DG-IMEX scheme developed in this paper.  
The first set of tests (Section~\ref{sec:smoothProblems}) are included to compare the time integration schemes in various regimes.  
We are not concerned with moment realizability in Section~\ref{sec:smoothProblems}, and we do not apply the realizability-enforcing limiter in these tests.  
The tests in Sections~\ref{sec:packedBeam} and \ref{sec:fermionImplosion} are designed specifically to demonstrate the robustness of the scheme to dynamics near the boundary of the realizable set $\cR$.  
The test in Section~\ref{sec:homogeneousSphere} (Homogeneous Sphere) is of astrophysical interest.  
Here we consider moment realizability and compare results obtained with various moment closures.  

\subsection{Problems with Known Smooth Solutions}
\label{sec:smoothProblems}

To compare the accuracy of the IMEX schemes, we present results from smooth problems in streaming, absorption, and scattering dominated regimes in one spatial dimension.  
For all tests in this subsection, we use third order accurate spatial discretization (polynomials of degree $k=2$) and we employ the maximum entropy closure in the low occupancy limit (i.e., the Minerbo closure).  
We compare results obtained using IMEX schemes proposed here (PA2+ and PD-ARS) with IMEX schemes from Hu et al. \cite{hu_etal_2018} (PA2), McClarren et al. \cite{mcclarren_etal_2008} (PC2), Pareschi \& Russo \cite{pareschiRusso_2005} (SSP2332), and Cavaglieri \& Bewley \cite{cavaglieriBewley2015} (RKCB2).  
In the streaming test, we also include results obtained with second-order and third-order accurate explicit strong stability-preserving Runge-Kutta methods \cite{gottlieb_etal_2001} (SSPRK2 and SSPRK3, respectively).  
See \ref{app:butcherTables} for further details.  
The time step is set to $\dt=0.1\times\dx$.  

When comparing the numerical results to analytic solutions, errors are computed in the $L^{1}$-error norm.  
We compare results either in the absolute error ($E_{\mbox{\tiny Abs}}^{1}$) or the relative error ($E_{\mbox{\tiny Rel}}^{1}$), defined for a scalar quantity $u_{h}$ (approximating $u$) as
\begin{equation}
  E_{\mbox{\tiny Abs}}^{1}[u_{h}](t)
  =\f{1}{|D|}\sum_{\bK\in\mathscr{T}}\int_{\bK}|u_{h}(\vect{x},t)-u(\vect{x},t)|\,d\vect{x}
  \label{eq:errorNormAbsolute}
\end{equation}
and
\begin{equation}
  E_{\mbox{\tiny Rel}}^{1}[u_{h}](t)
  =\f{1}{|D|}\sum_{\bK\in\mathscr{T}}\int_{\bK}|u_{h}(\vect{x},t)-u(\vect{x},t)|/|u(\vect{x},t)|\,d\vect{x},
  \label{eq:errorNormRelative}
\end{equation}
respectively.  
The integrals in Eqs.~\eqref{eq:errorNormAbsolute} and \eqref{eq:errorNormRelative} are computed with a simple $3$-point equal weight quadrature.  

\subsubsection{Sine Wave: Streaming}

The first test involves the streaming part only, and does not include any collisions ($\sigma_{\Ab}=\sigma_{\Scatt}=0$).  
We consider a periodic domain $D=\{x:x\in[0,1]\}$, and let the initial condition be given by
\begin{equation}
  \cJ(x,t=0)=\cH_{x}(x,t=0)=0.5+0.49\times\sin\big(2\pi\,x\big).  
  \label{eq:initialConditionStreaming}
\end{equation}
We evolve until $t=10$, when the sine wave has completed 10 crossings of the computational domain.  
We vary the number of elements ($N$) from $8$ to $128$ and compute errors for various time stepping schemes.  

In Figure~\ref{fig:SineWaveStreaming}, the absolute error for the number density $E_{\mbox{\tiny Abs}}^{1}[\cJ_{h}](t=10)$ is plotted versus $N$ (see figure caption for details).  
Errors obtained with SSPRK3 are smallest and decrease as $N^{-3}$ (cf. bottom black dash-dot reference line), as expected for a scheme combining third-order accurate time stepping with third-order accurate spatial discretization.  
For all the other schemes, using second-order accurate explicit time stepping, the error decreases as $N^{-2}$.  
Among the second-order accurate methods, SSP2332 has the smallest error, followed by RKCB2.  
Errors for the remaining schemes (including SSPRK2) are indistinguishable on the plot.  
\begin{figure}[H]
  \centering
    \includegraphics[width=\textwidth]{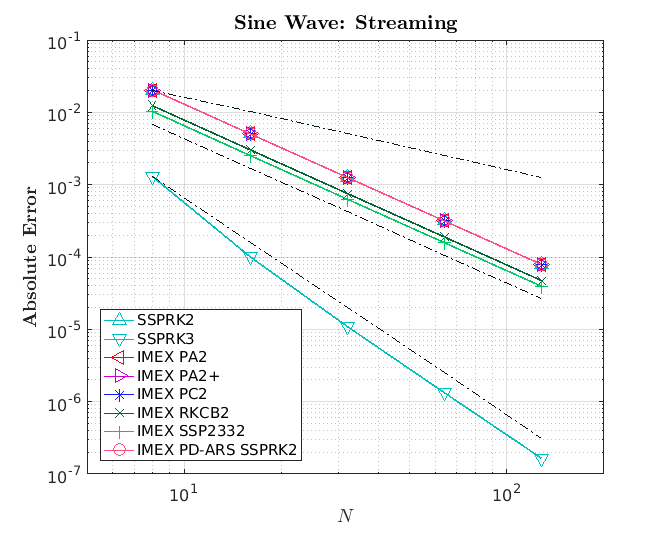}
   \caption{Absolute error (cf. Eq.~\eqref{eq:errorNormAbsolute}) versus number of elements $N$ for the streaming sine wave test.  Results employing various time stepping schemes are compared: SSPRK2 (cyan triangles pointing up), SSPRK3 (cyan triangles pointing down), PA2 (red), PA2+ (purple), PC2 (blue), RKCB2 (dark green), SSP2332 (green), and PD-ARS (light red circles).  Black dash-dot reference lines are proportional to $N^{-1}$ (top), $N^{-2}$ (middle), and $N^{-3}$ (bottom), respectively.}
  \label{fig:SineWaveStreaming}
\end{figure}

\subsubsection{Sine Wave: Damping}

The next test we consider, adapted from \cite{skinnerOstriker_2013}, consists of a sine wave propagating with unit speed in a purely absorbing medium ($f_{0}=0$, $\sigma_{\Scatt}=0$), which results in exponential damping of the wave amplitude.  
We consider a periodic domain $D=\{x:x\in[0,1]\}$, and let the initial condition ($t=0$) be given as in Eq.~\eqref{eq:initialConditionStreaming}.  
For a constant absorption opacity $\sigma_{\Ab}$, the analytical solution at $t>0$ is given by
\begin{equation}
  \cJ(x,t)=\cJ_{0}(x-t)\times\exp(-\sigma_{\Ab} t)
  \quad\text{and}\quad
  \cH_{x}(x,t)=\cJ(x,t),
\end{equation}
where $\cJ_{0}(x)=\cJ(x,0)$.  

We compute numerical solutions for three values of the absorption opacity ($\sigma_{\Ab}=0.1$, $1$, and $10$), and adjust the end time $t_{\mbox{\tiny end}}$ so that $\sigma_{\Ab}t_{\mbox{\tiny end}}=10$, and the initial condition has been damped by factor $e^{-10}$.  
Thus, for $\sigma_{\Ab}=0.1$ the sine wave crosses the domain 100 times, while for $\sigma_{\Ab}=10$, it crosses the grid once.  

Figure~\ref{fig:SineWaveDamping} shows convergence results, obtained using different values of $\sigma_{\Ab}$, for various IMEX schemes at $t=t_{\mbox{\tiny end}}$.  
Results for $\sigma_{\Ab}=0.1$, $1$, and $10$ are plotted with red, green, and blue lines, respectively (see figure caption for further details).  
All the second-order accurate schemes (PA2, PA2+, RKCB2, and SSP2332) display second-order convergence rates (cf. bottom, black dash-dot reference line).  
For $\sigma_{\Ab}=0.1$, SSP2332 is the most accurate among these schemes, while PA2+ is the most accurate for $\sigma_{\Ab}=10$.  
On the other hand, PC2 and PD-ARS are indistinguishable and display at most first-order accurate convergence, as expected.  
(For $\sigma_{\Ab}=0.1$, PC2 and PD-ARS are the most accurate schemes for $N=8$ and $N=16$.)

\begin{figure}[H]
  \centering
    \includegraphics[width=\textwidth]{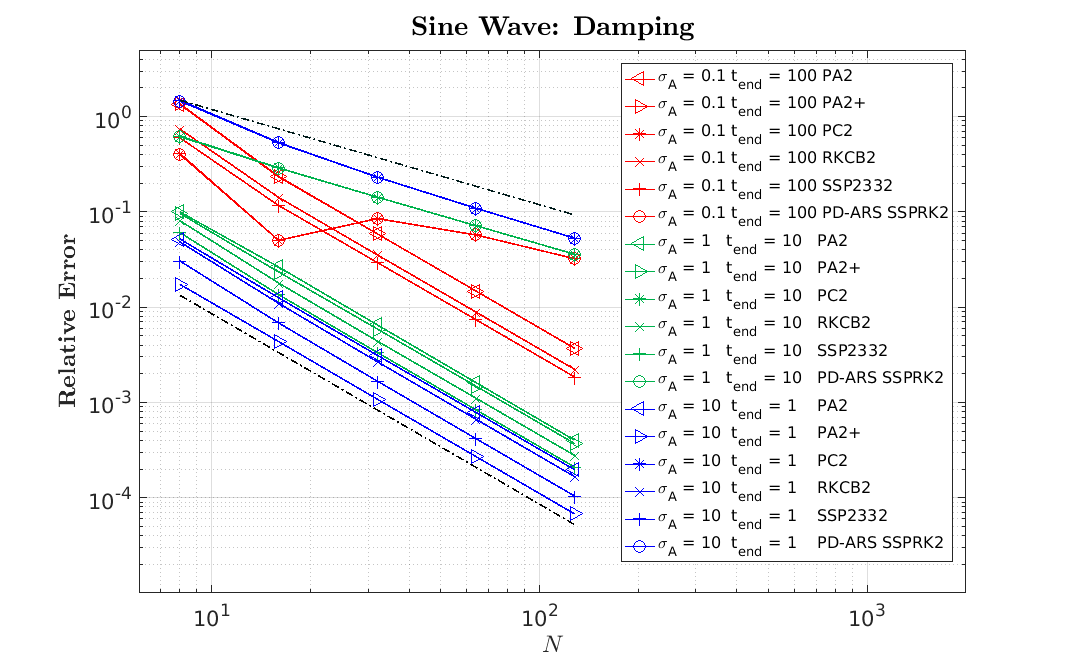}
   \caption{Relative error (cf. Eq.~\eqref{eq:errorNormRelative}) versus number of elements for the damping sine wave test.  Results for different values of the absorption opacity $\sigma_{\Ab}$, employing various IMEX time stepping schemes, are compared.  Errors for $\sigma_{\Ab}=0.1$, $1$, and $10$ are plotted with red, green, and blue lines, respectively.  The IMEX schemes employed are: PA2 (triangles pointing left), PA2+ (triangles pointing right), PC2 (asterisk), RKCB2 ($\times$), SSP2332 ($+$), and PD-ARS (circles).  Black dash-dot reference lines are proportional to $N^{-1}$ (top) and $N^{-2}$ (bottom), respectively.}
  \label{fig:SineWaveDamping}
\end{figure}

\subsubsection{Sine Wave: Diffusion}

The final test with known smooth solutions, adopted from \cite{radice_etal_2013}, is diffusion of a sine wave in a purely scattering medium ($f_{0}=0$, $\sigma_{\Ab}=0$).  
The computational domain $D=\{x:x\in[-3,3]\}$ is periodic, and the initial condition is given by
\begin{equation}
  \cJ_{0}(x)=0.5+0.49\times\sin\big(\f{\pi\,x}{3}\big)
  \quad\text{and}\quad
  \cH_{x,0}
  =-\f{1}{3\sigma_{\Scatt}}\pderiv{\cJ_{0}}{x}.  
  \label{eq:initialConditionDiffusion}
\end{equation}
For a sufficiently high scattering opacity, the moment equations limit to a diffusion equation for the number density (deviations appear at the $1/\sigma_{\Scatt}^{2}$-level).  
With the initial conditions in Eq.~\eqref{eq:initialConditionDiffusion}, the analytical solution to the limiting diffusion equation is given by
\begin{equation}
  \cJ(x,t)=\cJ_{0}(x)\times\exp\big(-\f{\pi^{2}\,t}{27\,\sigma_{\Scatt}}\big),
\end{equation}
and $\cH_{x}=(3\,\sigma_{\Scatt})^{-1}\pd{\cJ}{x}$.  
When computing errors for this test, we compare the numerical results obtained with the two-moment model to the analytical solution to the limiting diffusion equation.  
We compute numerical solutions using three values of the scattering opacity ($\sigma_{\Scatt}=10^{2}$, $10^{3}$, and $10^{4}$), and adjust the end time so that $t_{\mbox{\tiny end}}/\sigma_{\Scatt}=1$.  
The initial amplitude of the sine wave has then been reduced by a factor $e^{-\pi^{2}/27}\approx0.694$ for all values of $\sigma_{\Scatt}$.  

\begin{figure}[H]
  \centering
  \includegraphics[width=1.0\textwidth]{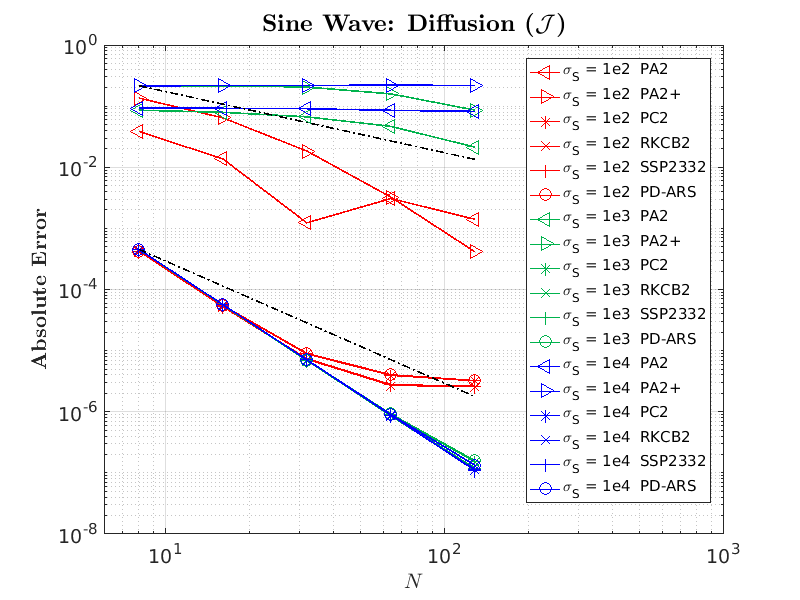}
   \caption{Absolute error (cf. Eq.~\eqref{eq:errorNormAbsolute}) for the number density $\cJ$ versus number of elements for the sine wave diffusion test.  Results with different values of the scattering opacity $\sigma_{\Scatt}$, employing different IMEX schemes, are compared.  Errors with $\sigma_{\Scatt}=10^{2}$, $10^{3}$, and $10^{4}$ are plotted with red, green, and blue lines, respectively.  The IMEX schemes employed are: PA2 (triangle pointing left), PA2+ (triangle pointing right), PC2 (asterisk), RKCB2 (cross), SSP2332 (plus), and PD-ARS (circle).  Black dash-dot reference lines are proportional to $N^{-1}$ (top) and $N^{-2}$ (bottom), respectively.}
  \label{fig:SineWaveDiffusionJ}
\end{figure}

\begin{figure}[H]
  \centering
  \includegraphics[width=1.0\textwidth]{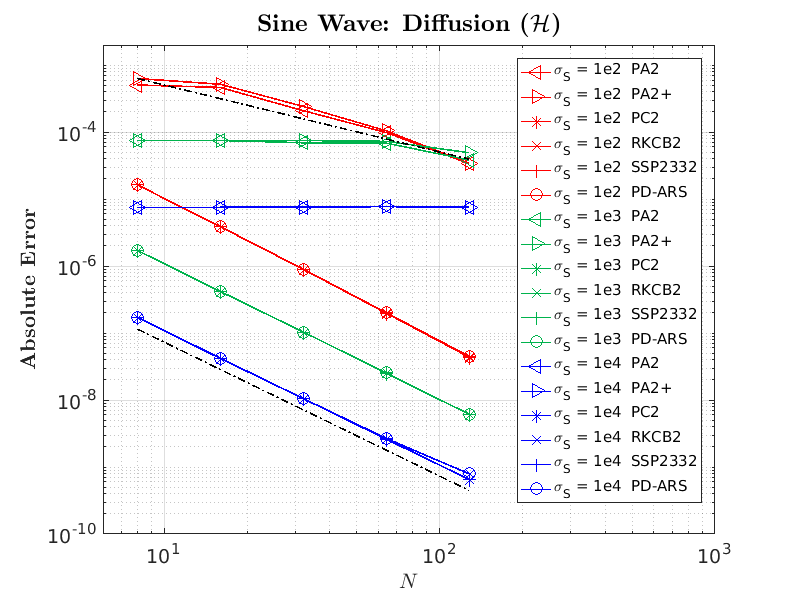}
   \caption{Same as in Figure~\ref{fig:SineWaveDiffusionJ}, but for the number flux $\cH_{x}$.}
  \label{fig:SineWaveDiffusionH}
\end{figure}

In Figures~\ref{fig:SineWaveDiffusionJ} and \ref{fig:SineWaveDiffusionH} we plot the absolute error, obtained using different values of $\sigma_{\Scatt}$, for various IMEX schemes at $t=t_{\mbox{\tiny end}}$.  
Results for $\sigma_{\Scatt}=10^{2}$, $10^{3}$, and $10^{4}$ are plotted with red, green, and blue lines, respectively (see figure caption for further details).  
(Scheme PC2 has been shown to work well for this test \cite{radice_etal_2013}, but is included here for comparison with the other IMEX schemes.)
Schemes PD-ARS, RKCB2, and SSP2332 are accurate for this test, and display third-order accuracy for the number density $\cJ$ and second-oder accuracy for $\cH_{x}$.  
For $\sigma=10^{2}$, the errors do not drop below $10^{-6}$ because of differences between the two-moment model and the diffusion equation used to obtain the analytic solution.  
For larger values of the scattering opacity, the two-moment model agrees better with the diffusion model, and we observe convergence over the entire range of $N$.  
Schemes PA2 and PA2+ do not perform well on this test (for reasons discussed in Section~\ref{sec:imex}).  
For $\sigma_{\Scatt}=10^{2}$, errors in $\cJ$ and $\cH_{x}$ decrease with increasing $N$, but for $\sigma_{\Scatt}=10^{4}$, errors remain constant with increasing $N$ over the entire range.  

\subsection{Packed Beam}
\label{sec:packedBeam}

Next we consider a one-dimensional test with discontinuous initial conditions.  
The purpose of this test is to further gauge the accuracy of the two-moment model and demonstrate the robustness of the DG scheme for dynamics close to the boundary of the realizable set $\cR$.  
The computational domain is $D=\{x:x\in[-1,1]\}$, and the initial condition is obtained from a distribution function given by
\begin{equation}
  f(x,\mu)
  =\left\{
  \begin{array}{cl}
    1        & \text{if} ~ x\le x_{\mbox{\tiny D}}, ~ \mu\ge\mu_{\mbox{\tiny D}} \\
    \delta & \text{if} ~ x\le x_{\mbox{\tiny D}}, ~ \mu<   \mu_{\mbox{\tiny D}} \\
    \delta & \text{otherwise},
  \end{array}
  \right.
\end{equation}
so that, with $\mu_{\mbox{\tiny D}}=0$, $\vect{\cM}\equiv\vect{\cM}_{\mbox{\tiny L}}=\big(0.5\,(1+\delta),0.25\,(1-\delta)\big)^{T}$ for $x\le x_{\mbox{\tiny D}}$, and $\vect{\cM}\equiv\vect{\cM}_{\mbox{\tiny R}}=\big(\delta,0\big)^{T}$ for $x> x_{\mbox{\tiny D}}$, where $\delta>0$ is a small parameter ($\delta\ll1$).  
We let $\delta=10^{-8}$, so that the initial conditions are very close to the boundary of the realizable domain (cf. Figure~\ref{fig:RealizableSetFermionic}).  
The analytical solution can be easily obtained by solving the transport equation for all angles $\mu$ (independent linear advection equations), and taking the angular moments.  
The numerical results shown in this section were obtained with the third-order scheme (polynomials of degree $k=2$ and the SSPRK3 time stepper) using $400$ elements.  
The time step is set to $\dt=0.1\times\dx$

Figure~\ref{fig:PackedBeam} shows results for various times obtained with the two-moment model.  
In the upper panels we plot the number density, while the number flux density is plotted in the lower panels.  
Numerical solutions are plotted with solid lines, while the analytical solution is plotted with dashed lines.  
In the left panels, the algebraic maximum entropy closure of Cernohorsky \& Bludman (CB) \cite{cernohorskyBludman_1994} (cf. Eqs.~\eqref{eq:eddingtonFactor} and \eqref{eq:closureMECB}) was used, while in the right panels the Minerbo closure (cf. Eqs.~\eqref{eq:eddingtonFactorLow} and \eqref{eq:closureMECB}) was used.  
For this test, the use of the realizability-preserving limiter described in Section~\ref{sec:limiter} was essential in order to avoid numerical problems.  
For the results obtained with the CB closure, the limiter was enacted whenever moments ventured outside the realizable set given by Eq.~\eqref{eq:realizableSet}.  
For the results obtained with the Minerbo closure, which is not based on Fermi-Dirac statistics, we used a modified limiter, which was enacted when the moments ventured outside the realizable domain of positive distributions; i.e., not bounded by $f < 1$, so that $\cJ > 0$ and $\cJ > \vect{\cH}|$ (e.g., \cite{levermore_1984}; see red line in Figure~\ref{fig:RealizableSetFermionic}).  

\begin{figure}[H]
  \centering
  \begin{tabular}{cc}
    \includegraphics[width=0.5\textwidth]{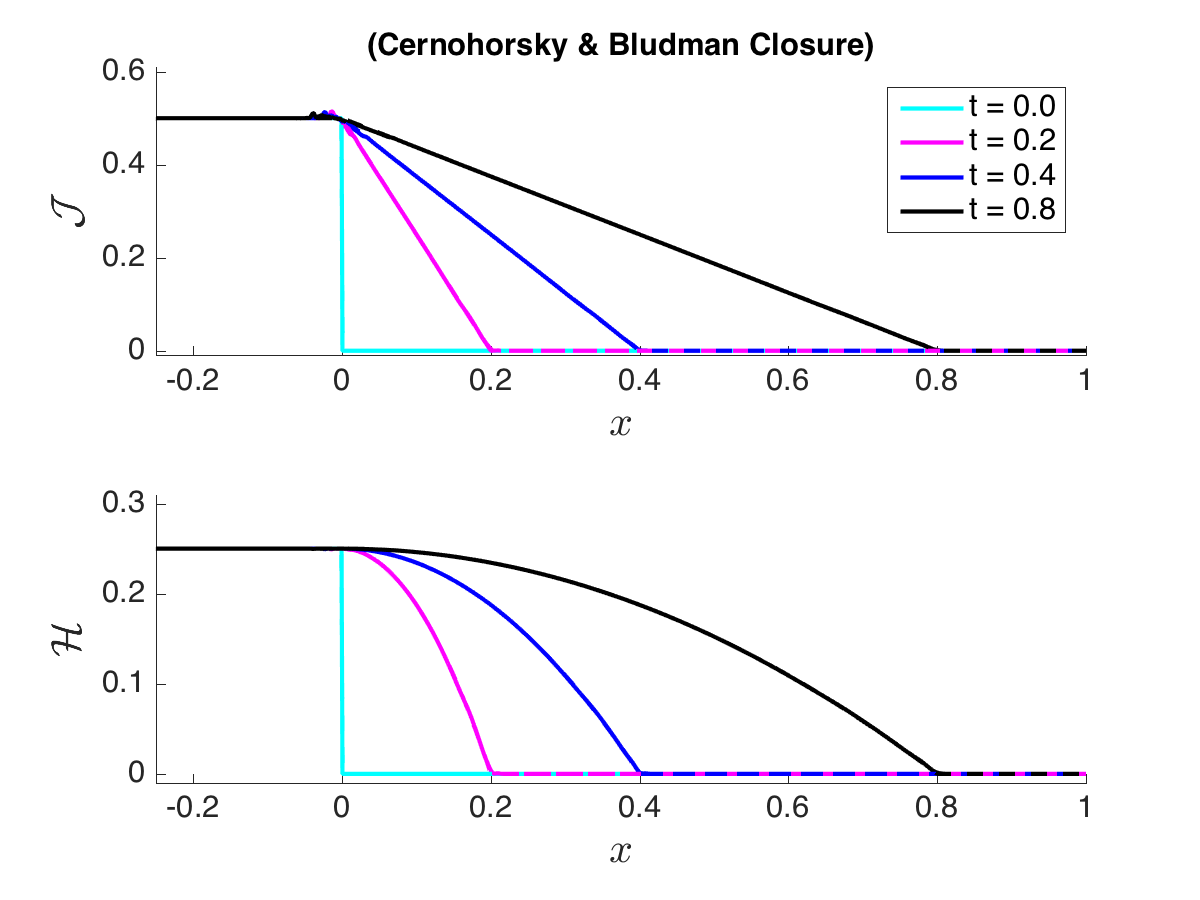} &
    \includegraphics[width=0.5\textwidth]{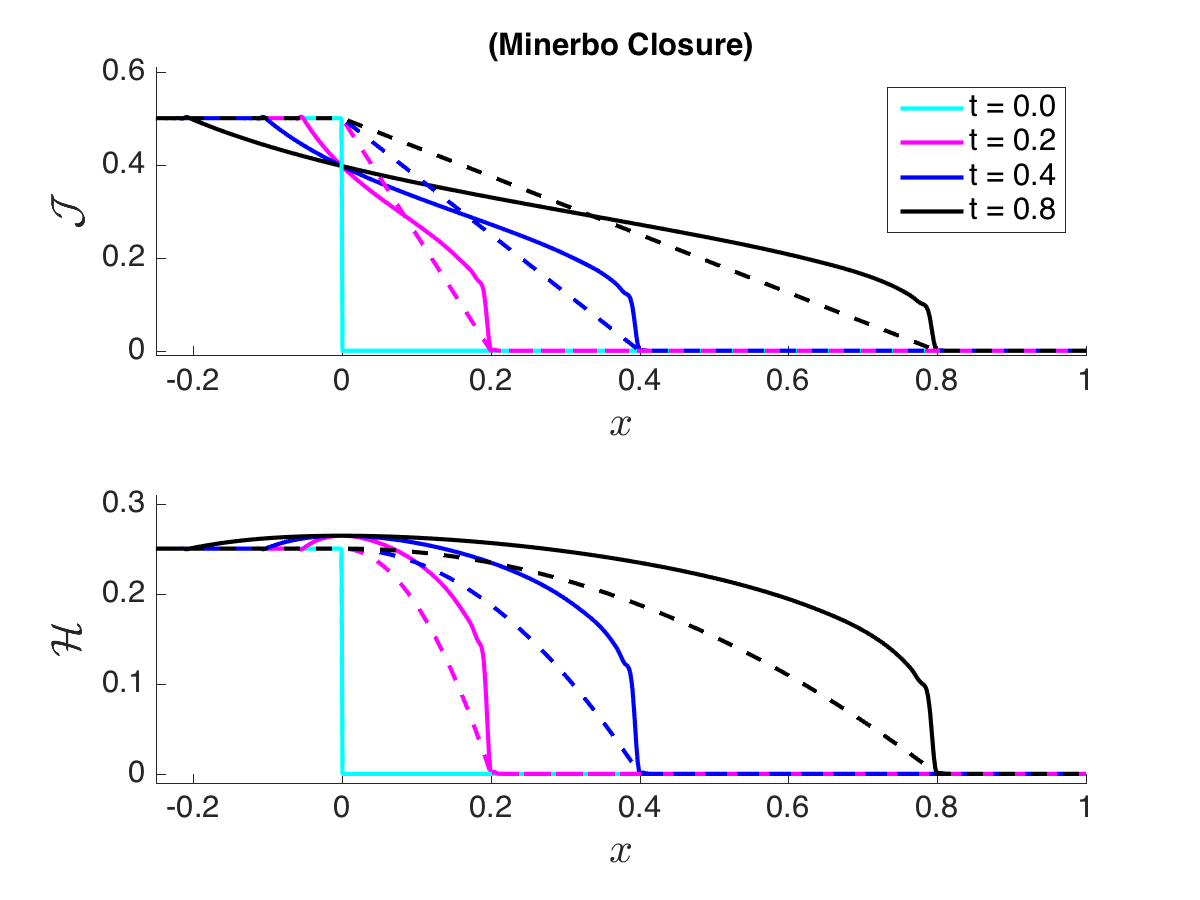}
  \end{tabular}
   \caption{Numerical results from the packed beam problem at various times: $t=0$ (cyan), $t=0.2$ (magenta), $t=0.4$ (blue), and $t=0.8$ (black).  Results obtained with the Cernohorsky \& Bludman closure are displayed in the left panels, while results obtained with the Minerbo closure are displayed in the right panels.  The analytical solution (dashed lines) is also plotted.}
  \label{fig:PackedBeam}
\end{figure}

As can be seen in Figure~\ref{fig:PackedBeam}, with the CB closure the numerical solution obtained with the two-moment model tracks the analytic solution well, while with the Minerbo closure the numerical solution deviates substantially from the analytic solution.  
With the Minerbo closure, the solution also evolves outside the realizable domain for Fermi-Dirac statistics.  

In the left panel in Figure~\ref{fig:PackedBeam_Realizability} we plot $\gamma(\vect{\cM})=\big(1-\cJ\big)\,\cJ-|\vect{\cH}|$ versus position for various times.  
With the Minerbo closure, $\gamma(\vect{\cM})$ becomes negative in regions of the computational domain (dashed lines), while $\gamma(\vect{\cM})$ remains positive for all $x$ and $t$ the CB closure.  
In the right panel of Figure~\ref{fig:PackedBeam_Realizability} we plot the numerical solutions in the $(\cH,\cJ)$-plane.  
Initially, the moments are located in two points: $\vect{\cM}_{\mbox{\tiny L}}$ and $\vect{\cM}_{\mbox{\tiny R}}$, for $x\le0$ and $x>0$, respectively (marked by circles in Figure~\ref{fig:PackedBeam_Realizability}).  
For $t>0$, the solutions trace out curves in the $(\cH,\cJ)$-plane, connecting $\vect{\cM}_{\mbox{\tiny L}}$ and $\vect{\cM}_{\mbox{\tiny R}}$.  
With the CB closure, the solution curve (blue points) follows the boundary of the realizable set $\cR$ defined in Eq.~\eqref{eq:realizableSet} (cf. black line in Figure~\ref{fig:PackedBeam_Realizability}).  
With the Minerbo closure (magenta points), the solution follows a different curve --- outside the realizable domain for distribution functions bounded by $f\in(0,1)$, but inside the realizable domain of positive distributions (cf. red line in Figure~\ref{fig:PackedBeam_Realizability}).  
We have also run this test using the algebraic maximum entropy closure of Larecki \& Banach \cite{lareckiBanach_2011} and the simpler Kershaw-type closure in \cite{banachLarecki_2017a}.  
The numerical solutions obtained with both of these closures follow the analytic solution well, and remain within the realizable set $\cR$.  
We point out that simply using the realizability-preserving limiter described in Section~\ref{sec:limiter} with the Minerbo closure does not result in a realizability-preserving scheme for Fermi-Dirac statistics because of the properties of this closure discussed in Section~\ref{sec:algebraicClosure}, and plotted in the right panel of Figure~\ref{fig:MabWithDifferentClosure}.  

\begin{figure}[H]
  \centering
  \begin{tabular}{cc}
    \includegraphics[width=0.485\textwidth]{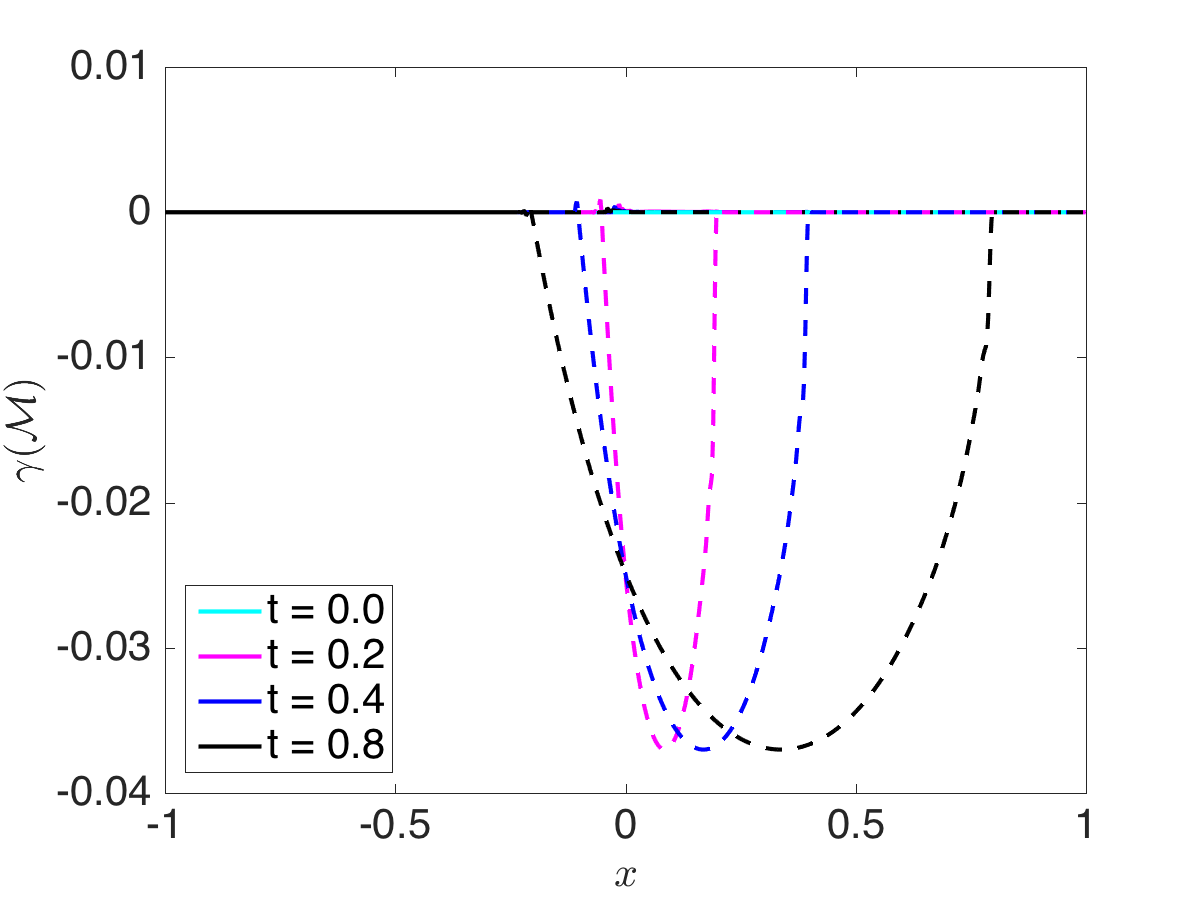} &
    \includegraphics[width=0.485\textwidth]{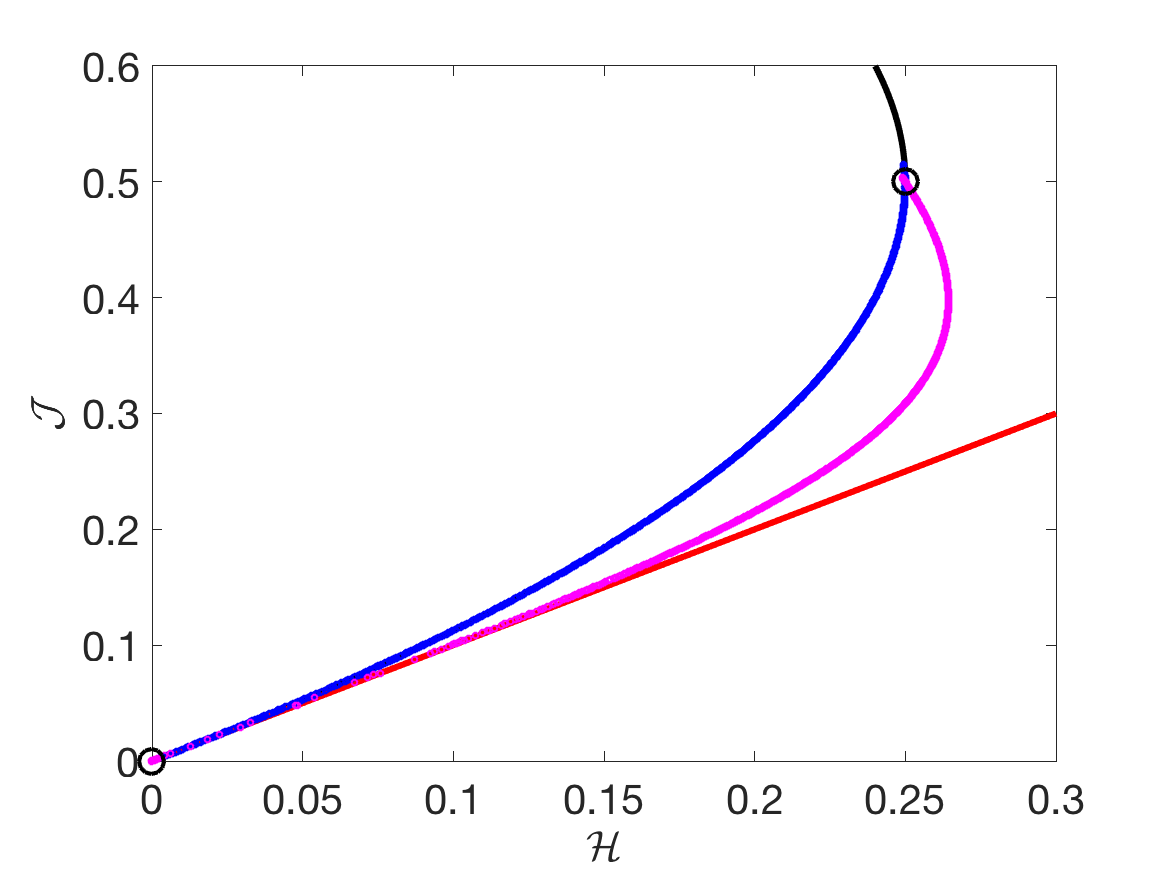}
  \end{tabular}
   \caption{In the left panel, $\gamma(\vect{\cM})=(1-\cJ)\,\cJ-|\vect{\cH}|$ is plotted versus $x$ for various times in the packed beam problem: $t=0$ (cyan), $t=0.2$ (magenta), $t=0.4$ (blue), and $t=0.8$ (black).  Results obtained with the CB closure, which remain positive throughout the evolution, are plotted with solid lines, while results obtained with the Minerbo closure are plotted with dashed lines.  In the right panel, the moments are plotted in the $(\cH,\cJ)$-plane for the same times as in the left panel.  Results obtained with the CB and Minerbo closures are plotted in blue and magenta, respectively.  The solid black and red lines are contours where $(1-\cJ)\,\cJ=\cH$ and $\cJ=\cH$, respectively.  The initial states are marked with black circles.}
  \label{fig:PackedBeam_Realizability}
\end{figure}

\subsection{Fermion Implosion}
\label{sec:fermionImplosion}

The next test is inspired by line source benchmark (cf. \cite{brunner_2002,garrettHauck_2013}), which is a challenging test for approximate transport algorithms.  
The original line source test consists of an initial delta function particle distribution in radius $R=|\vect{x}|$; i.e., $f_{0}=\delta(R)$.  
For $t>0$, a radiation front propagates in the radial direction, away from $R=0$.  
Apart from capturing details of the exact transport solution, maintaining realizability of the two-moment solution is challenging.  

Here, a modified version of the line source --- dubbed \emph{Fermion Implosion},  designed to test the realizability-preserving properties of the two-moment model for fermion transport --- is computed on a two-dimensional domain $D=\{\vect{x}\in\bbR^{2}:x^{1}\in[-1.28,1.28], x^{2}\in[-1.28,1.28]\}$.  
Instead of initializing with a delta function, we follow the initialization procedure in \cite{garrettHauck_2013}, and approximate the initial condition using an isotropic Gaussian distribution function.  
However, different from \cite{garrettHauck_2013}, the initial distribution function is bounded $f_{0}\in(0,1)$, and reaches a minimum in the center of the computational domain (hence implosion)
\begin{equation}
  f_{0}
  =1-\max\Big[\,e^{-R^{2}/(2\,\sigma_{0}^{2})},10^{-8}\,\Big].  
\end{equation}
We set $\sigma_{0}=0.03$, and evolve to a final time of $t=1.0$.  
We run this test using a grid of $512^{2}$ elements, polynomials of degree $k=1$, and the SSPRK2 time stepping scheme with $\dt=0.1\times\dx^{1}$.  
(There are no collisions included in this test; i.e., $\sigma_{\Ab}=\sigma_{\Scatt}=0$.)  
For comparison, we present results using the algebraic closures of Cernohorsky \& Bludman (CB) and Minerbo.  

\begin{figure}[H]
  \centering
  \begin{tabular}{cc}
    \includegraphics[width=0.495\textwidth]{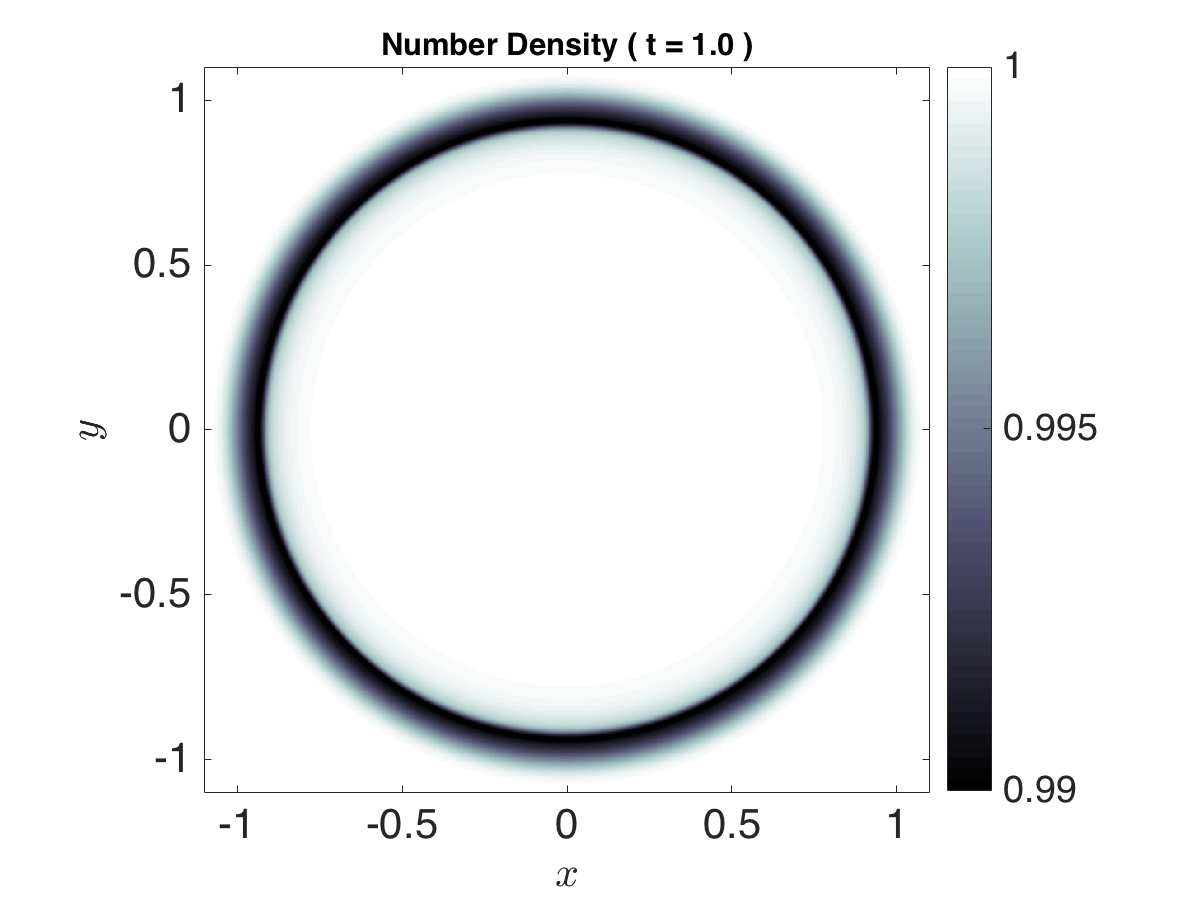} &
    \includegraphics[width=0.495\textwidth]{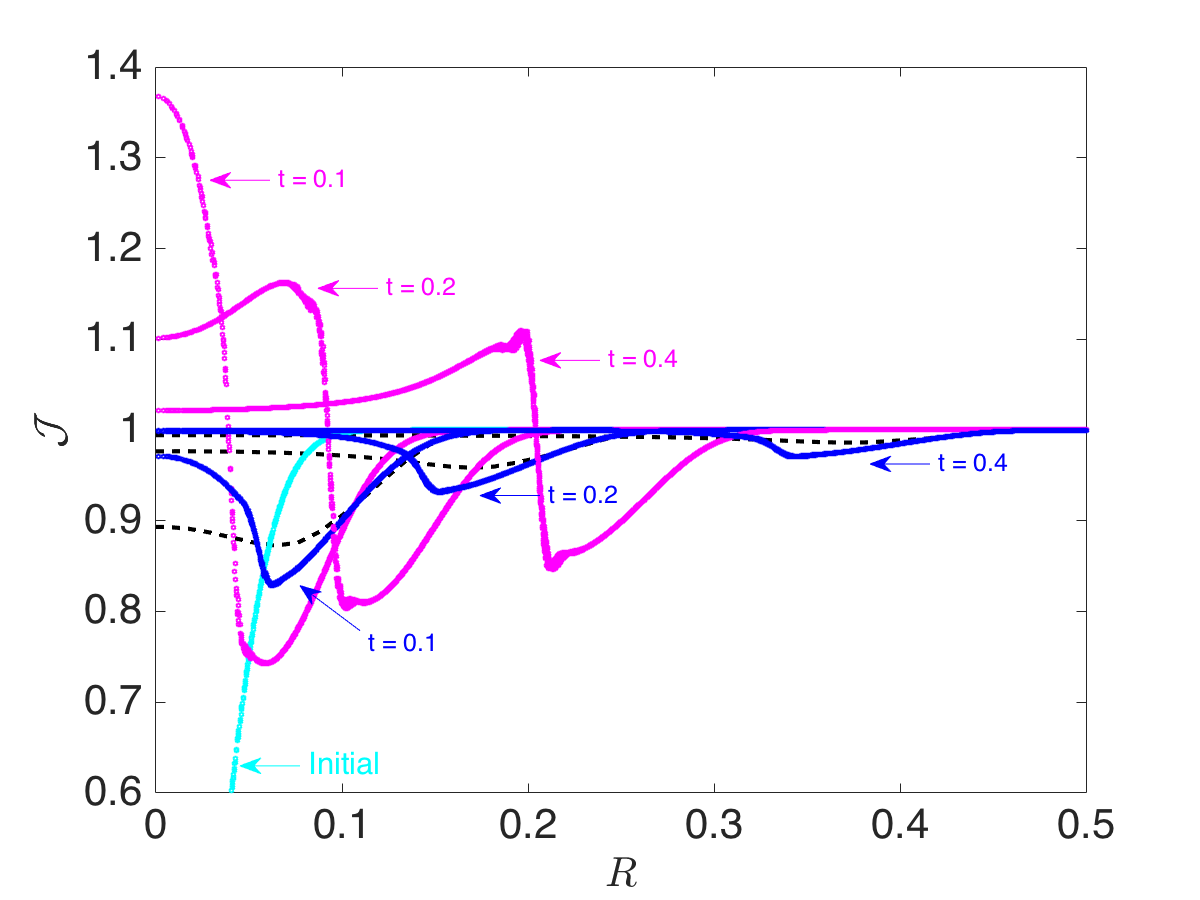} \\
    \includegraphics[width=0.495\textwidth]{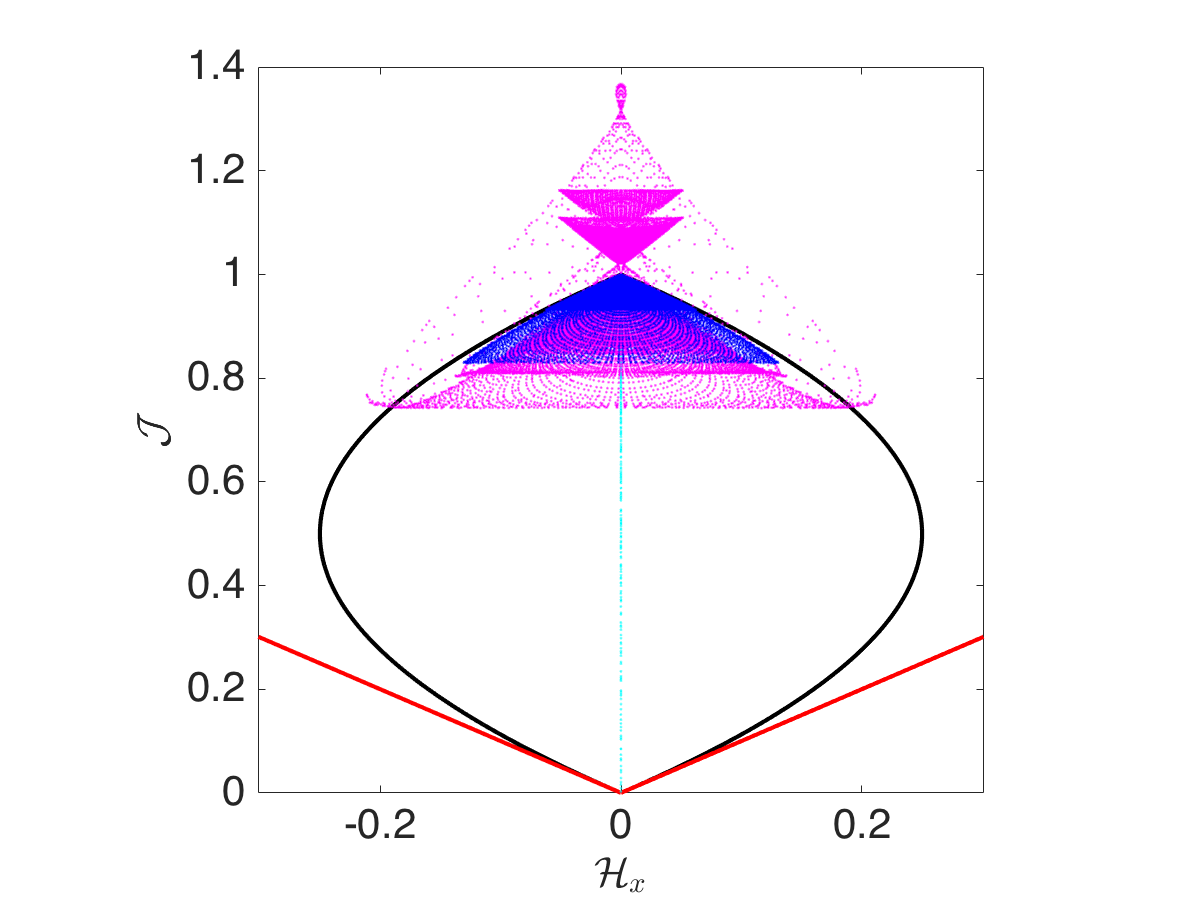} &
    \includegraphics[width=0.495\textwidth]{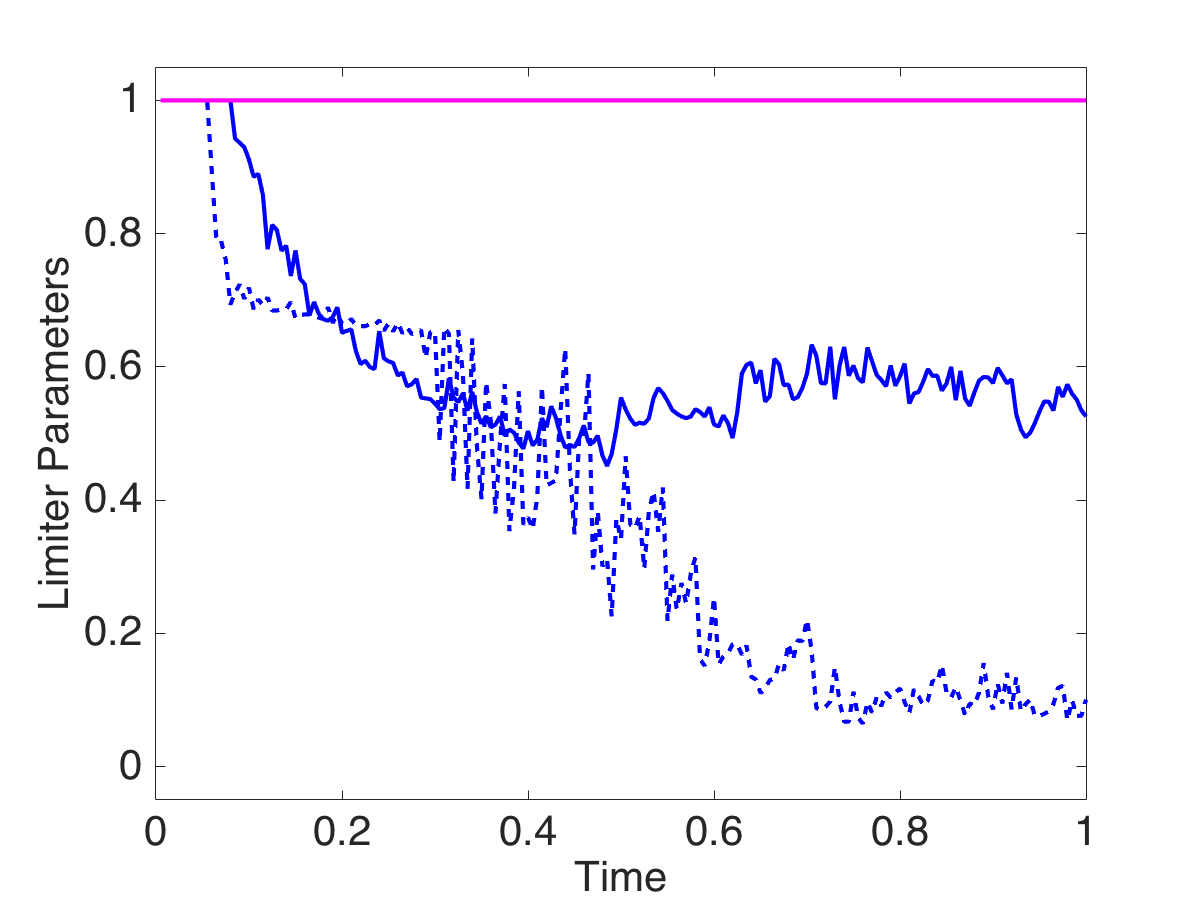}
  \end{tabular}
   \caption{Numerical results for the Fermion Implosion problem, computed both the CB and Minerbo closures.  Spatial distribution of the number density $\cJ$ (CB closure only) at $t=1$ (upper left panel).  In the upper right panel we plot the number density $\cJ$ versus radius $R=|\vect{x}|$ for various times ($t=0$, $0.1$, $0.2$, and $0.4$) for the CB (blue), the Minerbo (magenta) closures, and the reference transport solution (dashed black lines).  (The initial condition, which is the same for all models, is plotted with cyan.)  Numerical solutions in the $(\cH_{x},\cJ)$-plane (lower left panel), for the same times as plotted in the upper right panel are plotted for CB and Minerbo.  Limiter parameters $\vartheta_{1}$ (solid) and $\vartheta_{2}$ (dashed) in Section~\ref{sec:limiter} (minimum over the whole computational domain) versus time (lower right panel).}
  \label{fig:Implosion}
\end{figure}

Numerical results for the Fermion Implosion problem are plotted in Figure~\ref{fig:Implosion}.  
For $t>0$, the low-density region in the center of the computational domain is quickly filled in, and a cylindrical perturbation propagates radially away from the center.  
For the model with the CB closure, this perturbation, seen as a depression in the density relative to the ambient medium, has reached $R\approx1$ for $t=1$ (upper left panel in Figure~\ref{fig:Implosion}).  
The right panel in Figure~\ref{fig:Implosion} illustrates the difference in dynamics resulting from the two closures.  
(We also plot a reference transport solution obtained using the filtered spherical harmonics scheme described in \cite{garrettHauck_2013}; dashed black lines.\footnote{Kindly provided by Dr. Ming Tse Paul Laiu (private communications).})  
With the CB closure (blue lines), the central density increases towards the maximum value of unity, and an low-density pulse propagates radially.  
The amplitude of the pulse decreases with time due to the geometry of the problem.  
For $t=0.4$, the peak depression in located around $R=0.34$ (with $\cJ\approx0.96$).  
With the Minerbo closure, the central density continues to increase beyond unity, and reaches a maximum of about $\cJ\approx1.37$ at $t=0.1$.  
The central density starts to decrease beyond this point in time, and a steepening pulse propagates radially away from the center.  
(This pulse is trailing the pulse in the model computed with the CB closure.)  
At $t=0.4$, a discontinuity appears to have formed around $R=0.2$, resulting in numerical oscillations.  
Except for the realizability-enforcing limiter (which is not triggered for this model), no other limiters are used to prevent numerical oscillations.  
Although the solutions obtained with the two-moment model differ from the reference transport solution, the results obtained with the CB closure are in closer agreement with the transport solution.  
This is likely because the CB closure is consistent with the bound $f<1$ satisfied by the transport solution in this test.  
(For tests involving lower occupancies, the CB and Minerbo closures are expected to perform similarly.)  
In the lower left panel in Figure~\ref{fig:Implosion}, the moments are plotted in the $(\cH_{x},\cJ)$-plane for the same times as plotted in the upper left panel.  
(Each dot represents the moments at a specific spatial point and time.)  
Initially, $\vect{\cH}=0$, and all the moments lie on the line connecting $(0,0)$ and $(0,1)$; cyan points.  
With the CB closure (blue points), the moments are confined to evolve inside the realizable domain $\cR$ (black), while with the Minerbo closure, the moments are not confined to $\cR$, but to the region above the red lines (the realizable domain for moments of positive distribution functions), and this is the reason for the difference in dynamics in the two models.  
For the model with the CB closure, some moments evolve very close to the boundary of the realizable domain, and the positivity limiter is continuously triggered to damp these moments towards the cell average, which is realizable by the design of the numerical scheme.  
In the lower right panel in Figure~\ref{fig:Implosion} we plot the limiter parameters $\vartheta_{1}$ (solid) and $\vartheta_{2}$ (dashed) (cf. \eqref{eq:limitDensity} and \eqref{eq:limitMoments}) versus time for the CB closure model (blue) and the Minerbo closure model (magenta); the minimum over the whole computational domain is plotted.  
For the CB closure model, the limiter is triggered to prevent both density overshoots and $\gamma(\cM)<0$.  
Late in the simulation ($t\gtrsim0.7$), the minimum value of $\vartheta_{2}$ is around $0.1$.  
For the model using the Minerbo closure, the limiter is not triggered ($\vartheta_{1}=\vartheta_{2}=1$).  

\subsection{Homogeneous Sphere}
\label{sec:homogeneousSphere}

The homogeneous sphere test (e.g., \cite{smit_etal_1997}) considers a sphere with radius $R$.  
Inside the sphere (radius $<R$), the absorption opacity $\sigma_{\Ab}$ and the equilibrium distribution function $f_{0}$ are set to constant values.  
The scattering opacity $\sigma_{\Scatt}$ is set to zero in this test (i.e., $\xi=1$).  
Outside the sphere, the absorption opacity is zero.  
The steady state solution, obtained by solving the transport equation in spherical symmetry, is given by
\begin{equation}
  f_{\mbox{\tiny A}}(r,\mu)=f_{0}\,\big(1-e^{-\chi_{0}\,s(r,\mu)}\big),
  \label{eq:distributionHomogeneousSphere}
\end{equation}
where $r=|\vect{x}|$, 
\begin{equation}
  s(r,\mu)
  =\left\{
  \begin{array}{lll}
    r\,\mu+R\,g(r,\mu) & \mbox{if}\quad r<R, & \mu\in[-1,+1], \\
    2\,R\,g(r,\mu) & \mbox{if}\quad r \ge R, & \mu\in[(1-(R/r)^{2})^{1/2},+1], \\
    0 & \mbox{otherwise},
  \end{array}
  \right.
\end{equation}
and $g(r,\mu)=[1-(r/R)^{2}(1-\mu^{2})]^{1/2}$.  
Thus, $f_{\mbox{\tiny A}}(r,\mu)\in(0,f_{0})~\forall~r,\mu$.  

Here, this test is computed using a three-dimensional Cartesian domain $D=\{\vect{x}\in\bbR^{3}:x^{1}\in[0,2], x^{2}\in[0,2], x^{3}\in[0,2]\}$.  
Because of the symmetry of the problem, and to save computational resources, we only compute the solution in one octant.  
On the inner boundaries, we impose reflecting boundary conditions, while we impose 'homogeneous' boundary conditions on the outer boundary in all three coordinate dimensions; i.e., values for all moments in a boundary element are set equal to the corresponding values in the nearest element just inside $D$.  
Since this test is computed with Cartesian coordinates using a relatively low spatial resolution ($64^{3}$), we have found it necessary to smooth out the opacity over a finite radial extent to avoid numerical artifacts due to a discontinuous absorption opacity.  
Specifically, we use an absorption opacity of the following form
\begin{equation}
  \sigma_{\Ab}(r)=\f{\sigma_{\Ab,0}}{(r/R_{0})^{p}+1}.  
\end{equation}
We set $f_{0}=1$, and compute three versions of this test: one with $\sigma_{\Ab,0}=1$, $R_{0}=1$, and $p=80$ (Test~A), one with $\sigma_{\Ab,0}=10$, $R_{0}=1$, and $p=80$ (Test~B), and one with $\sigma_{\Ab}=10^{3}$, $R_{0}=0.85$, and $p=40$ (Test~C).  
(These values for $R_{0}$ and $p$ result in similar radius for where the optical depth equals $2/3$ in Test~B and Test~C.)
We compute until $t=5$, when the system has reached an approximate steady state.  
In all the tests, we use the IMEX scheme PD-ARS with $\dt=0.1\times\dx^{1}$ --- the least compute-intensive of the convex-invariant IMEX schemes presented here.  
The main purpose of this test is to compare the results obtained using the different algebraic closures discussed in Section~\ref{sec:algebraicClosure}.  

In Figure~\ref{fig:HomogeneousSphere}, we plot results obtained for all tests at $t=5$: Test~A (top panels), Test~B (middle panels), and Test~C (bottom panels).  
The particle density $\cJ$ and the flux factor $h=|\vect{\cH}|/\cJ$ (left and right panels, respectively) are plotted versus radius $r=|\vect{x}|$.  
In each panel, results obtained with the various algebraic closures discussed in Section~\ref{sec:algebraicClosure} are plotted: Minerbo (magenta), CB (blue), BL (green), and Kershaw (cyan).  
The analytical solution is also plotted (dashed black lines).  
\begin{figure}[H]
  \centering
  \begin{tabular}{cc}
    \includegraphics[width=0.5\textwidth]{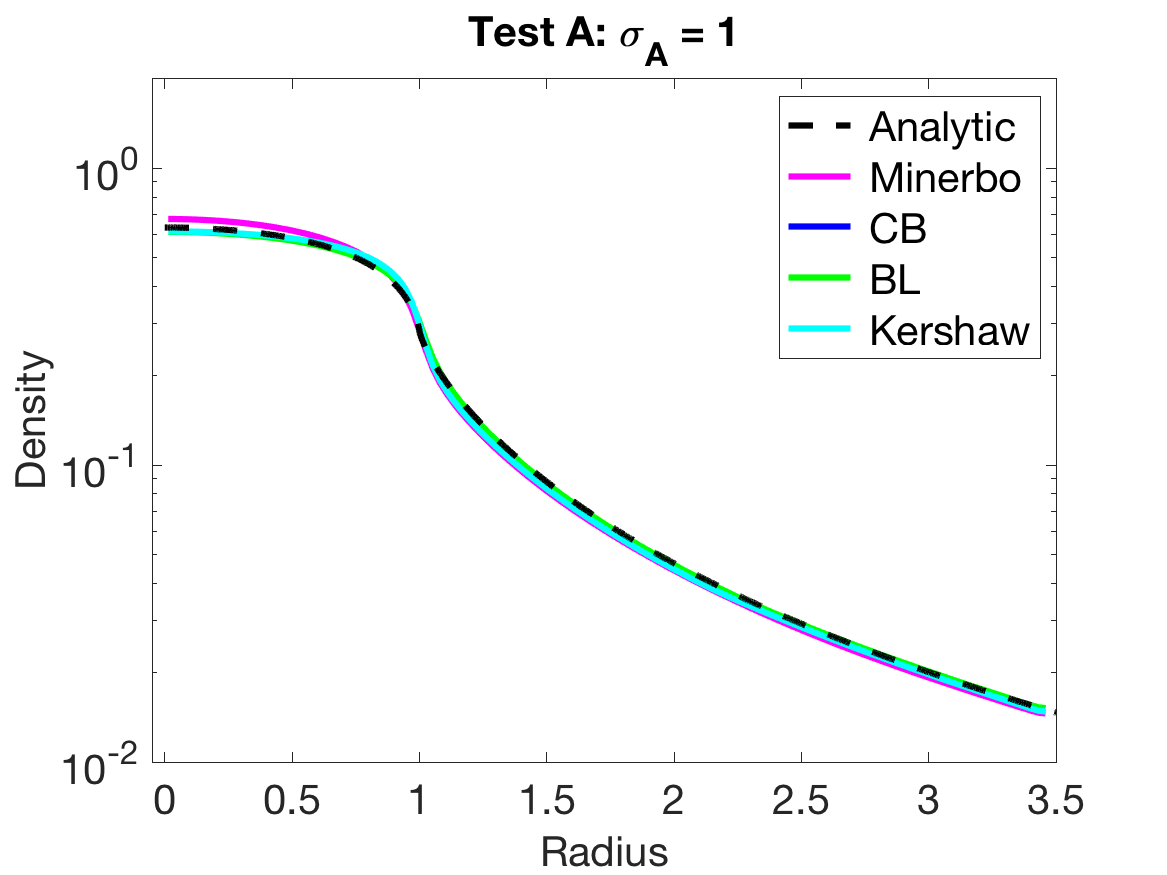}
    \includegraphics[width=0.5\textwidth]{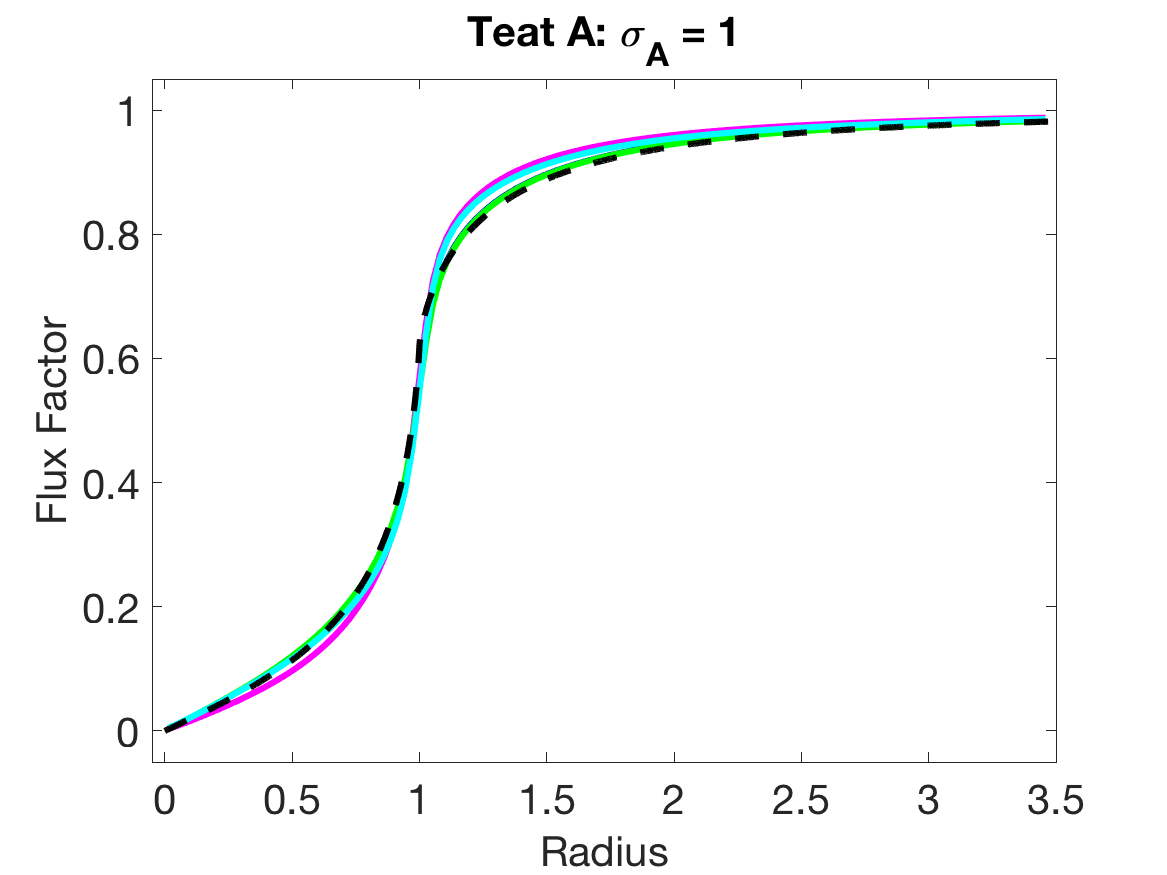} \\
    \includegraphics[width=0.5\textwidth]{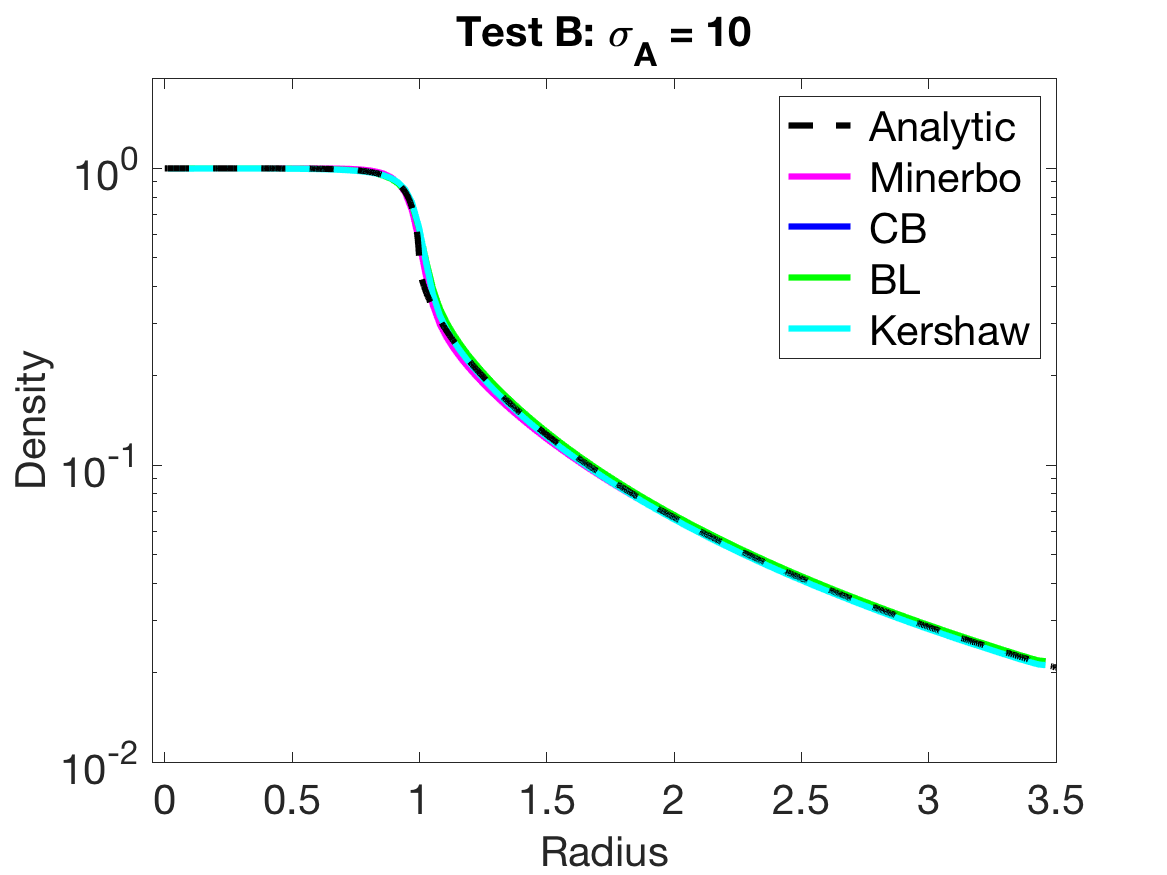}
    \includegraphics[width=0.5\textwidth]{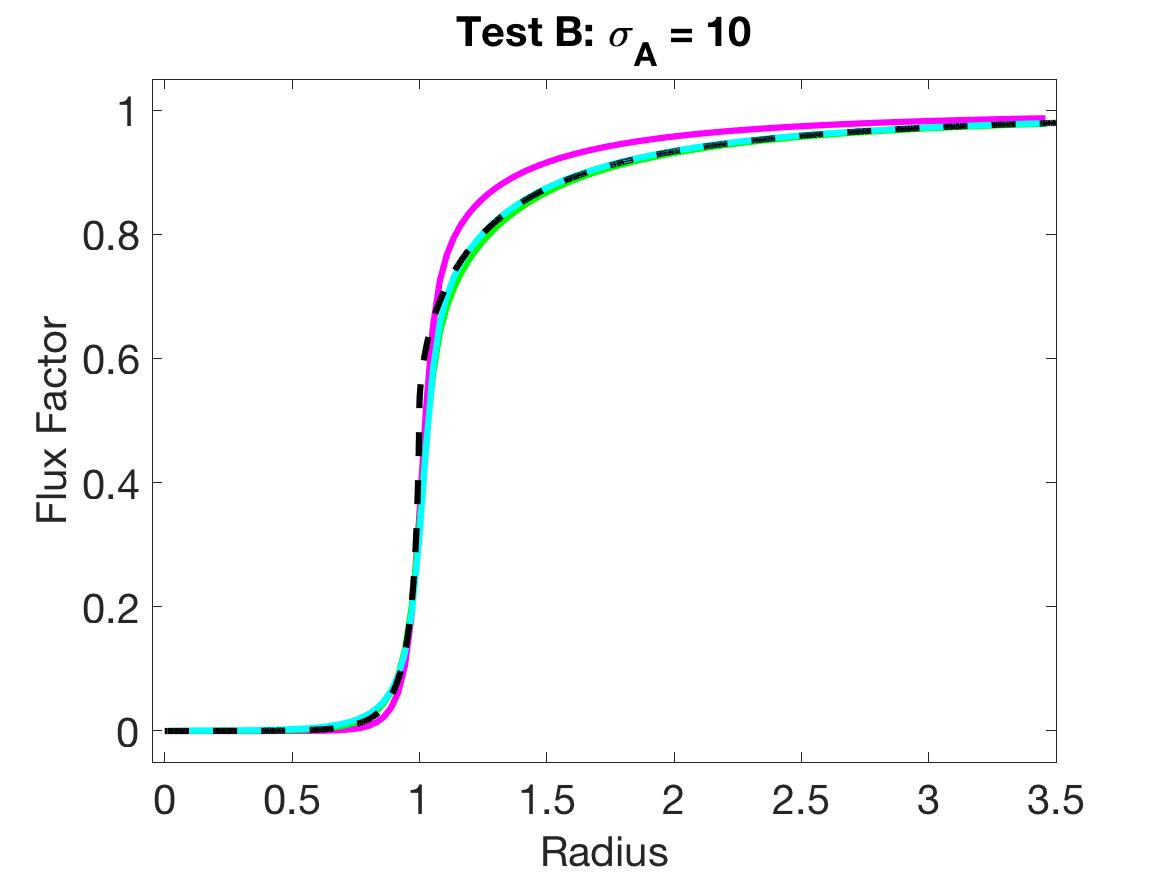} \\
    \includegraphics[width=0.5\textwidth]{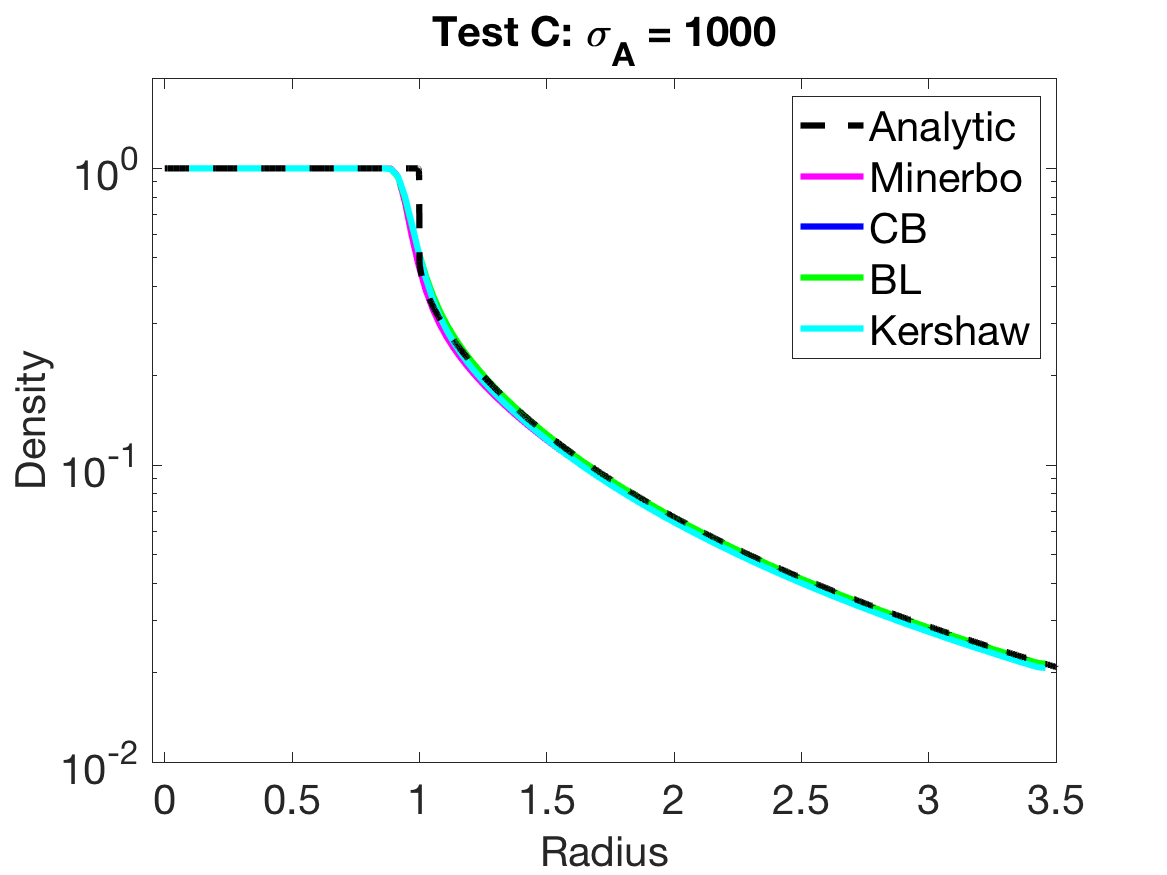}
    \includegraphics[width=0.5\textwidth]{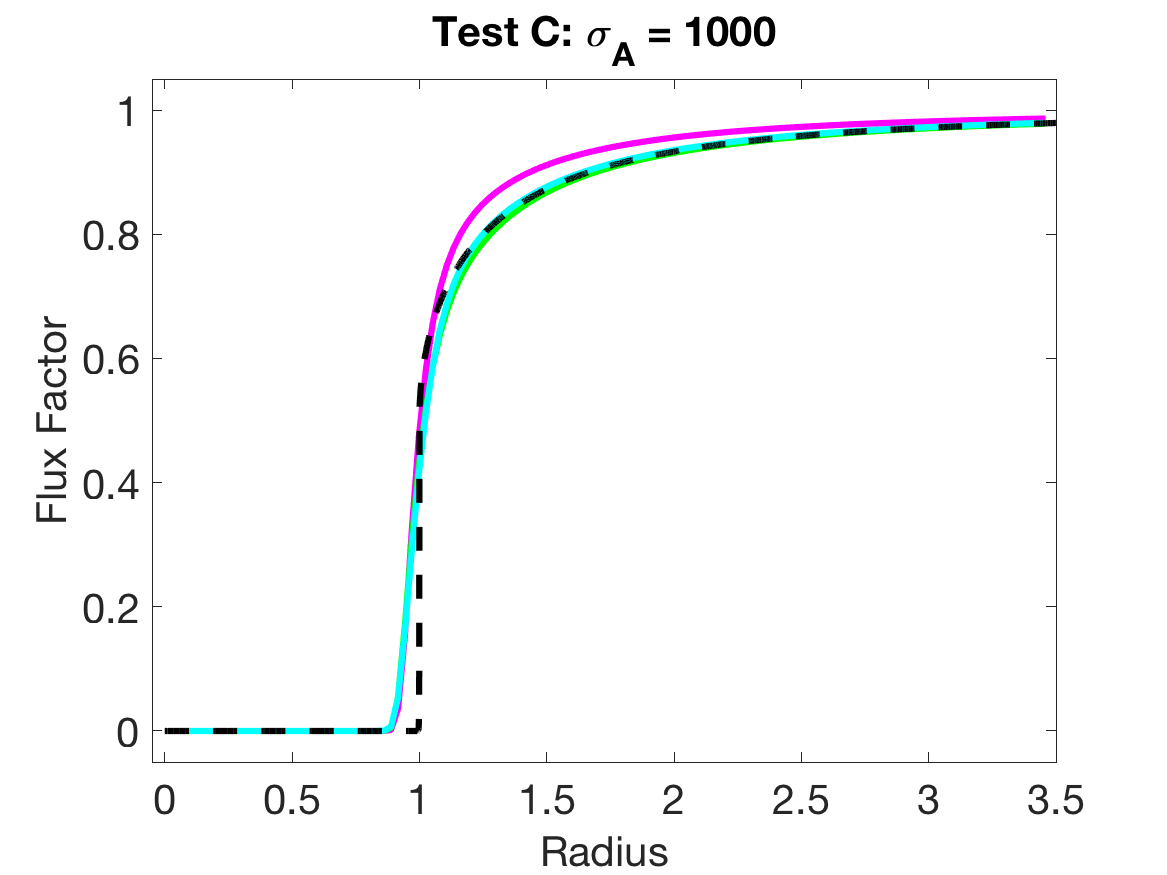}
  \end{tabular}
   \caption{Results obtained for the homogeneous sphere problem with the two-moment model for different values of the absorption opacity $\sigma_{\Ab,0}$: $1$ (top panels), $10$ (middle panels), and $1000$ (bottom panels).  The particle density (left panels) and the flux factor (right panels) are plotted versus radius.  Numerical results obtained with the algebraic closures of Minerbo (magenta), CB (blue), BL (green), and Kershaw (cyan) are compared with the analytic solution (dashed black lines).}
  \label{fig:HomogeneousSphere}
\end{figure}

We find good overall agreement between the results obtained with the two-moment model and the analytical solution.  
Partly due to the smoothing of the absorption opacity around the surface, the numerical and analytical solutions naturally differ around $r=1$.  
Aside from some differences discussed in more detail below, the numerical and analytical solutions --- for all values of the absorption opacity $\sigma_{\Ab,0}$ and all closures --- agree well as $r$ tends to zero, as well as when $r\gg1$.  
The results obtained with the maximum entropy closures CB and BL are practically indistinguishable on the plots.  
This is consistent with the similarity of the Eddington factors for these two closures, as shown in Figure~\ref{fig:EddingtonFactorsWithDifferentClosure}.  
We also find that the results obtained with the fermionic Kershaw closure agree well with the maximum entropy closures based on Fermi-Dirac statistics (CB and BL).  
From the plots of the particle density (left panels in Figure~\ref{fig:HomogeneousSphere}), the results obtained with all the closures, including Minerbo, appear very similar.  
(For Test~A, the particle density obtained with the Minerbo closure deviates the most from the analytic solution inside $r\approx0.75$; upper left panel).  
From the plots of the flux factor (right panels in Figure~\ref{fig:HomogeneousSphere}), it is evident that the results obtained with the Minerbo closure --- the only closure not based on Fermo-Dirac statistics --- deviates the most from the analytic solution outside $r=1$, where the flux factor is consistently higher than the analytical solution for all values of $\sigma_{\Ab,0}$.  
The fermionic closures (CB, BL, and Kershaw) track the analytic solution better.  
Similar agreement between the numerical and analytical solutions was reported by Smit et al. \cite{smit_etal_1997}, when using the CB maximum entropy closure with $f_{0}=0.8$ and an unsmoothed absorption opacity $\sigma_{\Ab}=4$.  
We also note that our results appear to be somewhat at odds with the results recently reported by Murchikova et al. \cite{murchikova_etal_2017}, who compared results obtained with the two-moment model using a large number of algebraic closures for this same problem (albeit using an unsmoothed and slightly different value for the absorption opacity).  
Murchikova et al. do not plot the particle density, but find essentially no difference in the flux factor and the Eddington factor when comparing results obtained with the maximum entropy closures of Minerbo and Cernohorsky \& Bludman (CB).  

In Figure~\ref{fig:HomogeneousSphereRealizability}, we further compare the results obtained when using the Minerbo and CB closures by plotting the solutions to the homogeneous sphere problem for Test~C at $t=5$ in the $(|\vect{\cH}|,\cJ)$-plane (cf. the realizable domain in Figure~\ref{fig:RealizableSetFermionic}).  
The numerical solution at each spatial point is represented by a blue (CB) or magenta (Minerbo) dot in the panels.  
In the lower two panel we zoom in on the results obtained with the two closures around the top and lower right regions of the realizable domain (lower left and lower right panel, respectively; cf. green boxes in the upper right panel).  
\begin{figure}[H]
  \centering
  \includegraphics[width=1.0\textwidth]{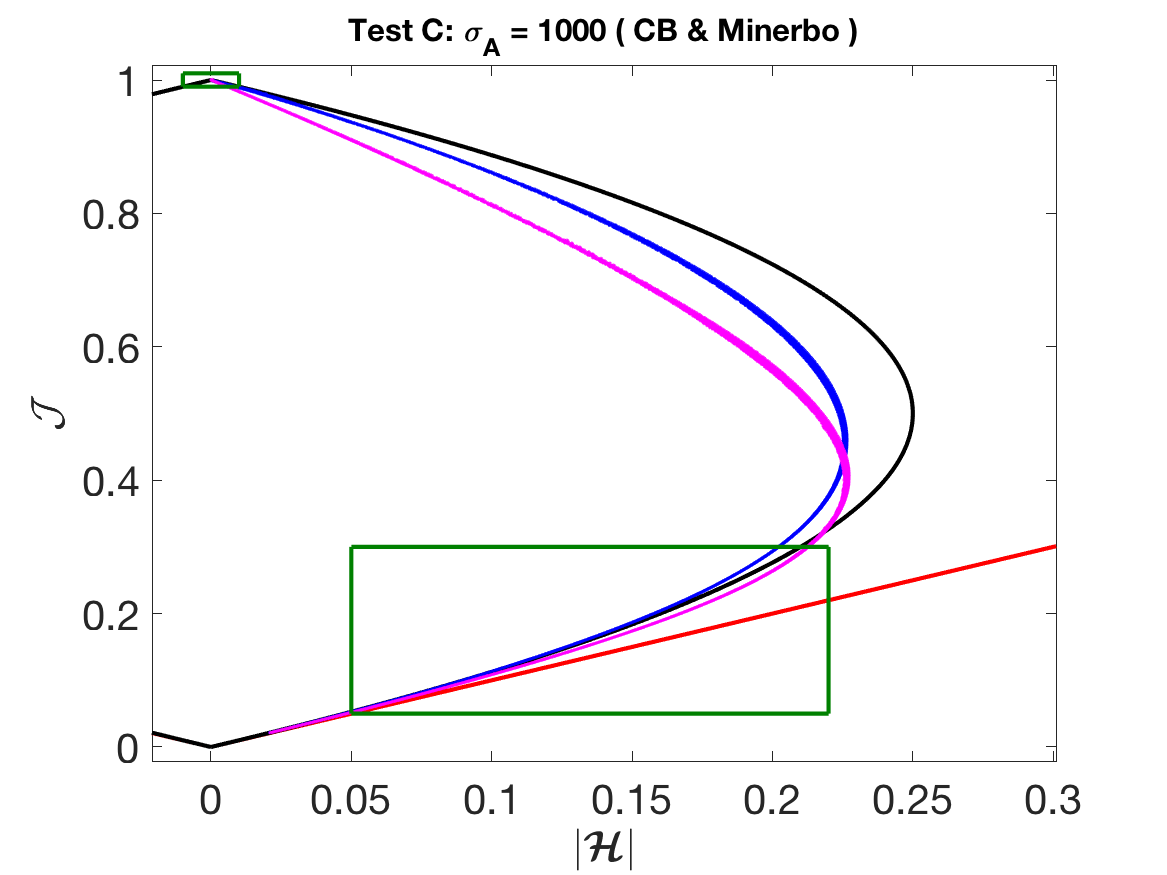}
  \begin{tabular}{cc}
    \includegraphics[width=0.5\textwidth]{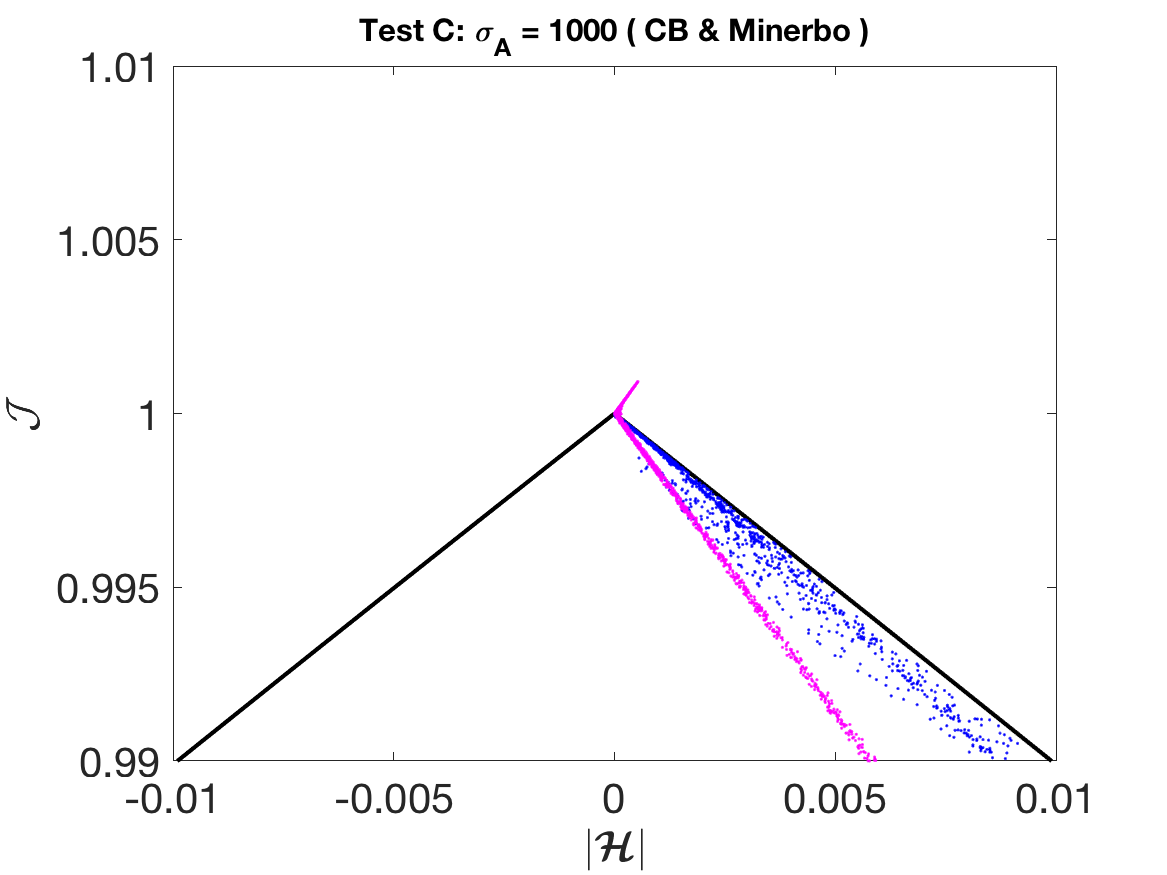}
    \includegraphics[width=0.5\textwidth]{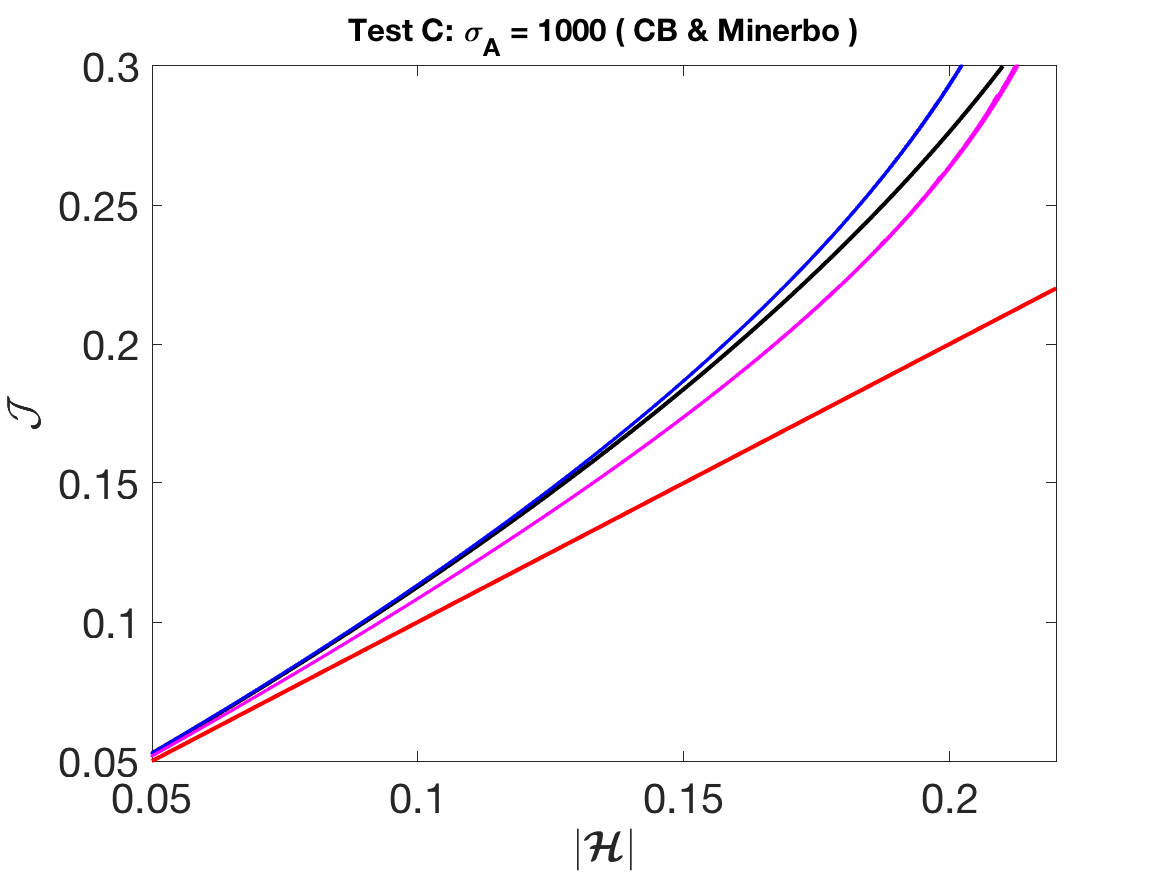}
  \end{tabular}
   \caption{Scatter plots of the numerical solution to the homogeneous sphere problem for Test~C in the $(|\vect{\cH}|,\cJ)$-plane.  Results obtained with the CB and Minerbo closures are plotted in the upper panel, blue and magenta points, respectively.  Zoom-ins on the solutions obtained with the two closures are plotted in the lower two panels (cf. green boxes in the upper panel).  The boundaries of the realizable domains $\cR$ and $\cR^{+}$ are indicated with solid black and solid red curves, respectively.  See text for further details.  }
  \label{fig:HomogeneousSphereRealizability}
\end{figure}
As can be seen in the upper panel in Figure~\ref{fig:HomogeneousSphereRealizability}, the solutions to the homogeneous sphere problem obtained with the two closures trace out distinct curves relative to the realizable domain $\cR$, whose boundary is indicated by solid black curves in each panel.  
When using the CB closure, the realizability-preserving DG-IMEX scheme developed here maintains solutions within $\cR$.  
When using the Minerbo closure, the appropriate realizable domain is given by $\cR^{+}$ (cf. Eq.~\eqref{eq:realizableSetPositive}), whose boundary is indicated by solid red lines in Figure~\ref{fig:HomogeneousSphereRealizability}, and we find that the numerical solution ventures outside $\cR$.  
Near the surface around $r=1$, the number density slightly exceeds unity (lower left panel), while for larger radii, the computed flux may exceed the value allowed by Fermi-Dirac statistics (lower right panel).  
\section{Summary and Conclusions}
\label{sec:conclusions}

We have developed a realizability-preserving DG-IMEX scheme for a two-moment model of fermion transport.  
The scheme employs algebraic closures based on Fermi-Dirac statistics and combines a time step restriction (CFL condition), a realizability-enforcing limiter, and a convex-invariant time integrator to maintain point-wise realizability of the moments.  
Since the realizable domain is a convex set, the realizability-preserving property is obtained from convexity arguments, building on the framework in \cite{zhangShu_2010a}.  

In the applications motivating this work, the collision term is stiff in regions of the computational domain, and we have considered IMEX schemes to avoid treating the transport operator implicitly.  
We have considered two recently proposed second-order accurate, convex-invariant IMEX schemes \cite{chertock_etal_2015,hu_etal_2018}, that restore second-order accuracy with an implicit correction step.  
However, we are unable to prove realizability (without invoking a very small time step) with the approach in \cite{chertock_etal_2015}, and we have demonstrated that the approach in \cite{hu_etal_2018} does not perform well in the diffusion limit.  
For these reasons, we have resorted to first-order, convex-invariant IMEX schemes.  
While the proposed scheme (dubbed PD-ARS) is formally only first-order accurate, it works well in the diffusion limit, is convex-invariant with a reasonable time step, and reduces to the optimal second-order accurate explicit SSP-RK scheme in the streaming limit.  

For each stage of the IMEX scheme, the update of the cell-averaged moments can be written as a convex combination of forward Euler steps (implying the Shu-Osher form for the explicit part), followed by a backward Euler step.  
Realizability of the cell-averaged moments due to the explicit part requires the DG solution to be realizable in a finite number of quadrature points in each element and the time step to satisfy a CFL condition.  
For the backward Euler step, realizability of the cell-averages follows easily from the simple form of the collision operator (which includes emission, absorption, and isotropic scattering without energy exchange), and is independent of the time step.  
The CFL condition is then solely due to the transport operator, and the time step can be as large as that of the forward Euler scheme applied to the explicit part of the cell-average update.  
After each stage update, the limiter enforces moment realizability point-wise by damping towards the realizable cell average.  
Numerical experiments are presented to demonstrate the accuracy and realizability-preserving property of the DG-IMEX scheme.  
The applicability of the PD-ARS scheme is not restricted to the fermionic two-moment model.  
It may therefore be a useful option in other applications of kinetic theory where physical constraints confine solutions to a convex set and capturing the diffusion limit is important.  

Realizability of the fermionic two-moment model depends sensitively on the closure procedure.  
For the algebraic closures adapted in this work, realizability of the scheme demands that lower and upper bounds on the Eddington factor are satisfied \cite{levermore_1984,lareckiBanach_2011}.  
The Eddington factors deriving from the maximum entropy closures of Cernohorsky \& Bludman \cite{cernohorskyBludman_1994} and Larecki \& Banach \cite{lareckiBanach_2011}, and the Kershaw-type closure of Larecki \& Banach \cite{banachLarecki_2017a} all satisfy these bounds and are suitable for the fermionic two-moment model.  
Further approximations of the closure procedure (e.g., employing the low occupancy limit, which results in the Minerbo closure \cite{minerbo_1978} when starting with the maximum entropy closure of \cite{cernohorskyBludman_1994}) is not compatible with realizability of the fermionic two-moment model, and we caution against this approach to modeling particle systems governed by Fermi-Dirac statistics; particularly if the low occupancy approximation is unlikely to hold (e.g., when modeling neutrino transport in core-collapse supernovae).  

In this work, we started with a relatively simple kinetic model.  
In particular, we adopted Cartesian coordinates, and assumed a linear collision operator and a fixed material background.  
Scattering with energy exchange and relativistic effects (e.g., due to a moving material and the presence of a strong gravitational field) were not included.  
To solve more realistic problems of scientific interest, some or all of these physical effects will have to be included.  
In the context of developing realizability-preserving schemes, these extensions will provide significant challenges suitable for future investigations, for which the scheme presented here may serve as a foundation.  
\appendix

\section{Butcher Tableau for IMEX Schemes}
\label{app:butcherTables}

For easy reference, we include the Butcher tableau for the IMEX schemes considered in this paper, which can be written in the standard double Butcher tableau
\begin{equation}
  \begin{array}{c | c}
    \tilde{\vect{c}} & \tilde{A} \\
    \hline
    & \tilde{\vect{w}}^{T}
  \end{array}
  \qquad
  \begin{array}{c | c}
    \vect{c} & A \\
    \hline
    \alpha & \vect{w}^{T}
  \end{array}.  
  \label{eq:butcher}
\end{equation}
The \emph{explicit tableau} (left; components adorned with a tilde) represents the explicit part of the IMEX scheme, and the \emph{implicit tableau} (right; unadorned components) represents the implicit part of the IMEX scheme.  
For $s$ stages, $\tilde{A}=(\tilde{a}_{ij})$, $\tilde{a}_{ij}=0$ for $j\ge i$, and $A=(a_{ij})$, $a_{ij}=0$ for $j>i$, are $s\times s$ matrices, and $\tilde{\vect{w}}=(\tilde{w}_{1},\ldots,\tilde{w}_{s})^{T}$ and $\vect{w}=(w_{1},\ldots,w_{s})^{T}$.  
The vectors $\tilde{\vect{c}}=(\tilde{c}_{1},\ldots,\tilde{c}_{s})^{T}$ and $\vect{c}=(c_{1},\ldots,c_{s})^{T}$, used for non autonomous systems, satisfy $\tilde{c}_{i}=\sum_{j=1}^{i-1}\tilde{a}_{ij}$ and $c_{i}=\sum_{j=1}^{i}a_{ij}$.  
For the implicit tableau, we have included the scalar $\alpha$, used for the correction step in Eq.~\eqref{eq:imexCorrection}.  

For the analysis of convex-invariant IMEX schemes, additional coefficients are defined \cite{hu_etal_2018} (cf. Eq.~\eqref{eq:imexStagesRewrite}).  
First, let
\begin{equation}
  b_{ii} = \f{1}{a_{ii}}, \quad
  b_{ij} = -\f{1}{a_{ii}}\sum_{l=j}^{i-1}a_{il}b_{lj}, \quad
  \tilde{b}_{ij} = -\f{1}{a_{ii}}\Big(\tilde{a}_{ij}+\sum_{l=j+1}^{i-1}a_{il}\tilde{b}_{lj}\Big).  
\end{equation}
Then, for IMEX schemes of Type~A \cite{dimarcoPareschi2013},
\begin{equation}
  \begin{aligned}
    c_{i0} &= 1-\sum_{j=1}^{i-1}\sum_{l=j}^{i-1}a_{il}b_{lj}, \quad &
    c_{ij} &= \sum_{l=j}^{i-1}a_{il}b_{lj}, \\
    \tilde{c}_{i0} &= 0, \quad &
    \tilde{c}_{ij} &= \tilde{a}_{ij} + \sum_{l=j+1}^{i-1}a_{il}\tilde{b}_{lj};
  \end{aligned}
  \label{eq:positivityCoefficientsA}
\end{equation}
for IMEX schemes of Type~ARS \cite{ascher_etal_1997},
\begin{equation}
  \begin{aligned}
    c_{i0} &= 1-\sum_{j=2}^{i-1}\sum_{l=j}^{i-1}a_{il}b_{lj}, \quad &
    c_{ij} &= \sum_{l=j}^{i-1}a_{il}b_{lj} \\
    \tilde{c}_{i0} &= \tilde{a}_{i1}+\sum_{j=2}^{i-1}a_{ij}\tilde{b}_{j1}, \quad &
    \tilde{c}_{ij} &= \tilde{a}_{ij}+\sum_{l=j+1}^{i-1}a_{il}\tilde{b}_{lj}.  
  \end{aligned}
  \label{eq:positivityCoefficientsARS}
\end{equation}
Note that $c_{i1}=\tilde{c}_{i1}=0$ in Eq.~\eqref{eq:positivityCoefficientsARS} so that $\sum_{j=0}^{i-1}c_{ij}=1$.  
Also note the difference between the matrix coefficients in Eqs.~\eqref{eq:positivityCoefficientsA} and \eqref{eq:positivityCoefficientsARS} and the vector components defined below Eq.~\eqref{eq:butcher}.  

\paragraph{IMEX PA2}

A second-order accurate, convex-invariant IMEX scheme of type $A$ (the matrix $A$ is invertible) with four implicit solves was given in \cite{hu_etal_2018}.  
We refer to this scheme as IMEX PA2.  
For this scheme, the non-zero components of $\tilde{A}$ and $A$ are given by
\begin{align*}
  \tilde{a}_{21} &= 0.7369502715, \\
  \tilde{a}_{31} &= 0.3215281691, \quad \tilde{a}_{32} = 0.6784718309, \\
  a_{11} &= 0.6286351712, \\
  a_{21} &= 0.2431004655, \quad a_{22} = 0.1959392570, \\
  a_{31} &= 0.4803651051, \quad a_{32} = 0.0746432814, \quad a_{33} = 0.4449916135. 
\end{align*}
The coefficient in the correction step is $\alpha = 0.2797373792$ and the CFL constant is $c_{\Sch} = 0.5247457524$.
This scheme is globally stiffly accurate (GSA), so that $\tilde{w}_{i}=\tilde{a}_{3i}$ and $w_{i}=a_{3i}$ for $i\le3$.

\paragraph{IMEX PA2+}

We have found another second-order accurate, convex-invariant IMEX scheme of type $A$ with four implicit solves, which we refer to as IMEX PA2+.  
This scheme allows for a larger value of $c_{\Sch}$ than IMEX PA2 (i.e., a larger time step while maintaining admissible solutions).  
The scheme was found by random sampling of the parameter space spanned by the IMEX coefficients and selecting the scheme with the largest $c_{\Sch}$.  
For IMEX PA2+, $c_{\Sch} = 0.895041066934$. 
The non-zero components of $\tilde{A}$ and $A$ are given by
\begin{align*}
  \tilde{a}_{21} &= 0.909090909090909, \\
  \tilde{a}_{31} &= 0.450000000000000, \quad \tilde{a}_{32} = 0.550000000000000, \\
  a_{11} &= 0.521932391842510, \\
  a_{21} &= 0.479820781424967, \quad a_{22} = 0.002234534340252, \\
  a_{31} &= 0.499900000000000, \quad a_{32} = 0.001100000000000, \quad a_{33} = 0.499000000000000.
\end{align*}
The coefficient in the correction step is $\alpha = 0.260444263529413$.  
This scheme is also GSA; $\tilde{w}_{i}=\tilde{a}_{3i}$ and $w_{i}=a_{3i}$ for $i\le3$.  

The rest of the IMEX schemes we consider here do not include the correction step in Eq.~\eqref{eq:imexCorrection}; i.e., $\alpha=0$.  

\paragraph{IMEX PC2}

Another IMEX scheme was given in \cite{mcclarren_etal_2008} (referred to there as a semi-implicit predictor-corrector method).  
This scheme has two implicit solves and can be written in the double Butcher tableau form, and we refer to this scheme as IMEX PC2.  
The non-zero components of $\tilde{A}$ and $A$ are given by
\begin{align*}
  \tilde{a}_{21} &= 0.5, \quad \tilde{a}_{32} = 1, \\
  a_{22} &= 0.5, \quad a_{33} = 1.0,
\end{align*}
$\alpha=0$, and $\tilde{w}_{i}=\tilde{a}_{3i} = w_{i}=a_{3i}$ for $i\le3$.  
IMEX PC2 is not convex-invariant, since $c_{\Sch} = 0$ (cf. discussion in Section~\ref{sec:imex}).  

\paragraph{IMEX PD-ARS}

We have found a family of convex-invariant, diffusion accurate IMEX schemes of type ARS that are second-order accurate in the streaming limit, which we refer to as IMEX PD-ARS; see \ref{app:PD-ARS}.  
For these schemes, $c_{\Sch}= 1 - 2\epsilon$ with $\epsilon \in [0, 1/2)$.
Here we give an example by setting $\epsilon=0.1$:
\begin{align*}
  \tilde{a}_{21} & = 1.0, \\
  \tilde{a}_{31} & = 0.5, \quad \tilde{a}_{32} = 0.5, \\
  a_{22} & = 1.0, \nonumber \\
  a_{32} & = 0.4 \,( = 0.5 - \epsilon\,), \quad a_{33} = 0.6 \,( = 0.5 + \epsilon\,). 
\end{align*}
This scheme is GSA, $\alpha=0$, and requires two implicit solves per time step (same as IMEX PC2).  

\paragraph{IMEX RKCB2}

We compare the performance of the convex-invariant IMEX schemes with two other (not convex-invariant) IMEX schemes.  
The first one is the second-order accurate IMEX scheme given in \cite{cavaglieriBewley2015} with two implicit solves.  
We refer to this scheme as IMEX RKCB2.  
The non-zero components of $\tilde{A}$ and $A$ are given by
\begin{align*}
  \tilde{a}_{21} &= 2/5, \quad \tilde{a}_{32} = 1, \\
  a_{22} &= 2/5, \nonumber \\
  a_{32} &= 5/6, \quad a_{33} = 1/6,
\end{align*}
$\alpha=0$, and $w_{i} = a_{3i} = \tilde{w}_{i}$ (stiffly accurate \cite{pareschiRusso_2005}).

\paragraph{IMEX SSP2332}

Another scheme that we use for comparison is the second-order accurate IMEX scheme given in \cite{pareschiRusso_2005} with three implicit solves.  
We refer to this scheme as IMEX SSP2332.  
The non-zero components of $\tilde{A}$ and $A$ are given by
\begin{align*}
  \tilde{a}_{21} &= 1/2, \\
  \tilde{a}_{31} &= 1/2, \quad \tilde{a}_{32} = 1/2, \\
  a_{11} &= 1/4, \\
  a_{22} &= 1/4, \\
  a_{31} &= 1/3, \quad a_{32} = 1/3, \quad a_{33} = 1/3, 
\end{align*}
$\alpha=0$, and $w_{i} = a_{3i} = \tilde{w}_{i}$ (stiffly accurate).

\paragraph{SSPRK2 and SSPRK3}

To compare the performance of the IMEX schemes in the streaming limit (no collisions), we also compute results with explicit strong stability-preserving Runge-Kutta methods \cite{gottlieb_etal_2001}.  
(All elements of the implicit Butcher tableau are zero.)  
The optimal second-order accurate, strong-stability-preserving Runge-Kutta scheme (SSPRK2) has the following non-zero components:
\begin{align}
  \tilde{a}_{21} &= 1, \nonumber \\ 
  \tilde{w}_{1}  &= 1/2, \quad \tilde{w}_{2} = 1/2. \nonumber 
\end{align}
The optimal third-order accurate, strong-stability-preserving Runge-Kutta scheme (SSPRK3) has the following non-zero components:
\begin{align}
  \tilde{a}_{21} &= 1, \nonumber \\
  \tilde{a}_{31} &= 1/4, \quad \tilde{a}_{32} = 1/4, \nonumber \\
  \tilde{w}_{1} &= 1/6, \quad \tilde{w}_{2} = 1/6, \quad \tilde{w}_{3} =2/3. \nonumber
\end{align}

\section{Construction of IMEX Scheme PD-ARS}
\label{app:PD-ARS}

Here we construct a three-stage PD-IMEX scheme of Type~ARS, conforming to Definition~\ref{def:PD-IMEX}.  
We refer to the resulting IMEX scheme as PD-ARS.  
For a 3-stage scheme, the double Butcher tableau is
\begin{equation}
  \begin{array}{c | c c c}
  	         0           & 0                 & 0                   & 0  \\
  	\tilde{c}_{2} & \tilde{a}_{21} & 0                   & 0  \\
  	\tilde{c}_{3} & \tilde{a}_{31} & \tilde{a}_{32} & 0  \\ \hline
  	                   & \tilde{a}_{31} & \tilde{a}_{32} & 0
  \end{array}
  \qquad
  \begin{array}{c | c c c}
  	     0  & 0  & 0         & 0          \\
  	c_{2} & 0 & a_{22} & 0          \\
  	c_{3} & 0 & a_{32} & a_{33}  \\ \hline
  	         & 0 & a_{32} & a_{33}
  \end{array}
\end{equation}
The problem is then to find the coefficients $\{ \tilde{a}_{21}, \tilde{a}_{31}, \tilde{a}_{32}, a_{22}, a_{32}, a_{33} \}$ satisfying the constraints in Definition~\ref{def:PD-IMEX} while maximizing
\begin{equation}
  c_{\Sch} = \min \Big\{\, \dfrac{c_{20}}{\tilde{c}_{20}},\, \dfrac{c_{30}}{\tilde{c}_{30}},\, \dfrac{c_{32}}{\tilde{c}_{32}} \,\Big\}.  
\end{equation}
By imposing the equality constraints (i.e., Eqs.~\eqref{eq:diffusionCondition}, \eqref{eq:orderConditionsEx}, and \eqref{eq:implicitConsistency}), the double Butcher tableau can be written in terms of two independent parameters ($x,y\in\bbR$) as
\begin{equation}
  \begin{array}{c | c c c}
  	     0       & 0            & 0 & 0 \\
  	\frac{1}{2x} & \frac{1}{2x} & 0 & 0 \\
  	     1       & 1-x          & x & 0 \\ \hline
  	             & 1-x          & x & 0
  \end{array}
  \qquad
  \begin{array}{c | c c c}
  	     0       & 0 & 0            & 0 \\
  	\frac{1}{2x} & 0 & \frac{1}{2x} & 0 \\
  	     1       & 0 & 1-y          & y \\ \hline
  	             & 0 & 1-y          & y
  \end{array}
\end{equation}
Computing the relevant coefficients in Eq.~\eqref{eq:positivityCoefficientsARS}, we find $c_{20}=1$, $c_{30}=1-2x(1-y)$, $c_{32}=2x(1-y)$, $\tilde{c}_{20}=\f{1}{2x}$, $\tilde{c}_{30}=(y-x)$, and $\tilde{c}_{32}=x$, so that
\begin{equation}
  c_{\Sch} = \min \Big\{\, 2x,\, \f{1-2x(1-y)}{y-x},\, 2(1-y) \,\Big\}.  
\end{equation}
The convex-invariant property requires imposing the inequality constraints $a_{22},a_{33}>0$, $c_{20},c_{30},c_{32}\ge0$, and $\tilde{c}_{20},\tilde{c}_{30},\tilde{c}_{32}\ge0$, which imply that
\begin{equation}
  0 < x \le y
  \quad\text{and}\quad
  0 < y \le 1.
\end{equation}
We chose $x=\f{1}{2}$, so that the explicit part of the IMEX scheme is equivalent to the optimal second-order SSP-RK scheme in \cite{gottlieb_etal_2001} (SSPRK2 in \ref{app:butcherTables}).  
Then, $y=\f{1}{2} + \epsilon$, where $\epsilon \in [0, \frac{1}{2})$, and $c_{\Sch} = 1 - 2\epsilon$ results in the PD-ARS IMEX scheme.  
Setting $\epsilon = 0$ gives the optimal scheme with $c_{\Sch} = 1$.  

\section{Nonexistence of Three-Stage PD-IMEX Scheme of Type~A}
\label{app:noTypeA}

Here we prove that a PD-IMEX scheme  of type~A (i.e., conforming to Definition~\ref{def:PD-IMEX}, but with Eq.~\eqref{eq:positivityConditionsTypeA} replacing Eq.~\eqref{eq:positivityConditionsTypeARS} in item 3) does not exist.  
First, for a three-stage, GSA IMEX scheme of type~A the double Butcher tableau is
\begin{equation*}
  \begin{array}{c | c c c}
  	      0       & 0              & 0              & 0 \\
  	\tilde{c}_{2} & \tilde{a}_{21} & 0              & 0 \\
  	\tilde{c}_{3} & \tilde{a}_{31} & \tilde{a}_{32} & 0 \\ \hline
  	              & \tilde{a}_{31} & \tilde{a}_{32} & 0
  \end{array}
  \qquad
  \begin{array}{c | c c c}
    c_{1} & a_{11} & 0      & 0      \\
  	c_{2} & a_{21} & a_{22} & 0      \\
  	c_{3} & a_{31} & a_{32} & a_{33} \\ \hline
  	      & a_{31} & a_{32} & a_{33}
  \end{array}.
\end{equation*}
First we consider the equality constraints.  
Consistency of the implicit coefficients and second-order accuracy in the streaming limit (Eqs.~\eqref{eq:implicitConsistency} and \eqref{eq:orderConditionsEx}, respectively) give
\begin{align}
  a_{31} + a_{32} + a_{33} = 1, \quad \tilde{a}_{31} + \tilde{a}_{32} = 1, \quad \text{and}\quad \tilde{a}_{32}\,\tilde{a}_{21} =  \f{1}{2}.
  \label{eq:1stOrderGeneral}
\end{align}
Accuracy in the diffusion limit (Eq.~\eqref{eq:diffusionCondition}) requires 
\begin{align}
  \f{\tilde{a}_{21}}{a_{22}} = 1
  \quad\text{and}\quad 
  -\f{a_{32}\,\tilde{a}_{21}}{a_{22}\,a_{33}} + \f{\tilde{a}_{31}+\tilde{a}_{32}}{a_{33}} = 1.
\label{eq:diffusionAccurate}
\end{align}
Eq.~\eqref{eq:diffusionAccurate} together with the second constraint in Eq~\eqref{eq:1stOrderGeneral} gives $a_{32}+a_{33}=1$, which together with the first constraint in Eq.~\eqref{eq:1stOrderGeneral} gives
\begin{align}
  a_{31} = 0.
  \label{eq:a31is0}
\end{align}
Next we consider the inequality constraints.  
The convex-invariant property in Eq.~\eqref{eq:positivityConditionsTypeA} requires $a_{11}, a_{22}, a_{33}>0$, and 
\begin{equation}
  c_{21} = \frac{a_{21}}{a_{11}} \geq 0, \quad
  c_{31} = \frac{a_{31}}{a_{11}} - \frac{a_{32}\,a_{21}}{a_{22}\,a_{11}} \geq 0, \quad\text{and}\quad
  c_{32} = \frac{a_{32}}{a_{22}} \geq 0.
  \label{eq:constraintsTypeA}
\end{equation}
As a consequence, $a_{21}, a_{32} \geq 0$.  
However, since $a_{31} = 0$, 
\begin{equation*}
  c_{31} = - \frac{a_{32}a_{21}}{a_{22}a_{11}} \le 0.  
\end{equation*}
Thus the inequality constraints in Eq.~\eqref{eq:constraintsTypeA} hold only for $c_{31} = 0$, which gives $c_{\Sch} = \min \Big\{\,\f{c_{21}}{\tilde{c}_{21}}, \f{c_{31}}{\tilde{c}_{31}}, \f{c_{32}}{\tilde{c}_{32}}\Big\} = 0$.
Therefore, a three-stage PD-IMEX scheme (Definition~\ref{def:PD-IMEX}) of type~A does not exist.

\section*{References}

\bibliography{./references/references.bib}

\begin{thebibliography}{10}
\expandafter\ifx\csname url\endcsname\relax
  \def\url#1{\texttt{#1}}\fi
\expandafter\ifx\csname urlprefix\endcsname\relax\def\urlprefix{URL }\fi
\expandafter\ifx\csname href\endcsname\relax
  \def\href#1#2{#2} \def\path#1{#1}\fi

\bibitem{zhangShu_2010a}
X.~{Zhang}, C.-W. {Shu}, {On maximum-principle-satisfying high order schemes
  for scalar conservation laws}, Journal of Computational Physics 229 (2010)
  3091--3120.

\bibitem{zhangShu_2010b}
X.~Zhang, C.-W. Shu, On positivity preserving high order discontinuous galerkin
  schemes for compressible euler equations on rectangular meshes, Journal of
  Computational Physics 229 (2010) 8918--8934.

\bibitem{cernohorskyBludman_1994}
J.~{Cernohorsky}, S.~A. {Bludman}, Maximum entropy distribution and closure for
  bose-einstein and {F}ermi-{D}irac radiation transport, Astrophysical Journal
  433~(1) (1994) 450 -- 455.

\bibitem{banachLarecki_2017a}
Z.~Banach, W.~Larecki, {Kershaw-type transport equations for fermionic
  radiation}, Zeitschrift f{\"u}r angewandte Mathematik und Physik 68~(4)
  (2017) 100.

\bibitem{shuOsher_1988}
C.-W. {Shu}, O.~S., {Efficient Implementation of Essentially Non-oscillatory
  Shock-Capturing Schemes}, Journal of Computational Physics 77 (1988)
  439--471.

\bibitem{braginskii_1965}
S.~I. {Braginskii}, {Transport Processes in a Plasma}, Reviews of Plasma
  Physics 1 (1965) 205.

\bibitem{chapmanCowling_1970}
S.~{Chapman}, T.~{Cowling}, The Mathematical Theory of Non-uniform Gases,
  Cambridge Mathematical Library, Cambridge University Press, 1970.

\bibitem{lifshitzPitaevskii_1981}
E.~M. {Lifshitz}, L.~P. {Pitaevskii}, {Physical Kinetics}, no.~10 in Course of
  Theoretical Physics, Pergamon Press, 1981.

\bibitem{lindquist_1966}
R.~W. {Lindquist}, {Relativistic transport theory}, Annals of Physics 37 (1966)
  487--518.

\bibitem{andersonSpiegel_1972}
J.~L. {Anderson}, E.~A. {Spiegel}, {The Moment Method in Relativistic Radiative
  Transfer}, Astrophysical Journal 171 (1972) 127.

\bibitem{thorne_1981}
K.~S. {Thorne}, {Relativistic radiative transfer - Moment formalisms}, MNRAS
  194 (1981) 439--473.

\bibitem{shibata_etal_2011}
M.~{Shibata}, K.~{Kiuchi}, Y.~{Sekiguchi}, Y.~{Suwa}, {Truncated Moment
  Formalism for Radiation Hydrodynamics in Numerical Relativity}, Progress of
  Theoretical Physics 125 (2011) 1255--1287.

\bibitem{cardall_etal_2013a}
C.~Y. {Cardall}, E.~{Endeve}, A.~{Mezzacappa}, {Conservative 3+1 general
  relativistic variable Eddington tensor radiation transport equations},
  Physical Review D 87 (2013) 103004.

\bibitem{brunnerHolloway_2005}
T.~A. {Brunner}, J.~P. {Holloway}, {Two-dimensional time dependent Riemann
  solvers for neutron transport}, Journal of Computational Physics 210 (2005)
  386--399.

\bibitem{mcclarrenHauck_2010}
R.~G. {McClarren}, C.~D. {Hauck}, {Robust and accurate filtered spherical
  harmonics expansions for radiative transfer}, Journal of Computational
  Physics 229 (2010) 5597--5614.

\bibitem{laboure_etal_2016}
V.~M. {Laboure}, R.~G. {McClarren}, C.~D. {Hauck}, Implicit filtered p n for
  high-energy density thermal radiation transport using discontinuous galerkin
  finite elements, Journal of Computational Physics 321 (2016) 624--643.

\bibitem{mihalasMihalas_1999}
D.~{Mihalas}, B.~W. {Mihalas}, {Foundations of radiation hydrodynamics}, Dover
  (New York), 1999.

\bibitem{kershaw_1976}
D.~{Kershaw}, {Flux limiting nature's own way --- a new method for numerical
  solution of the transport equation}, Tech. Rep. UCRL-78378, Lawrence
  Livermore Laboratory (1976).

\bibitem{minerbo_1978}
G.~N. {Minerbo}, {Maximum entropy Eddington factors.}, JQSRT 20 (1978)
  541--545.

\bibitem{olbrant_etal_2013}
E.~{Olbrant}, C.~D. {Hauck}, M.~{Frank}, {Perturbed, entropy-based closure for
  radiative transfer}, Kinetic and Related Models 6 (2013) 557--587.

\bibitem{leveque_1992}
R.~J. {Leveque}, {Numerical Methods for Conservation Laws}, Lectures in
  Mathematics. ETH Z{\"u}rich, Birkh{\"a}user, 1992.

\bibitem{levermore_1984}
C.~D. {Levermore}, {Relating Eddington factors to flux limiters.}, JQSRT 31
  (1984) 149--160.

\bibitem{levermore_1996}
C.~D. {Levermore}, {Moment closure hierarchies for kinetic theories}, Journal
  of Statistical Physics 83 (1996) 1021--1065.

\bibitem{junk_1998}
M.~{Junk}, Domain of definition of levermore's five-moment system, Journal of
  Statistical Physics 93 (1998) 1143--1167.

\bibitem{hauck_2008}
C.~D. Hauck, C.~D. Levermore, A.~L. Tits, {Convex duality and entropy-based
  moment closures: Characterizing degenerate densities}, SIAM Journal on
  Control and Optimization 47 (2008) 1977--2015.

\bibitem{olbrant_etal_2012}
E.~Olbrant, C.~D. Hauck, M.~Frank, A realizability-preserving discontinuous
  galerkin method for the m1 model of radiative transfer, Journal of
  Computational Physics 231~(17) (2012) 5612--5639.

\bibitem{lareckiBanach_2011}
W.~{Larecki}, Z.~{Banach}, {Entropic Derivation of the Spectral Eddington
  Factors}, JQSRT 112 (2011) 2486--2506.

\bibitem{banachLarecki_2013}
Z.~Banach, W.~Larecki, {Spectral maximum entropy hydrodynamics of fermionic
  radiation: a three-moment system for one-dimensional flows}, Nonlinearity 26
  (2013) 1667--1701.

\bibitem{banachLarecki_2017b}
Z.~Banach, W.~Larecki, {Entropy-based mixed three-moment description of
  fermionic radiation transport in slab and spherical geometries}, Kinetic \&
  Related Models 10~(4) (2017) 879--900.

\bibitem{cockburnShu_2001}
B.~{Cockburn}, C.-W. {Shu}, {Runge-Kutta Discontinuous Galerkin Methods for
  Convection-Dominated Problems}, Journal of Scientific Computing 16 (2001)
  173--261.

\bibitem{hesthavenWarburton_2008}
J.~S. {Hesthaven}, T.~{Warburton}, {Nodal discontinuous {G}alerkin methods:
  {A}lgorithms, analysis and applications}, Springer, 2008.

\bibitem{klockner_etal_2009}
A.~{Kl{\"o}ckner}, T.~{Warburton}, J.~{Bridge}, J.~S. {Hesthaven}, {Nodal
  discontinuous Galerkin methods on graphics processors}, Journal of
  Computational Physics 228 (2009) 7863--7882.

\bibitem{teukolsky_2016}
S.~A. {Teukolsky}, {Formulation of discontinuous Galerkin methods for
  relativistic astrophysics}, Journal of Computational Physics 312 (2016)
  333--356.

\bibitem{larsenMorel_1989}
E.~W. {Larsen}, J.~E. {Morel}, {Asymptotic Solutions of Numerical Transport
  Problems in Optically Thick, Diffusive Regimes II}, Journal of Computational
  Physics 83 (1989) 212--236.

\bibitem{adams_2001}
M.~L. Adams, Discontinuous finite element transport solutions in thick
  diffusive problems, Nuclear science and engineering 137~(3) (2001) 298--333.

\bibitem{guermondKanschat_2010}
J.-L. {Guermond}, G.~{Kanschat}, {Asymptotic Analysis of Upwind Discontinuous
  Galerkin Approximation of the Radiative Transport Equation in the Diffusive
  Limit}, SIAM J. Numer. Anal. 48 (2010) 53--78.

\bibitem{reedHill_1973}
W.~{Reed}, T.~{Hill}, {Triangular mesh methods for the neutron transport
  equation}, Tech. Rep. LA-UR-73-479, Los Alamos National Laboratory (1973).

\bibitem{shu_2016}
C.-W. {Shu}, {High order WENO and DG methods for time-dependent
  convection-dominated PDEs: A brief survey of several recent developments},
  Journal of Computational Physics 316 (2016) 598--613.

\bibitem{xing_etal_2010}
Y.~{Xing}, X.~{Zhang}, C.-W. {Shu}, Positivity-preserving high order
  well-balanced discontinuous galerkin methods for the shallow water equations,
  Advances in Water Resources 33 (2010) 1476--1493.

\bibitem{zhangShu_2011}
X.~{Zhang}, C.-W. {Shu}, Positivity-preserving high order discontinuous
  {G}alerkin schemes for compressible {E}uler equations with source terms,
  Journal of Computational Physics 230 (2011) 1238--1248.

\bibitem{cheng_etal_2013}
Y.~{Cheng}, F.~{Li}, J.~{Qiu}, L.~{Xu}, Positivity-preserving dg and central dg
  methods for ideal mhd equations, Journal of Computational Physics 238 (2013)
  255--280.

\bibitem{zhang_etal_2013}
Y.~{Zhang}, X.~{Zhang}, C.-W. {Shu}, Maximum-principle-satisfying second order
  discontinuous galerkin schemes for convection-diffusion equations on
  triangular meshes, Journal of Computational Physics 234 (2013) 295--316.

\bibitem{endeve_etal_2015}
E.~{Endeve}, C.~D. {Hauck}, Y.~{Xing}, A.~{Mezzacappa}, {Bound-Preserving
  Discontinuous Galerkin Methods for Conservative Phase Space Advection in
  Curvilinear Coordinates}, Journal of Computational Physics 287 (2015)
  151--183.

\bibitem{wuTang_2015}
K.~{Wu}, H.~{Tang}, {High-order accurate physical-constraint-preserving finite
  difference WENO schemes for special relativistic hydrodynamics}, Journal of
  Computational Physics 298 (2015) 539--564.

\bibitem{ascher_etal_1997}
U.~{Ascher}, S.~{Ruuth}, R.~{Spiteri}, {Implicit-explicit Runge-Kutta methods
  for time-dependent partial differential equations}, Applied Numerical
  Mathematics 25 (1997) 151--167.

\bibitem{pareschiRusso_2005}
L.~{Pareschi}, G.~{Russo}, {Implicit-Explicit Runge-Kutta Schemes and
  Application to Hyperbolic Systems with Relaxation}, Journal of Scientific
  Computing 25 (2005) 129--155.

\bibitem{gottlieb_etal_2001}
E.~{Gottlieb}, C.-W. {Shu}, E.~{Tadmor}, {Strong Stability-Preserving
  High-Order Time Discretization Methods}, SIAM Review 43 (2001) 89--112.

\bibitem{chertock_etal_2015}
A.~{Chertock}, S.~{Cui}, A.~{Kurganov}, T.~{Wu}, {Steady State and Sign
  Preserving Semi-Implicit Runge-Kutta Methods for ODEs with Stiff Damping
  Term}, SIAM J. Numer. Anal. 53 (2015) 2008--2029.

\bibitem{hu_etal_2018}
J.~{Hu}, R.~{Shu}, X.~{Zhang}, {Asymptotic-preserving and positivity-preserving
  implicit-explicit schemes for the stiff BGK equation}, SIAM Journal on
  Numerical Analysis 56~(2) (2018) 942--973.

\bibitem{roberts_etal_2016}
L.~F. {Roberts}, C.~D. {Ott}, R.~{Haas}, E.~P. {O'Connor}, P.~{Diener},
  E.~{Schnetter}, {General-Relativistic Three-Dimensional Multi-group Neutrino
  Radiation-Hydrodynamics Simulations of Core-Collapse Supernovae},
  Astrophysical Journal 831 (2016) 98.

\bibitem{foucart_etal_2015}
F.~{Foucart}, E.~{O'Connor}, L.~{Roberts}, M.~D. {Duez}, R.~{Haas}, L.~E.
  {Kidder}, C.~D. {Ott}, H.~P. {Pfeiffer}, M.~A. {Scheel}, B.~{Szilagyi},
  {Post-merger evolution of a neutron star-black hole binary with neutrino
  transport}, Phys. Rev. D 91~(12) (2015) 124021.

\bibitem{janka_etal_1992}
H.-T. {Janka}, R.~{Dgani}, L.~J. {van den Horn}, {Fermion angular distribution
  and maximum entropy Eddington factors}, Astronomy \& Astrophysics 265 (1992)
  345--354.

\bibitem{pons_etal_2000}
J.~A. {Pons}, J.~M. {Ib{\'a}{\~n}ez}, J.~A. {Miralles}, {Hyperbolic character
  of the angular moment equations of radiative transfer and numerical methods},
  MNRAS 317 (2000) 550--562.

\bibitem{smit_etal_2000}
J.~M. {Smit}, L.~J. {van den Horn}, S.~A. {Bludman}, {Closure in flux-limited
  neutrino diffusion and two-moment transport}, Astronomy \& Astrophysics 356
  (2000) 559--569.

\bibitem{just_etal_2015}
O.~{Just}, M.~{Obergaulinger}, H.-T. {Janka}, {A new multidimensional,
  energy-dependent two-moment transport code for neutrino-hydrodynamics}, MNRAS
  453 (2015) 3386--3413.

\bibitem{murchikova_etal_2017}
E.~M. {Murchikova}, E.~{Abdikamalov}, T.~{Urbatsch}, {Analytic closures for M1
  neutrino transport}, MNRAS 469 (2017) 1725--1737.

\bibitem{shohatTamarkin_1943}
J.~Shohat, J.~Tamarkin, The Problem of Moments, Mathematical Surveys and
  Monographs, American Mathematical Society, 1943.

\bibitem{dimarcoPareschi2013}
G.~Dimarco, L.~Pareschi, Asymptotic preserving implicit-explicit runge--kutta
  methods for nonlinear kinetic equations, SIAM Journal on Numerical Analysis
  51~(2) (2013) 1064--1087.

\bibitem{jinLevermore_1996}
S.~{Jin}, C.~{Levermore}, {Numerical Schemes for Hyperbolic Conservation Laws
  with Stiff Relaxation Terms}, Journal of Computational Physics 126 (1996)
  942--973.

\bibitem{mcclarren_etal_2008}
R.~{McClarren}, T.~{Evans}, R.~{Lowrie}, J.~{Densmore}, {Semi-implicit time
  integration for PN thermal radiative transfer}, Journal of Computational
  Physics 227~(16) (2008) 7561--7586.

\bibitem{radice_etal_2013}
D.~{Radice}, E.~{Abdikamalov}, L.~{Rezzolla}, C.~D. {Ott}, {A new spherical
  harmonics scheme for multi-dimensional radiation transport I. Static matter
  configurations}, Journal of Computational Physics 242 (2013) 648--669.

\bibitem{liuOsher_1996}
X.~D. {Liu}, S.~{Osher}, Nonoscillatory high order accurate self-similar
  maximum principle satisfying shock capturing schemes {I}, SIAM J. Numer.
  Anal. 33~(2) (1996) 760--779.

\bibitem{cavaglieriBewley2015}
D.~{Cavaglieri}, T.~{Bewley}, {Low-storage implicit/explicit Runge–Kutta
  schemes for the simulation of stiff high-dimensional ODE systems}, Journal of
  Computational Physics 286 (2015) 172 -- 193.

\bibitem{skinnerOstriker_2013}
M.~A. {Skinner}, E.~C. {Ostriker}, {A Two-moment Radiation Hydrodynamics Module
  in Athena Using a Time-explicit Godunov Method}, Astrophysical Journal
  Supplement Series 206 (2013) 21.

\bibitem{brunner_2002}
T.~A. {Brunner}, {Forms of approximate radiation transport}, Technical Report,
  Sandia National Laboratories SAND2002-1778 (2002) 1--43.

\bibitem{garrettHauck_2013}
K.~C. {Garrett}, C.~D. {Hauck}, {A Comparison of Moment Closures for Linear
  Kinetic Transport Equations: The Line Source Benchmark}, Transport Theory and
  Statistical Physics 42 (2013) 203--235.

\bibitem{smit_etal_1997}
J.~M. {Smit}, J.~{Cernohorsky}, C.~P. {Dullemond}, {Hyperbolicity and critical
  points in two-moment approximate radiative transfer.}, Astronomy \&
  Astrophysics 325 (1997) 203--211.

\end{thebibliography}
\end{document}